\documentclass[11pt]{article}

\usepackage{enumitem}

\usepackage[square,numbers]{natbib}

\usepackage{subfloat}
\usepackage{subfig}
\usepackage[utf8]{inputenc}

\usepackage{setspace}
\usepackage{algpseudocode}
\usepackage[format=plain,
            labelfont={it,it},
            textfont=it,
            font=small]{caption}

\usepackage[colorlinks,citecolor=blue,urlcolor=blue,linkcolor=blue]{hyperref} 

\usepackage{color}              
\usepackage{color-edits}
\addauthor{bc}{cyan}
\addauthor{bs}{purple}

\usepackage{mathrsfs}  
\usepackage{dsfont} 
\usepackage{ulem}
\usepackage{tikz}
\usetikzlibrary{graphs,graphs.standard}
\usetikzlibrary{arrows}
\usetikzlibrary{arrows.meta}
\usetikzlibrary{graphs,graphs.standard,quotes}
\usetikzlibrary{shapes.geometric}
   \usetikzlibrary{math}
   \usetikzlibrary{shapes.misc}
\usetikzlibrary{shapes.multipart,positioning,patterns,backgrounds}

\usepackage[hang,flushmargin]{footmisc}

\usepackage{amsmath}
   \usepackage{amssymb}
\usepackage{amsthm}
 \usepackage{thmtools}

\newcommand\blfootnote[1]{%
  \begingroup
  \renewcommand\thefootnote{}\footnote{#1}%
  \addtocounter{footnote}{-1}%
  \endgroup
}

\newtheorem{theorem}{Theorem}
\newtheorem{lemma}{Lemma}
\newtheorem{proposition}{Proposition}
\newtheorem{corollary}{Corollary}

\theoremstyle{definition}%
\newtheorem{example}{Example}
\newtheorem{remark}{Remark}%
\newtheorem{definition}{Definition}%
\newtheorem{assumption}{Assumption}%

   \usepackage[most]{tcolorbox}
\usetikzlibrary{patterns}
\pgfdeclarepatternformonly{mystrikeout}{
\pgfqpoint{-1pt}{-1pt}}{\pgfqpoint{11pt}{11pt}}{\pgfqpoint{10pt}{10pt}}%
{
  \pgfsetlinewidth{0.4pt}
  \pgfpathmoveto{\pgfqpoint{0pt}{0pt}}
  \pgfpathlineto{\pgfqpoint{10.1pt}{10.1pt}}
  \pgfusepath{stroke}
}
\newtcolorbox{tcbstrikeout}{breakable,
 enhanced jigsaw,
 opacityback=0,
 parbox=false,
 boxrule=0mm,
 top=0mm,bottom=0pt,left=0pt,right=0pt,
 boxsep=0pt,
 frame hidden,
 finish={\fill[pattern=mystrikeout,pattern color=red] (frame.north west) rectangle (frame.south east);}
}

\usepackage[letterpaper, margin=1in]{geometry}

\usepackage{afterpage}

\usepackage{color}              
\usepackage{color-edits}
\addauthor{bc}{cyan}
\addauthor{bs}{purple}

 \usepackage{thmtools}

\usepackage{mathrsfs}  
\usepackage{mathtools}
\usepackage{dsfont} 
\usepackage{ulem}
\usepackage{tikz}
\usetikzlibrary{graphs,graphs.standard}
\usetikzlibrary{arrows}
\usetikzlibrary{arrows.meta}
\usetikzlibrary{graphs,graphs.standard,quotes}
\usetikzlibrary{shapes.geometric}
\usepackage{fixltx2e}    
   \usetikzlibrary{math}
   \usetikzlibrary{shapes.misc}
   \usepackage{amssymb}
\usetikzlibrary{shapes.multipart,positioning,patterns,backgrounds}

   \usepackage[most]{tcolorbox}
\usetikzlibrary{patterns}

\let\N\Natural

\let\R\Real

\title{Improving the Security of United States Elections with
Robust Optimization}

\author
{Braden L. Crimmins$^{1}$, J. Alex Halderman$^{1}$, Bradley Sturt$^{2}$\\
\\
\normalsize{$^{1}$Computer Science and Engineering, University of Michigan, Ann Arbor}\\
\normalsize{$^{2}$Information and Decision Sciences, University of Illinois Chicago}
}

\date{}

\begin{document} 



\maketitle

\begin{abstract}
For more than a century, election officials across the United States have inspected voting machines before elections using a procedure called Logic and Accuracy Testing (LAT). This procedure consists of election officials casting a test deck of ballots into each voting machine and confirming the machine produces the expected vote total for each candidate. We bring a scientific perspective to LAT by introducing the first formal approach to designing test decks with rigorous security guarantees. Specifically, our approach employs robust optimization to find test decks that are guaranteed to detect any voting machine misconfiguration that would cause votes to be swapped across candidates. Out of all the test decks with this security guarantee, our robust optimization problem yields the test deck with the minimum number of ballots, thereby minimizing implementation costs for election officials. To facilitate deployment at scale, we develop a practically efficient exact algorithm for solving our robust optimization problems based on the cutting plane method. In partnership with the Michigan Bureau of Elections, we retrospectively applied our approach to all 6928 ballot styles from Michigan’s November 2022 general election; this retrospective study reveals that the test decks with rigorous security guarantees obtained by our approach require, on average, only 1.2\% more ballots than current practice. Our approach has since been piloted in real-world elections by the Michigan Bureau of Elections as a low-cost way to improve election security and increase public trust in democratic institutions.\looseness=-1  
\end{abstract}
\blfootnote{First version: August 4, 2023. 
Revisions submitted on May 21, 2024 and September 3, 2024.
Accepted for publication on September 27, 2024.}

\clearpage

\section{Introduction} \label{sec:intro} 

Computerized voting machines are widely used to scan  ballots and determine election outcomes throughout the United States and around the world. Voting machines are used instead of hand counting because voters are often invited to participate in a large number of contests in an election---including political offices from the President to local school boards---which causes hand counting to be impractically costly and time consuming. In this paper, we develop a low-cost approach to reducing the security risks of voting machines and improving public trust in democratic institutions by drawing on  techniques from the field of robust optimization. 

\subsection{Background} \label{sec:intro:background}
For voting machines that scan ballots to count votes accurately, they must be configured
with correct mappings between the voting targets on the ballot---i.e., the boxes or ovals that voters mark---and the candidates who should receive the votes.  If a voting machine is  configured with an incorrect mapping, then the machine may count votes for the wrong candidates.  As illustrated by the following examples, voting machines  misconfigurations  can produce {dramatically} wrong vote totals and damage public trust in elections:
\begin{itemize}
    \item 
During the 2020 election, voting machines in Antrim County, Michigan were accidentally misconfigured with mappings  that caused votes for Republicans to be tallied for Democrats and votes for Democrats to go uncounted~\cite{halderman2021antrim}.  The erroneous vote totals announced as a result of this flaw received widespread media coverage \cite{HowaMich33,Official52}, and this incident  served as the basis for a draft executive order, later obtained by the Congressional committee investigating the events of January 6, 2021, that would have directed the Secretary of Defense to seize voting machines~\cite{SWAN_2022}.  
\item Similar accidental misconfigurations affected announced election results in Pennsylvania~\cite{corasaniti2019} and Georgia~\cite{dekalb2022} in the past five years.  Although the errors were quickly caught and corrected, they similarly resulted in the initial publication of incorrect vote totals and generated significant negative publicity for the affected jurisdictions. In a particularly recent example, an accidental misconfiguration in Northampton County, Pennsylvania during their November 2023 election caused votes to be swapped across two judge contests, leading to voter confusion and long lines on election day~\cite{politiconorthamp,missingnorthampton}.  
\item Misconfigurations could also be induced deliberately by adversaries with very little technical expertise. Indeed, a group that contends the outcome of the 2020 presidential election was fraudulent recently released a video that demonstrates exactly how one could strategically induce these misconfigurations to manipulate future election outcomes~\cite{deperno}. Such deliberate manipulations would allow an adversary to sow doubt in election systems, influence who wins prominent political offices, and---in at least 24 states~\cite{ncsl-initiative}---directly affect the passage of laws on matters ranging from environmental policy to abortion rights.
\end{itemize}

Past work has sought to address the potential dangers of compromised voting machines through post-election interventions such as risk-limiting audits and cryptographic systems that make announced results publicly verifiable~\cite{lindeman2012gentle,bernhard2017public}. These post-election procedures are increasingly being implemented in the United States and have received  attention in popular media \cite{youtube}. 
However, no prior work has developed a procedure that is guaranteed to detect important classes of possible attacks \emph{before} an election takes place. Such a procedure could help safeguard election integrity and public confidence by detecting attacks before they affect reported vote totals.  Our paper  develops a rigorous and low-cost pre-election defense against misconfiguration-based cyberattacks by applying robust optimization to a  widely-used testing procedure called Logic and Accuracy Testing (LAT).\looseness=-1

\begin{figure}[t]
\centering
\includegraphics[width=0.9\linewidth]{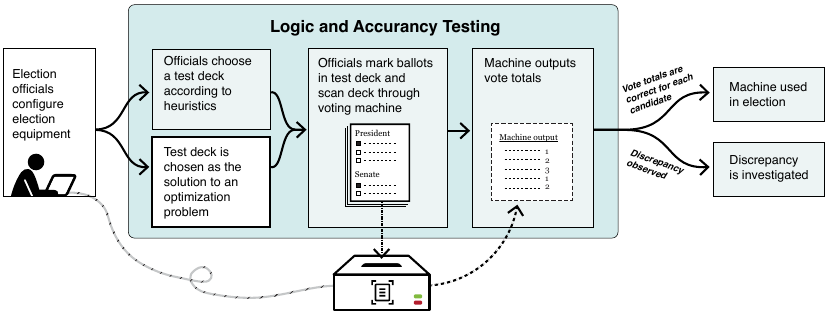} 
\caption{ Visualization of Logic and Accuracy Testing (LAT).  The procedure is conducted chronologically from left to right on voting machines before each election. The modification to LAT proposed in this paper is denoted by the white box with the text ``Test deck is chosen as the solution to an optimization problem".} \label{fig:lat}
\end{figure}

\subsection{Logic and Accuracy Testing}

For more than a century,  election officials throughout the United States have used LAT to  inspect voting machines prior to elections. 
  The idea behind LAT is simple: officials prepare a set of ballots
with known votes—dubbed a test deck—then cast the ballots through each voting
machine and confirm that the machine outputs the expected tallies for all candidates (see Figure~\ref{fig:lat}).
Any discrepancy indicates a potential malfunction, which can be addressed
before the machine is used to count real votes. LAT was initially developed in the early 1900s
to protect against breakdowns of mechanical lever-based voting machines and is today
required by law before each election in all fifty states~\cite{Jones_Simons_2012,walker2022lat}. 

\begin{figure}[t]
\centering
\includegraphics[width=0.6\linewidth]{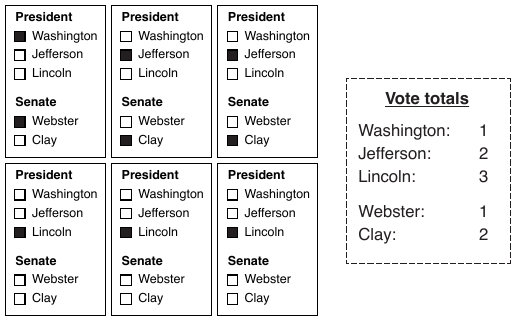}

\caption{ A test deck composed of six  ballots for a simple election with two contests. The first contest  is a presidential contest with three candidates; the second contest is a senatorial contest with two candidates. In each contest, a voter is allowed to vote for at most one candidate.  } \label{fig:testdeck}

\end{figure}

Despite the widespread use of LAT,  no prior work has used LAT for detecting attacks on {modern} computerized voting machines. In fact, LAT is not an obvious candidate for securing  modern elections;  it cannot, for example, detect  malicious alterations to a voting machine's software that cause the voting machine to operate fraudulently only after testing has concluded.\footnote{Such manipulations are sometimes called ``Volkswagen attacks'', in reference to the 2015 Volkswagen emissions scandal wherein vehicle motors were programmed to reduce their emission  levels only when the vehicles were undergoing testing for compliance with environmental efficiency regulations~\cite{burki2015diesel}.}  
 Nonetheless,  LAT has a number of properties that make it potentially attractive for election security. 
First, the legally mandated use of LAT across the United States means that repurposing this procedure as a modern security tool would require little investment on already-overburdened election administrators.  Second, the fact that LAT is performed prior to elections means that it is well-situated to detect cyberattacks before they affect the public.  Third, developing sophisticated cyberattacks that cannot be detected by LAT requires technical capabilities that are out of reach for many would-be adversaries. In particular,  we show in this paper that  LAT has
the potential to be an effective defense against less sophisticated (yet still practically significant) classes of attacks that are based on deliberate misconfiguration of voting machines.

The set of misconfigurations which would be detected by LAT  hinges on the design of the test deck, 
i.e., the decision of which voting targets to fill out on each ballot. Until
now, test decks throughout the United States have been designed following
simple heuristics that are based on human intuition~\cite{walker2022lat}. For example, Figure~\ref{fig:testdeck} shows an example of a test deck 
constructed by a common heuristic that gives each candidate within each contest a different
number of votes. However, Figure~\ref{fig:latesting_simpleexample} shows the output of voting machines using the test
deck from Figure~\ref{fig:testdeck} under three examples of misconfigured mappings between voting targets
and candidates, including one misconfiguration which  this test deck would not
detect. This demonstrates that this simple heuristic is not guaranteed to secure
the voting machine from misconfiguration attacks. 
 If each candidate on the ballot received a different number of votes, then all misconfigured mappings between voting targets and candidates would be detected, but this strategy for designing test decks is not used in practice because it requires impractically many ballots for real-world elections (see Appendix~\ref{appx:upperbound}). The difficulty of marking and scanning test decks scales with the number of ballots included, so short test decks are imperative for practical implementation.

 \begin{figure}[bp]
\centering 
\subfloat[]{%
\includegraphics[width=0.7\linewidth]{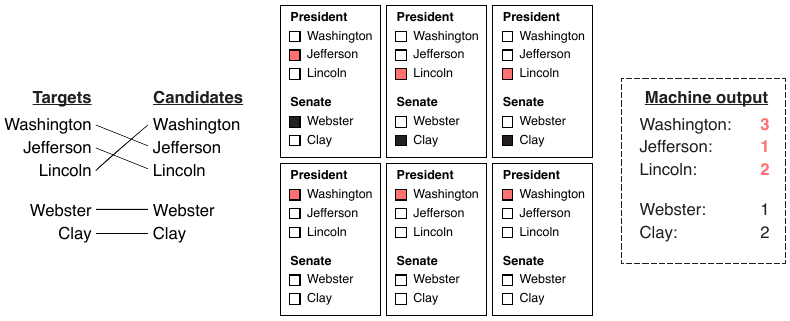}
\label{fig:latesting_simpleexample:swap_jefferson_lincoln}
}

\subfloat[]{%
\centering
\includegraphics[width=0.7\linewidth]{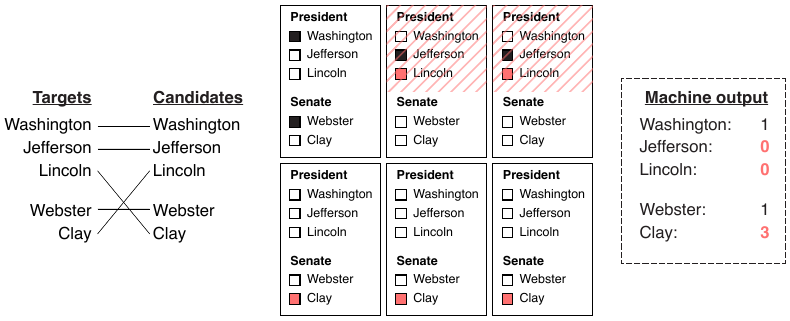}
\label{fig:latesting_simpleexample:swap_lincoln_clay}
}


\subfloat[\label{fig:latesting_simpleexample:swap_jefferson_clay}]{%
\includegraphics[width=0.7\linewidth]{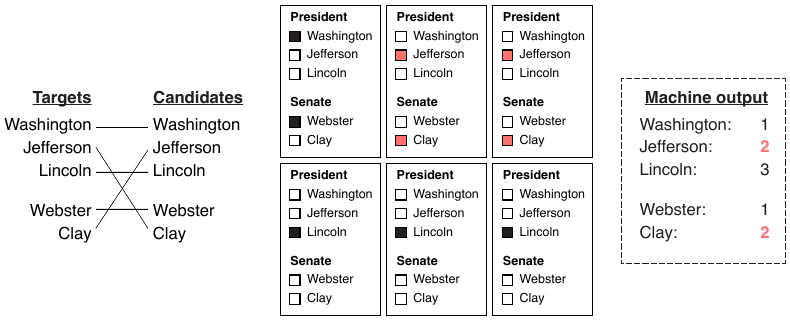}
}

\caption{ Each example shows a misconfiguration of the mapping between voting targets and candidates (left), the misconfigured voting machine's interpretation of the test deck from Figure~\ref{fig:testdeck} (center), and the vote tally that is output by the misconfigured voting machine (right).   The color red indicates the aspects of the interpretation of the test deck and the machine output that are impacted by the misconfiguration of the voting machine. Diagonal lines through a contest indicate that the filled-out ballot is interpreted as containing an overvote in that contest, in which case the voting machine interprets the filled-out ballot as if no candidates were selected in that contest. \textnormal{\textbf{(a)}}
The misconfiguration is detected because the output of the voting machine  includes incorrect vote totals for Washington, Jefferson, and Lincoln. \textnormal{\textbf{(b)}}  The misconfiguration is detected because the output of the voting machine  includes incorrect vote totals for Jefferson, Lincoln, and Clay. \textnormal{\textbf{(c)}}  The misconfiguration is not detected because the output of the voting machine includes correct vote totals for all candidates (see Figure~\ref{fig:testdeck}).\looseness=-1
} \label{fig:latesting_simpleexample}
\end{figure}

\subsection{Contributions}

We bring a scientific perspective to LAT by introducing the first formal approach to designing test decks for LAT with rigorous security guarantees. 
Specifically, our approach employs mathematical optimization---rather than heuristics---to find test decks that are guaranteed to detect any misconfiguration that swaps votes between candidates. 
Moreover, out of all the test decks that are guaranteed to detect these swaps, our approach yields a test deck with the minimum number of ballots, thereby minimizing implementation difficulties for election officials.

In greater detail, our approach to designing test decks consists of constructing and solving a robust optimization problem \cite{ben2009robust,bertsimas2011theory}. The input to the robust optimization problem is a ``ballot style,'' which for our purposes means the set of contests that appear on a given ballot, the set of candidates who are running in each of those contests, the maximum number of candidates that a voter is allowed to select in each contest, and the correct mapping of voting targets to candidates.\footnote{Because the contests available on a ballot depend on granular political subdivisions like county, municipality, legislative district, and school district, there are often many thousands of ballot styles used across the different jurisdictions of a state in any given election. Voting machines across the state are configured separately for each of the different ballot styles, and a different test deck must be used to evaluate each such configuration.} The output of the robust  optimization problem is the design of a minimum-length test deck that is guaranteed to detect whether a voting machine has an incorrect bijective mapping from candidates to voting targets for the ballot style in question. The robust optimization problem is stated formally in \S\ref{sec:math:swap}. \looseness=-1

One of the key difficulties in solving our  robust optimization problem  lies in the large number of incorrect bijective mappings from candidates to voting targets. In a ballot style with $N$ candidates and voting targets, there are $ N! - 1$ possible ways that voting targets can be swapped, one swap for each bijection over candidates {less the single correct bijection}. In United States elections, the number of candidates across the contests of a ballot style is often greater than one hundred. Consequently, formulating our robust optimization problem often requires more than $100! -1 \approx 10^{157}$ constraints, a number far greater  than the estimated number of particles in the observable universe~\cite{whittaker1945eddington}. An optimization problem that explicitly encodes all possible swaps thus cannot be represented nor solved on any extant computer for many real-world elections. It is currently unknown whether  the robust optimization problem  is NP-hard, and it is unknown whether there exists a  mixed-integer linear programming reformulation of the robust optimization problem of size that is polynomial in the number of candidates $N$.\looseness=-1

We contend with the above computational challenge by developing an exact algorithm for the robust optimization problem inspired by the cutting plane method (see \S\ref{appx:cutting}). The cutting plane method is a classical technique for solving optimization problems with many constraints by solving a sequence of optimization problems with small numbers of constraints. In our setting, each iteration of the cutting plane method solves a relaxation of the robust optimization problem that contains a small subset of the $N! - 1$ swaps. If the optimal test deck of the relaxed optimization problem detects all of the $N! -1$ swaps of the original problem, then the algorithm terminates. Otherwise, the algorithm finds a swap that is undetected by the optimal test deck of the relaxed problem, adds the undetected swap into the relaxed problem, and then solves the relaxed problem again. This process repeats until a feasible solution for the original robust optimization problem is obtained.

To make the cutting plane method terminate in practical computation times in real-world elections, we make a number of novel  algorithmic developments. First, we reformulate the relaxed optimization problem as well as the problem of finding an undetected swap as mixed-integer linear optimization problems (see \S\ref{appx:mip}). These reformulations enable the cutting plane method to be easily implemented using widely available open-source and commercial optimization software such as Gurobi and Mosek.  Second, we offer a variety of theoretically-justified improvements to our mixed-integer linear optimization formulations (see \S\ref{appx:improvements}) that aim to decrease the number of iterations and decrease the per-iteration computation time of the cutting plane method. 
These improvements include dynamically identifying and removing unnecessary decision variables from the mixed-integer linear optimization problems (\S\ref{appx:improvements:reducingdecisions}), adding  constraints that impose the structure of optimal test decks into the mixed-integer linear optimization problems (\S\ref{appx:improvements:symmetry} and \S\ref{appx:improvements:n_choose_k}), developing a combinatorial framework for identifying which swap to add to the relaxed optimization problem in each iteration (\S\ref{appx:improvements:swap}), and combining all contests that are not competitive (\S\ref{appx:improvements:noncompetitive}). 
In Appendix~\ref{appx:additional_experiments},  we demonstrate  via experiments on synthetic elections that each of our improvements yields significant decreases in the computation time and number of iterations of the cutting plane method.\looseness=-1

We conclude by  showcasing the value of our robust optimization approach in application to real world elections.  
In partnership with the Michigan Bureau of Elections, we applied our approach to each of the state's 6928 ballot styles from the  November 2022 general election.  Our results for this election (see  \S\ref{sec:experiments}) reveal that our approach only required a 1.2\% average increase in the number of test ballots  compared to current practice across the state's 6928 ballot styles. Hence, our approach can be deployed with minimal financial cost or operational overhead while providing significant  security benefits to election jurisdictions. Moreover, our cutting plane method for solving the robust optimization problems enabled our approach to obtain optimal test decks for all 6928 ballot styles in less than seven hours. These findings demonstrate that our cutting plane method can find optimal test decks for all of the ballot styles across a state in computation times that are practical from the perspective of election officials.  Our approach described in this paper has  been piloted by the Michigan Bureau of Elections in real-world elections during the summer of 2023, and  
we hope that  our approach will be adopted by more states and countries in upcoming elections as a low-cost tool to improving the security
and increasing public confidence in election outcomes. 

An open source portion of the code from this paper is available at \url{https://github.com/ballotiq/deck-checker}.

\section{Vulnerabilities of Existing Heuristics for Designing Test Decks} \label{sec:vulnerabilities}
Our proposed approach to designing test decks with rigorous security guarantees is presented in \S\ref{sec:math}. To motivate our approach, we begin in this section by describing three examples of misconfiguration attacks against United States voting machines. We show in each of the three examples how the attack could be strategically deployed by an adversary to undermine public trust or change the outcome of an election. Finally, we show how the examples of attacks could evade detection by LAT when test decks are designed by commonly used heuristics. 

\paragraph{Swaps of Individual Candidates.} Suppose that the goal of an adversary is to decrease the number of votes received by a specific candidate in a high-stakes contest near the top of the ballot (such as a presidential contest). In this case, an example of a misconfiguration that would be appealing to the adversary is one that swaps the voting target of the specific candidate with the target of a candidate from a contest that is lower on the ballot (such as the contest to elect a sanitation commissioner). Because fewer people vote in downballot contests~\cite{kimball2008voting}, this misconfiguration could result in the adversary's disfavored presidential candidate receiving fewer votes than they should. Moreover, if the test deck for LAT is designed using a common heuristic in which a single ballot contains votes for the first candidate in each contest, two ballots contain votes for the second candidate in each contest, and so on, then LAT would not detect any misconfiguration that swaps the targets for two candidates at corresponding indices in their respective contests. An example of a test deck constructed by this common heuristic is shown in Figure~\ref{fig:testdeck}, and the misconfiguration depicted in Figure~\ref{fig:latesting_simpleexample:swap_jefferson_clay} is an example of such a swap that goes undetected, since it swaps the second candidate in the presidential contest with the second candidate in the senatorial contest. 

\paragraph{Swaps of Entire Contests.} In many states, elections put certain yes-or-no questions---commonly called initiatives, proposals, or referendums---directly to voters. The effect of these contests range from modifying a state's constitution on matters such as abortion rights~\cite{abortion} and environmental policy~\cite{nyt,scienceballot2016} to recalling sitting politicians from their office~\cite{wsj}. If an adversary wished to swap the outcome of two such contests, they could  misconfigure the voting machine to swap the voting targets for `yes' and for `no' between the two contests. Moreover, if LAT is conducted with a test deck that includes the same number of votes for `yes' and the same number of votes for `no' in each of the two contests---which is the case under every common heuristic for test deck preparation used today~\cite{walker2022lat}---then this misconfiguration would not be detected by LAT (see \S\ref{appx:improvements:n_choose_k}). This attack could thus be used to ensure a favored proposal passes or a disfavored proposal fails, and would allow an adversary to directly influence the laws or constitution of a jurisdiction.

\paragraph{Deliberately Flawed Test Decks.} It is common for jurisdictions to contract outside vendors to configure their voting machines as well as design the test decks used to conduct LAT. If this vendor is untrustworthy, they could misconfigure the machine according to their own preference, then deliberately construct a test deck which would fail to detect the modification. Indeed, we show in Appendix~\ref{appx:malicious-decks}  that a vendor has significant freedom in the misconfiguration they choose, even when the test deck they produce is constrained by some of the most stringent  legal requirements in use by states today.

\vspace{.5cm}
\noindent In the following section, we introduce an approach to designing test decks that enables LAT to become a rigorous pre-election defense against an important class of misconfiguration attacks. This class includes, among many others, the three examples of attacks described above.

\section{Robust Logic and Accuracy Testing}\label{sec:math}

In this section, we  introduce  Robust Logic and Accuracy Testing (RLAT), an optimization-based framework for designing test decks in LAT with rigorous security guarantees. 
This section has the following organization. 
 \S\ref{sec:math:notation} develops the terminology and mathematical notation that will be used throughout the paper. 
 \S\ref{sec:math:optimization} presents a general formulation of RLAT and discusses its value from the perspective of various stakeholders in United States elections.
 \S\ref{sec:math:swap} uses  the RLAT framework to derive our robust optimization problem~\eqref{prob:robust} for finding  a minimum-length test deck that will detect whether a voting machine is misconfigured to swap votes between candidates.
 \S\ref{sec:math:discussion} establishes the fundamental structural properties  of test decks that are feasible for the robust optimization problem~\eqref{prob:robust}.

\subsection{Preliminaries}  \label{sec:math:notation}
 A ballot style is composed of a set of contests $\mathcal{C} \triangleq \{1,\ldots,C\}$ and a set of candidates $\mathcal{N} \triangleq  \{1,\ldots,N\}$. For each contest $c \in \mathcal{C}$, we let $\mathcal{N}_c \subseteq \mathcal{N}$ denote the subset of candidates that appear in contest $c$, and we let $v_c$ denote the maximum number of candidates in that contest that may be legally selected by a voter. For example, for a contest that corresponds to the senatorial election, the set $\mathcal{N}_c$ would contain the indices of the  candidates that are running for Senator, and the equality $v_c = 1$ would denote that each voter is permitted to select at most one candidate in the contest. 
In a contest for a local school board with five vacancies, we would alternatively have the equality  $v_c=5$.  We assume that each candidate $i \in \mathcal{N}$ appears in exactly one contest.  We say that a contest $c$ is noncompetitive if the maximum number of votes $v_c$ is equal to the number of candidates $| \mathcal{N}_c|$ in the contest.\footnote{Noncompetitive contests often arise when an incumbent to some local office runs unopposed for re-election. This is especially common in states which elect judges, since there is a strong normative prohibition against challenging a sitting judge's re-election bid~\cite{olson2022incumbency}. 
}  In real-world elections such as those from Michigan (see Figure~\ref{fig:michigan_data}),   the number of contests in each ballot style typically satisfies $15 \le C \le 40$, and the number of candidates in each ballot style typically satisfies $60 \le N \le 120$. 

 \begin{figure}[t]
    \centering
\includegraphics[width=0.9\linewidth]{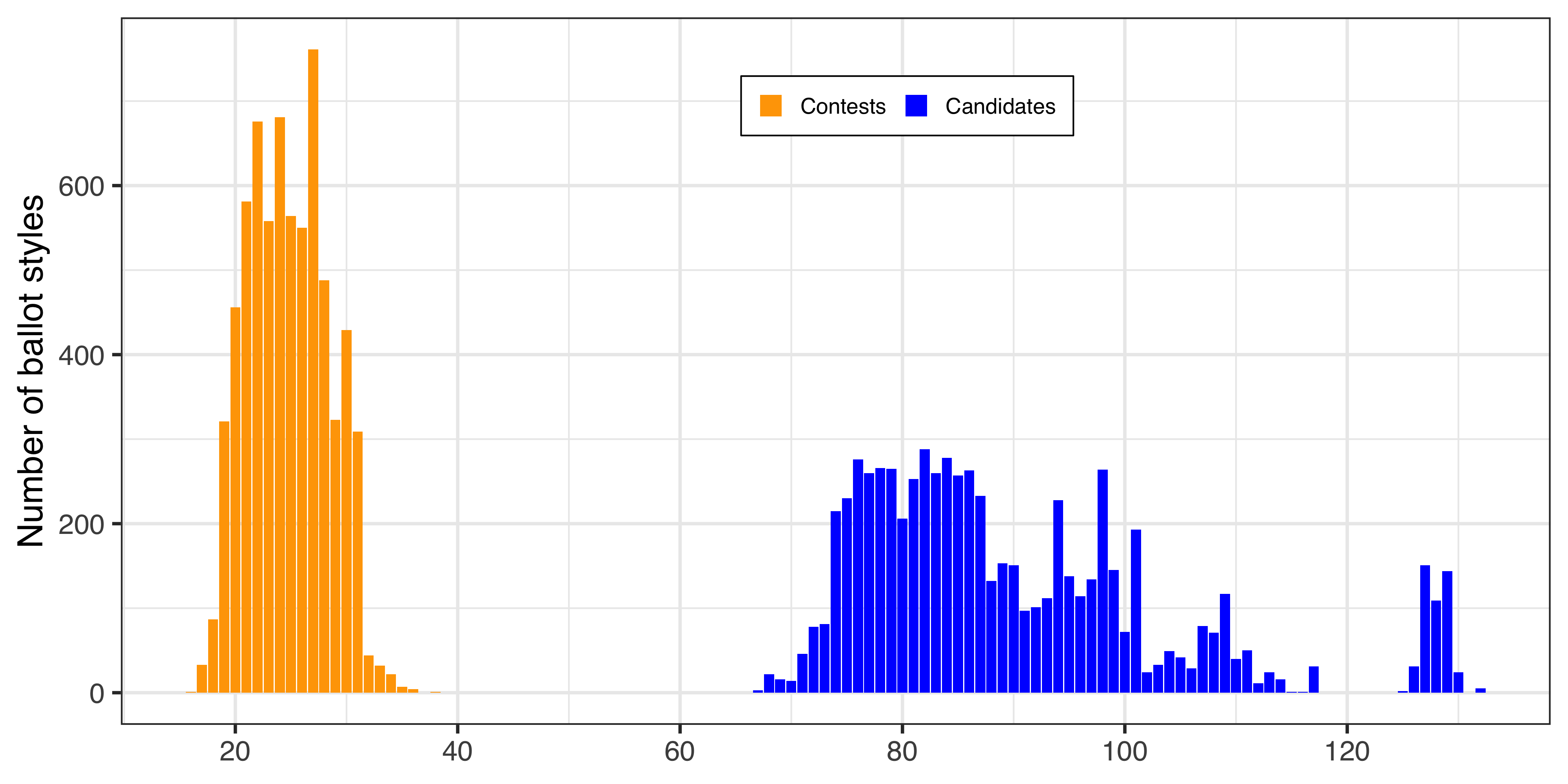}  

\caption{Histogram of the total number of contests (orange) and total number of candidates (blue) that appeared across  the 6928 ballot styles in Michigan's November 2022 general election.}\label{fig:michigan_data}

\end{figure}

A ballot refers to a physical document  that contains a box or oval beside each candidate, termed targets, that are used by voters to record their choices. With a slight abuse of notation, we denote the targets on a ballot by  $\mathcal{N}  \triangleq \{1,\ldots,N\}$, where each target $i \in \mathcal{N}$ refers to the box or oval that is beside candidate $i \in \mathcal{N}$. A {filled-out ballot}  is represented by a subset of targets $\beta \subseteq \mathcal{N}$, with the interpretation that the  filled-out ballot satisfies $i \in \beta$ if and only if the filled-out ballot selected  target $i$. It follows  that  the number of targets beside candidates in contest $c \in \mathcal{C}$ that are selected by a filled-out ballot $\beta \subseteq \mathcal{N}$ is equal to $|\mathcal{N}_c \cap \beta|$. A deck refers to any finite-length sequence of filled-out ballots $(\beta_1,\ldots,\beta_B)$.\looseness=-1

 When a voting machine operates correctly, the machine will receive a deck of filled-out ballots as its input, and the machine will output the total number of targets that are selected for each candidate in the filled-out ballots that do not have an overvote in that candidate's contest. For any input deck $(\beta_1,\ldots,\beta_B)$, we denote the output of a voting machine that operates correctly  by the vector-valued function $$T^*(\beta_1,\ldots,\beta_B) \equiv (T^*_1(\beta_1,\ldots,\beta_B), \ldots, T^*_N(\beta_1,\ldots,\beta_B)),$$ 
 with the output 
 for each candidate $i \in \mathcal{N}_c$ in each contest $c \in \mathcal{C}$ defined as 
 \begin{align*}T_i^*(\beta_1,\ldots,\beta_B) & \triangleq \sum_{b=1}^B \mathbb{I} \left \{ i \in \beta_b \textnormal{ and }\left| \mathcal{N}_c \cap \beta_b \right| \le v_c \right \}.
 \end{align*}
 In the above definition, and throughout the rest of this paper, we let   $\mathbb{I} \left \{ \cdot \right \}$ represent the indicator function that is equal to one if $\cdot$ is true and is equal to zero if $\cdot$ is false. The inclusion $i \in \beta_b$ holds if and only if filled-out ballot $\beta_b$ has selected the target that is beside candidate $i$, and   the inequality 
$\left| \mathcal{N}_c \cap \beta_b \right| \le v_c$ holds if and only if filled-out ballot $\beta_b$ has selected at most $v_c$ of the targets that are beside the candidates in contest $c$. In other words, the inequality $\left| \mathcal{N}_c \cap \beta_b \right| \le v_c$ holds if and only if filled-out ballot $\beta_b$ is interpreted  by the voting machine that operates correctly as not containing an overvote in contest $c$.  
For notational convenience, we denote the set of ballots that do not overvote {any} contest by
$\mathscr{B} \triangleq \left \{ \beta \subseteq \mathcal{N}: \left| \mathcal{N}_c \cap \beta \right| \le v_c \; \forall c \in \mathcal{C} \right \}.$

\begin{remark} \label{remark:overvotes}
If the filled-out ballots in an input deck do not contain  overvotes,  then the output of the voting machine that operates correctly  will equal the number of filled-out ballots that select the target   associated with  each candidate. In other words, if $\beta_1,\ldots,\beta_B \in \mathscr{B}$, then the equality $T^*_i(\beta_1,\ldots,\beta_B) = | \left \{ b \in \{1,\ldots,B \}: i \in \beta_b \right \}|$ holds for each candidate $i \in \mathcal{N}$. 
\end{remark}

To represent the output of a specific voting machine that may or may not be operating correctly, we use the vector-valued function $$\widehat{T}(\beta_1,\ldots,\beta_B)  \equiv (\widehat{T}_1(\beta_1,\ldots,\beta_B), \ldots, \widehat{T}_N(\beta_1,\ldots,\beta_B)).$$ This function represents the output of the voting machine for any input deck of filled-out ballots $(\beta_1,\ldots,\beta_B)$. We say that the voting machine represented by the vector-valued function $\widehat{T}(\cdot)$ is not operating correctly if there exists a deck $(\beta_1,\ldots,\beta_B)$ and a candidate $i \in \mathcal{N}$ such that $\widehat{T}_i(\beta_1,\ldots,\beta_B) \neq  T^*_i(\beta_1,\ldots,\beta_B)$. If the voting machine is not operating correctly, and if the  voting machine is used to count votes in an actual election, then it could produce results inconsistent with the actual ballots cast and change the outcome of the election.

\subsection{Formulation of RLAT} \label{sec:math:optimization}
We now introduce the mathematical description of Robust Logic and Accuracy Testing (RLAT), an optimization-based framework for designing test decks in LAT with rigorous security guarantees.  Specifically, given an uncertainty set $\mathcal{U}$ of possible ways that a voting machine might be operating incorrectly, RLAT designs the test deck by solving the following optimization problem:
\begin{align}\tag{RO} \label{prob:robust_general}
\begin{aligned}
\underset{B \in \N, \; \beta_1,\ldots,\beta_B \in \mathscr{B}}{\textnormal{minimize}}\quad & B\\
\textnormal{subject to} \quad &  \widehat{T}(\beta_1,\ldots,\beta_B) \neq T^*(\beta_1,\ldots,\beta_B) \quad \forall \widehat{T}(\cdot) \in \mathcal{U}. 
\end{aligned}
\end{align}
The optimization problem~\eqref{prob:robust_general}  yields a minimum-length test deck that is {guaranteed} to detect whether a voting machine is operating incorrectly in any of the ways specified by the uncertainty set.\looseness=-1

In greater detail, the decision variables of the optimization problem~\eqref{prob:robust_general} consist of the length of the test deck, $B \in \N$, as well as the test deck of filled-out ballots without any overvotes, $\beta_1,\ldots,\beta_B \in \mathscr{B}$.
The constraints of the optimization problem~\eqref{prob:robust_general} ensure that if the test deck is cast into a voting machine,  and if the voting machine is operating incorrectly in any of the ways specified by the uncertainty set, then the output of the voting machine will be different from the output of a voting machine that is operating correctly. In other words, if $(B,\beta_1,\ldots,\beta_B)$ is an optimal solution for the optimization problem~\eqref{prob:robust_general}, then we have a guarantee that the output of a voting machine  $\widehat{T}(\beta_1,\ldots,\beta_B)$ will be different from the output of a voting machine that operates  correctly  $T^*(\beta_1,\ldots,\beta_B)$ whenever the voting machine $\widehat{T}(\cdot)$ is operating incorrectly in any of the ways specified by the uncertainty set $\mathcal{U}$. We elaborate on the construction of the uncertainty set in \S\ref{sec:math:swap}. The objective of the optimization problem~\eqref{prob:robust_general} is find a test deck that satisfies the constraints that consists of the fewest number of ballots. 

The optimization problem~\eqref{prob:robust_general} for designing test decks in LAT can be  viewed as attractive  from the perspective of the relevant  stakeholders including election administrators, policy makers, voters, and the computer security community. We elaborate below on the attractiveness and the design of the optimization problem~\eqref{prob:robust_general} through the perspectives of these various stakeholders: 

\paragraph{Election Administrators.} The administration of U.S.\ elections is a complicated endeavor, conducted in parallel by thousands of local officials across the country. Across so diverse and decentralized a system, even marginal increases to the difficulty or complexity of election procedures carry a very high administrative cost. This cautions against testing procedures that are substantially more difficult or resource intensive than those already in use.

 From an implementation standpoint, RLAT aims to minimize the burden the solution confers to election administrators. By finding a test deck that minimizes the number of ballots, the optimization problem~\eqref{prob:robust_general} yields a suitable test deck that minimizes the time it takes to fill out and insert the decks into a machine. Moreover, a computer algorithm for solving the optimization problem~\eqref{prob:robust_general} can be integrated seamlessly at many stages, like in the vendor provided Election Management System (EMS) or by third-party ballot-printing companies that are often contracted to prepare test decks under current practice. This means we can implement RLAT through the operational processes that election administrators already have in place, with minimal change to the official's direct experience. Finally, using a computer algorithm can relieve election workers from the arduous task of manually designing test decks.

\paragraph{Policy Makers and Voters.} RLAT is attractive for policy makers and voters because it enables LAT to provide strong and interpretable guarantees regarding the security of an election and the legitimacy of its outcome. Indeed, the optimization problem~\eqref{prob:robust_general} provides policy makers with the flexibility to specify the uncertainty set of possible ways that a voting machine may operate incorrectly. Policy makers can make this decision based on their evaluation of the cost-security trade-offs for their own state or jurisdiction, and based on factors like the known traits of the voting machines that are used in their elections (see Remark~\ref{remark:uncertainty_set} in \S\ref{sec:math:swap}). The optimization problem~\eqref{prob:robust_general} can also easily  integrate minimum legal requirements on test decks that are specified by policy makers, as we elaborate in Appendix~\ref{appx:statelaws}. Moreover, given a defined uncertainty set $\mathcal{U}$ and a test deck $(\beta_1, \ldots, \beta_B)$, any voter can verify that the test  deck satisfies the constraints of the optimization problem~\eqref{prob:robust_general}. This can enhance voter confidence that testing is being conducted fairly, and provides concrete and voter-verifiable assurances that the outcome of the election has not been accidentally or maliciously altered in ways similar to those from the examples given in \S\ref{sec:vulnerabilities}. 

\paragraph{Computer Security Community.} In the election security community, and in the computer security  community more broadly, risk is defined and minimized by considering a hypothetical {adversary}. This adversary aims to interfere with a system, and is constrained by a {threat model} which specifies the scope of their capabilities. This allows for the development of security interventions which have a definite effect with respect to certain assumptions about the options available to an adversary. The  optimization problem~\eqref{prob:robust_general} thus works constructively with the computer security mindset. Indeed, the uncertainty set $\mathcal{U}$ is in essence a formalization of the threat model---it describes the potential modifications to the machine, which the adversary is able to choose between. By changing the uncertainty set, this formulation can flexibly substitute threat models as needed.

\subsection{RLAT with the  Swap Uncertainty Set} \label{sec:math:swap}
The key to achieving strong and interpretable security  guarantees through RLAT  is selecting an appropriate uncertainty set $\mathcal{U}$ in the optimization problem~\eqref{prob:robust_general}. If the uncertainty set accounts for only a small number of possible misconfigurations or errors, then the security guarantees afforded by RLAT will be limited. On the other hand, if the uncertainty set is overly expansive, then \eqref{prob:robust_general} might yield a test deck comprised of an impractically large number of ballots.  Naturally, the task of choosing an uncertainty set that strikes an appropriate balance between the expressiveness  and  conservatism is a central challenge when constructing robust optimization problems such as \eqref{prob:robust_general}
.
\looseness=-1

We focus throughout this paper on solving the optimization problem~\eqref{prob:robust_general} with a  specific construction of the uncertainty set that we henceforth refer to as the {swap uncertainty set}. The swap uncertainty set consists of all of the voting machines that have an incorrect bijective mapping from candidates to voting targets. Hence, the optimization problem~\eqref{prob:robust_general} with the swap uncertainty set will  yield the shortest test deck  that is guaranteed to detect whether a voting machine is swapping votes across candidates. The optimization problem~\eqref{prob:robust_general} with the swap uncertainty set  is stated formally at the end of the present \S\ref{sec:math:swap} as the optimization problem~\eqref{prob:robust}.

The swap uncertainty set is attractive from a security standpoint because it encompasses a general class of  misconfigurations that would be difficult to detect for a well-implemented voting machine. 
The premise of the swap uncertainty set when scanning hand-marked ballots is that a voting machine is configured with a $(x,y)$ coordinate for each candidate, which specifies the location of that candidate's voting target on the physical ballot. A well-implemented voting machine's software should perform two basic sanity checks of this configuration. First, the voting machine should not allow any candidate to be associated with multiple targets. Second, the voting machine should not allow different candidates' targets to overlap. This ensures a bijective mapping between candidates and targets. The premise of the swap uncertainty set when scanning ballots produced with a ballot-marking device (BMD) is that the BMD and optical scanner may be configured with inconsistent data representations of the candidates~\cite{dekalb2022}, such that votes encoded by the BMD as corresponding to one candidate may be read by the scanner as corresponding to another. Well-implemented software should enforce that each candidate has precisely one data representation, so this mismatch must also be a bijection.
The swap uncertainty set is thus a natural choice for the RLAT problem under either of these models, since it describes each possible mapping from candidates to targets.

The formal definition of the swap uncertainty set requires the following additional notation. Let $\Sigma$  denote the set of all {non-identity bijections} of the form $\sigma: \mathcal{N} \to \mathcal{N}$, where we say that the function $\sigma(\cdot)$  is a  non-identity bijection if and only if the function satisfies the following two criteria:\looseness=-1
\begin{enumerate}
    \item For every target $j \in \mathcal{N}$, there exists one candidate $i \in \mathcal{N}$ that satisfies $\sigma(i) = j$. 
    \item There exists $i \in \mathcal{N}$ that satisfies $\sigma(i) \neq i$. 
\end{enumerate}
Each non-identity bijection can be understood as an incorrect mapping from candidates to targets.\footnote{We note that the voting machine that operates correctly can be represented by the identity function $* : \mathcal{N} \to \mathcal{N}$, defined as the function that satisfies the equality $*(i) = i$ for all $i \in \mathcal{N}$.}  The output of a voting machine whose mapping from  candidates to targets is the bijection $\sigma: \mathcal{N} \to \mathcal{N}$ is given for each candidate $i \in \mathcal{N}_c$ in each contest $c \in \mathcal{C}$ by 
\begin{align*}
    T^\sigma_i(\beta_1,\ldots,\beta_B) \triangleq \sum_{b=1}^B  \mathbb{I}\left \{   \sigma(i) \in \beta_b  \textnormal{ and } \left| \left \{ \sigma(j) \in \beta_b: j \in \mathcal{N}_c \right \}  \right| \leq v_c \right \}. 
\end{align*}
To make sense of the above definition, we remark that the inclusion $\sigma(i) \in \beta_b$ holds if and only if the $b$-th filled-out ballot in the test deck is interpreted by a voting machine with mapping $\sigma$ to  contain a vote for candidate $i$. Similarly, we observe that the inequality $\left| \left \{ \sigma(j) \in \beta_b: j \in \mathcal{N}_c \right \}  \right| \le v_c$ holds if and only if the $b$-th filled-out ballot in the test deck is interpreted by a voting machine with mapping $\sigma$ as containing votes for at most $v_c$ candidates in contest $c$.

In view of the above notation, we define the swap uncertainty set as the set of the voting machines that correspond to each of the non-identity bijections: 
\begin{align*}
\mathcal{U}  \triangleq \left \{ T^\sigma(\cdot) \equiv (T^\sigma_1(\cdot),\ldots,T^\sigma_N(\cdot)) : \sigma \in \Sigma \right \}. 
\end{align*}
 Hence, we conclude that a test deck comprised of filled-out ballots $\beta_1,\ldots,\beta_B \in \mathscr{B}$ will satisfy the constraints of the optimization problem~\eqref{prob:robust_general} with the swap uncertainty set if and only if the test deck is guaranteed to detect whether a voting machine has been misconfigured to swap votes across candidates.\looseness=-1

Equipped with the swap uncertainty set,  we are  ready to formally state the key optimization problem of this paper, that is, the optimization problem of finding a minimum-length test deck that is guaranteed to detect whether a voting machine has been misconfigured to swap votes across candidates. This optimization problem~\eqref{prob:robust_general} with the swap uncertainty set is stated below as \eqref{prob:robust}:\looseness=-1 
  \begin{align} \label{prob:robust} \tag{RO-$\Sigma$}
\begin{aligned}
\underset{B \in \N, \; \beta_1,\ldots,\beta_B \in \mathscr{B}}{\textnormal{minimize}}\quad & B\\
\textnormal{subject to} \quad &  T^{\sigma}(\beta_1,\ldots,\beta_B) \neq T^*(\beta_1,\ldots,\beta_B) \quad \forall \sigma \in \Sigma.
\end{aligned}
\end{align}

Having established that the test decks obtained by \eqref{prob:robust} offer  attractive and rigorous security guarantees, we show in the rest of this paper  that \eqref{prob:robust} leads to test decks that can be   practically deployed  in real world elections. In \S\ref{appx:cutting_sec} and \S\ref{appx:improvements}, we develop an exact algorithm for solving the optimization problem~\eqref{prob:robust}. In  \S\ref{sec:experiments}, we show that our exact algorithm scales to Michigan's November 2022 elections and that the test decks obtained by  \eqref{prob:robust} in those elections are not much longer than the test decks produced according to the heuristics Michigan currently uses. Hence, RLAT with the swap uncertainty set strikes a balance  between producing test decks that account for a large number of possible voting machine misconfigurations and producing test decks with a practically small number of ballots.

\begin{remark} \label{remark:uncertainty_set}
Although this paper focuses on solving \eqref{prob:robust}, 
we note that RLAT offers election officials the flexibility to use uncertainty sets that include a more expansive or narrow model of the ways in which a voting machine could be wrong.
 For instance, in  states that currently use weaker heuristics than Michigan's to prepare their test decks, election officials may be accustomed to using very short test decks and thus might balk at the lengths of test decks produced by \eqref{prob:robust}. To accommodate election officials in such states, one can solve \eqref{prob:robust_general} with an uncertainty set that is a subset of the swap uncertainty set to obtain shorter test decks with weaker, albeit still rigorous defined, security guarantees (e.g. by opting to ignore the possibility of swaps between candidates in noncompetitive contests). Conversely, the swap uncertainty set can be made more expansive (e.g. by considering cases where the mapping of targets to candidates need not be bijective for voting machines whose software implementation allows the same target to be associated with multiple candidates, or vice versa). That being said, we emphasize that the algorithms presented in this paper are designed for solving  \eqref{prob:robust}, i.e., the specific case of  \eqref{prob:robust_general} in which the uncertainty set is the swap uncertainty set.    
\end{remark}

 \subsection{Discussion} \label{sec:math:discussion}

We conclude \S\ref{sec:math} by characterizing the key structural properties of test decks that satisfy the constraints of the optimization problem~\eqref{prob:robust}.  Specifically, the main contribution of \S\ref{sec:math:discussion} is a technical result, denoted below by Theorem~\ref{thm:diff_votes},   that characterizes  the  situations in which the output of a voting machine that operates correctly will be different from the output of a voting machine whose mapping from candidates  
  to targets is a non-identity bijection. The characterization established by the following theorem will be used extensively for designing algorithms in the rest of the sections. 
\begin{theorem}\label{thm:diff_votes}
Let $\beta_1, \ldots, \beta_B \in \mathscr{B}$ and $\sigma \in \Sigma$. Then  $T^\sigma(\beta_1,\ldots,\beta_B) \neq T^*(\beta_1,\ldots,\beta_B)$ if and only if at least one of the following two conditions hold:
\begin{itemize}
    \item  There exists a candidate $i \in \mathcal{N}$ that satisfies $$| \{ b \in \{1,\ldots,B\}: i \in \beta_b \}|  \neq  |\{ b \in \{1,\ldots,B\}: \sigma(i) \in \beta_b \}|.$$
    \item There exist a contest $c \in \mathcal{C}$ and a filled-out ballot $\beta_b$ for some $b \in \{1,\ldots,B\}$ that satisfy $$\left| \left \{ \sigma(j) \in \beta_b: j \in \mathcal{N}_c \right \}  \right| > v_c.$$
\end{itemize}
\end{theorem}
\noindent The proof of this theorem and all other technical proofs in this paper can be found in Appendix~\ref{appx:proofs}.

In words, the above theorem establishes that the output of a voting machine that operates correctly will not equal  the output of a voting machine whose mapping from candidates   to targets is a non-identity bijection $\sigma \in \Sigma$ if and only if the test deck comprised of filled-out ballots $\beta_1,\ldots,\beta_B \in \mathscr{B}$ satisfies at least one of two conditions. The
 first condition is that there exists a candidate $i \in \mathcal{N}$ such that the number of filled-out ballots that selected target $i$,  $| \{ b \in \{1,\ldots,B\}: i \in \beta_b \}|$,  is different from the number of filled-out ballots that selected target $\sigma(i)$, $|\{ b \in \{1,\ldots,B\}: \sigma(i) \in \beta_b \}|$. The second condition is that there exists a contest $c \in \mathcal{C}$ in one of the ballots $b \in \{1,\ldots,B\}$  that is interpreted as overvoted by the voting machine whose mapping from candidates to targets is $\sigma$. 
 As an immediate corollary of Theorem~\ref{thm:diff_votes}, we obtain the following characterization of the test decks that are feasible for the optimization problem~\eqref{prob:robust}. 

 \begin{corollary} \label{cor:diff_votes}
A  tuple  $(B,\beta_1, \ldots, \beta_B)$ is feasible for the  optimization problem~\eqref{prob:robust} if and only if  $B \in \N$, $\beta_1,\ldots,\beta_B \in \mathscr{B}$, and for every $\sigma \in \Sigma$, at least one of the following two conditions hold:
\begin{itemize}
    \item  There exists a candidate $i \in \mathcal{N}$ that satisfies $$| \{ b \in \{1,\ldots,B\}: i \in \beta_b \}|  \neq  |\{ b \in \{1,\ldots,B\}: \sigma(i) \in \beta_b \}|.$$
    \item { There exist a contest $c \in \mathcal{C}$ and a filled-out ballot $\beta_b$ for some $b \in \{1,\ldots,B\}$ that satisfy $$\left| \left \{ \sigma(j) \in \beta_b: j \in \mathcal{N}_c \right \}  \right| > v_c.$$}
\end{itemize}
\end{corollary}

 Corollary~\ref{cor:diff_votes} implies  that the optimization problem~\eqref{prob:robust}  always has a feasible solution. Specifically, it follows  from Corollary~\ref{cor:diff_votes} that any test deck that gives a distinct total number of votes to each of the candidates across all of the contests  is a feasible solution for the optimization problem~\eqref{prob:robust} (a formal proof of this can be found in the proof of Proposition~\ref{prop:heuristic1} from Appendix~\ref{appx:upperbound}). We note that while the Corollary~\ref{cor:diff_votes} implies that a feasible solution for optimization problem~\eqref{prob:robust} can be obtained by simply giving every candidate across every contest a distinct number of votes, we show in Appendix~\ref{appx:upperbound} using real-world data that  heuristics based on  assigning a distinct number of votes to each candidate will result in test decks that contain too many ballots to be implementable in practice. Thus motivated, we proceed in \S\ref{appx:cutting_sec} to develop an exact algorithm which solves the optimization problem~\eqref{prob:robust} in order to find test decks that are feasible solutions for \eqref{prob:robust} with the fewest possible number of ballots.

\section{Exact Algorithm} \label{appx:cutting_sec}
In this section, we present our exact algorithm for solving the optimization problem~\eqref{prob:robust}.

\subsection{Overview of Exact Algorithm}\label{appx:cutting}
 To begin our discussion of our exact algorithm, we recall from \S\ref{sec:intro:background} that one of the key challenges in solving the optimization problem~\eqref{prob:robust}  is that the problem contains an enormous number of constraints. Indeed, we observe that the number of constraints in the optimization problem~\eqref{prob:robust} is driven by the cardinality of the set of non-identity bijections $\Sigma$, and it follows readily from \S\ref{sec:math:swap} that the number of non-identity bijections satisfies $| \Sigma| = N!-1$  for a ballot style with $N$ candidates.  Because the number of candidates in real-world ballot styles often satisfies $N \ge 100$, an optimization problem that explicitly encodes all possible non-identity bijections thus cannot be represented nor solved on any extant computer for many real-world elections. It is currently unknown whether  \eqref{prob:robust}  is NP-hard, and it is unknown whether there exists a  mixed-integer linear programming reformulation of \eqref{prob:robust} of size that is polynomial in the number of candidates $N$.

To contend with the computational challenge of solving the optimization problem~\eqref{prob:robust}, we draw inspiration from  an algorithmic strategy known as the {cutting plane method} \cite{kelley1960cutting,gilmore1961linear}. The goal of the cutting plane method is to circumvent the need to solve an optimization problem with a large number of constraints by solving a sequence of optimization problems with small numbers of constraints. The  application of the cutting plane method to the optimization problem~\eqref{prob:robust} takes the form of an iterative algorithm that is 
described below and visualized in Figure~\ref{fig:cutting}.

\begin{figure}[t]
    \centering
\includegraphics[width=0.85\linewidth]{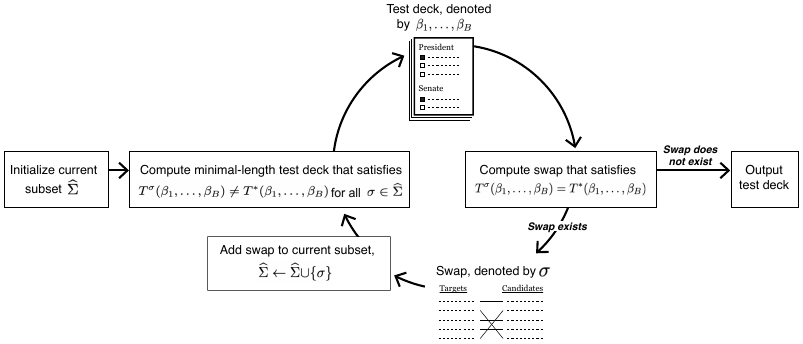}

\caption{Visualization of our exact algorithm from \S\ref{appx:cutting} for solving the optimization problem~\eqref{prob:robust}. } \label{fig:cutting}

\end{figure}

In each iteration of our algorithm, we start with a  subset of non-identity bijections $\widehat{\Sigma} \subseteq \Sigma$, and  we solve the following variant  of the optimization problem~\eqref{prob:robust}: 
\begin{align} \tag{RO-$\widehat{\Sigma}$} \label{prob:robust-subset}
\begin{aligned}
\underset{B \in \N, \; \beta_1,\ldots,\beta_B \in \mathscr{B}}{\textnormal{minimize}}\quad & B\\
\textnormal{subject to} \quad &  T^{\sigma}(\beta_1,\ldots,\beta_B) \neq T^*(\beta_1,\ldots,\beta_B) \quad \forall \sigma \in \widehat{\Sigma}.
\end{aligned}
\end{align}
To make sense of the optimization problem~\eqref{prob:robust-subset}, let us reflect on the relationship between \eqref{prob:robust-subset} and \eqref{prob:robust}. We observe that the optimization problem~\eqref{prob:robust-subset} is nearly identical to the optimization problem~\eqref{prob:robust}, with the only  difference being that the former only has a constraint for each $\sigma \in \widehat{\Sigma}$ instead of a constraint for each $\sigma \in \Sigma$. The optimization problem~\eqref{prob:robust-subset} can thus be viewed as a relaxation of the optimization problem~\eqref{prob:robust}, in the sense that the optimal objective value of the optimization problem~\eqref{prob:robust-subset} is less than or equal to the optimal objective value of the optimization problem~\eqref{prob:robust-subset}, but an optimal solution for the optimization problem~\eqref{prob:robust-subset} might not be a feasible solution for the optimization problem~\eqref{prob:robust}. The potential attractiveness of the optimization problem~\eqref{prob:robust-subset} can be attributed to  practical tractability: if the cardinality of $\widehat{\Sigma}$ is significantly less than the cardinality of $\Sigma$, then it will be possible to solve the  optimization problem~\eqref{prob:robust-subset}   much faster by a computer compared to the optimization problem~\eqref{prob:robust}.

After computing an optimal solution for the optimization problem~\eqref{prob:robust-subset}, the next step of the current iteration of our algorithm is determining whether  the optimal solution for  the optimization problem~\eqref{prob:robust-subset} is a feasible solution for the optimization problem~\eqref{prob:robust}. This step is performed by solving the optimization problem
\begin{equation} \tag{CUT} \label{prob:oracle}
\begin{aligned}
&\underset{\sigma \in \Sigma}{\textnormal{minimize}}&& 0\\
&\textnormal{subject to}  &&T^\sigma(\beta_1,\ldots,\beta_{B}) = T^*(\beta_1,\ldots,\beta_{B}),
\end{aligned}
\end{equation}
where $(B,\beta_1,\ldots,\beta_{B})$ denotes the optimal solution of the optimization problem~\eqref{prob:robust-subset}. The  optimization problem~\eqref{prob:oracle}  has two possible outputs. First, if the  optimization problem~\eqref{prob:oracle}  outputs an optimal solution $\sigma \in \Sigma$, then we  conclude that $(B,\beta_1,\ldots,\beta_{B})$ is not a feasible solution of the optimization problem~\eqref{prob:robust}, since the constraint  $T^{\sigma}(\beta_1,\ldots,\beta_B) \neq T^*(\beta_1,\ldots,\beta_B)$  in the optimization problem~\eqref{prob:robust} is violated by the test deck $(\beta_1,\ldots,\beta_B)$. 
Second, if the   optimization problem~\eqref{prob:oracle} does not have any optimal solution, then we conclude that $(B,\beta_1,\ldots,\beta_{B})$ is a feasible solution for the optimization problem~\eqref{prob:robust}. 

The final step of each iteration of the algorithm depends on the output of the optimization problem~\eqref{prob:oracle}. If that optimization problem does not have a feasible solution, then we observe that $(B,\beta_1,\ldots,\beta_{B})$ must be a feasible solution for the optimization problem~\eqref{prob:robust}. Moreover, since $(B,\beta_1,\ldots,\beta_{B})$ was an optimal solution for the optimization problem~\eqref{prob:robust-subset}, and since the optimization problem~\eqref{prob:robust-subset} is a relaxation of the optimization problem~\eqref{prob:robust}, it must be the case that $(B,\beta_1,\ldots,\beta_{B})$ is also an optimal solution for the optimization problem~\eqref{prob:robust}. Hence, if the optimization problem~\eqref{prob:oracle} does not have a feasible solution, then we have found an optimal solution to the optimization problem~\eqref{prob:robust}, and the algorithm terminates. Otherwise, if the optimization problem~\eqref{prob:oracle} outputs $\sigma \in \Sigma$,     then we conclude the current iteration by updating $\widehat{\Sigma} \leftarrow \widehat{\Sigma} \cup \{ \sigma \}$ and starting a new iteration of the algorithm. 

It follows from straightforward arguments that the algorithm described above will terminate at an optimal solution for the optimization problem~\eqref{prob:robust} after finitely many iterations, regardless of the choice of the subset $\widehat{\Sigma} \subseteq \Sigma$ in the first iteration. Indeed, the finite convergence of the algorithm follows from the fact that $\Sigma$ is a finite set and from the fact that the optimization problem~\eqref{prob:oracle} will never output a non-identity bijection $\sigma \in \Sigma$ that is an element of $\widehat{\Sigma}$ when the test deck  $(\beta_1,\ldots,\beta_{B})$ satisfies the constraints of the optimization problem~\eqref{prob:robust-subset}. Therefore, the number of iterations of the algorithm is always upper bounded by $| \Sigma| = N! - 1$. In our numerical experiments throughout this paper, we initialize $\widehat{\Sigma}$ in the first iteration to be the empty set.

In order for the algorithm described above to be practically efficient in real-world elections, three important properties must hold. First, it must be possible to quickly solve the optimization problem~\eqref{prob:robust-subset} when $| \widehat{\Sigma}| \ll | \Sigma|$. Second, the algorithm must terminate after a relatively small number of iterations, as this property is essential for ensuring that the cardinality of $\widehat{\Sigma}$ remains much smaller than the cardinality of $ \Sigma$. Third, it must be possible to have a fast implementation of the optimization problem~\eqref{prob:oracle} for finding a constraint that is violated by the test deck obtained by solving the optimization problem~\eqref{prob:robust-subset}. In the subsequent \S\ref{appx:mip} and \S\ref{appx:improvements}, we show that these three important properties for obtaining a practically efficient algorithm can be achieved simultaneously.\looseness=-1

\subsection{Mixed-Integer Reformulations} \label{appx:mip}

In each iteration of the cutting plane method from \S\ref{appx:cutting}, we are tasked with solving the optimization problems~\eqref{prob:robust-subset} and \eqref{prob:oracle}. Here, we show that optimal solutions for these two optimization problems can be obtained by solving mixed-integer linear optimization problems. In doing so, this subsection enables the cutting plane method from \S\ref{appx:cutting} to be easily implementable using widely available open-source and commercial optimization software such as Gurobi and Mosek. Improvements to the mixed-integer linear optimization reformulations from the present subsection are proposed and analyzed in the subsequent \S\ref{appx:improvements}.

\subsubsection{Mixed-Integer Reformulation of \eqref{prob:robust-subset}} \label{appx:mip:robust-subset}
At a high level, our procedure for solving the optimization problem~\eqref{prob:robust-subset} consists of the following steps. First,  we  fix $B$ to be an integer that is less than or equal to the optimal objective value of the optimization problem~\eqref{prob:robust-subset}. We then solve the following optimization problem:
\begin{align} \tag{RO-$\widehat{\Sigma}$-$B$} \label{prob:robust-subset-B}
\begin{aligned}
\underset{\beta_1,\ldots,\beta_B \in \mathscr{B}}{\textnormal{minimize}}\quad & 0\\
\textnormal{subject to} \quad &  T^{\sigma}(\beta_1,\ldots,\beta_B) \neq T^*(\beta_1,\ldots,\beta_B) \quad \forall \sigma \in \widehat{\Sigma}.
\end{aligned}
\end{align}
If the optimization problem~\eqref{prob:robust-subset-B} does not have any feasible solutions, then we observe that $B$ must be strictly less than the optimal objective value of the optimization problem~\eqref{prob:robust-subset}. In that case, we update $B \leftarrow B+1$ and re-solve the optimization problem~\eqref{prob:robust-subset-B} with the new value for the parameter $B$. We repeat this loop until the optimization problem~\eqref{prob:robust-subset-B} is feasible, at which point~\eqref{prob:robust-subset-B} will yield an optimal solution for the optimization problem~\eqref{prob:robust-subset}. 
\begin{remark}
    We use the above procedure to solve the optimization problem~\eqref{prob:robust-subset} because the number of decision variables in the optimization problem~\eqref{prob:robust-subset} depends on the integer $B$, where the integer $B$ is itself a decision variable in the optimization problem~\eqref{prob:robust-subset}. In contrast, the number of decision variables in the optimization problem~\eqref{prob:robust-subset-B} is known because the integer $B$ is fixed externally. Because the number of decision variables is known a priori,  the optimization problem~\eqref{prob:robust-subset-B} can be  reformulated as a mixed-integer linear optimization problem.
\end{remark}
\begin{remark}
    The above procedure requires $B$ to be initialized to an integer that is less than or equal to the optimal objective value of the optimization problem~\eqref{prob:robust-subset}. In our implementation of the procedure, we initialize the integer to $B \leftarrow 1$ in the first iteration of the cutting plane method. In all subsequent iterations of the cutting plane method, we initialize $B$ to the optimal objective value of the optimization problem~\eqref{prob:robust-subset} from the previous iteration of the cutting plane method.\footnote{We recall that  $\widehat{\Sigma}$ in the current iteration of the cutting plane method is a superset of $\widehat{\Sigma}$ from the previous iteration of the cutting plane method. As a result, the constraints of the  optimization problem~\eqref{prob:robust-subset} in the current iteration of the cutting plane method is always a strict superset of the constraints of the optimization problem~\eqref{prob:robust-subset} in the previous iteration of the cutting plane method. This implies that the optimal objective value of the optimization problem~\eqref{prob:robust-subset}  from the previous iteration of the cutting plane method is always less than or equal to the optimal objective value of the optimization problem~\eqref{prob:robust-subset} in the current iteration of the cutting plane method.   }
\end{remark}

In the remainder of \S\ref{appx:mip:robust-subset}, we show that the optimization problem~\eqref{prob:robust-subset-B} can be reformulated as a mixed-integer linear optimization problem.   Indeed, let  $\mathcal{B} \triangleq \{1,\ldots,B\}$ and  $\mathcal{B}_0 \triangleq \{ 0 \} \cup \mathcal{B}$. With this notation, we first observe that the optimization problem~\eqref{prob:robust-subset-B} can be rewritten equivalently  as the following intermediary optimization problem:  
\begin{subequations} \label{prob:complete_transposition} 
\begin{align}
\underset{\beta \in \{0,1\}^{\mathcal{B} \times \mathcal{N}}}{\textnormal{minimize}}\quad & 0 \label{prob:complete_transposition:obj}\\
\textnormal{subject to} \quad & \sum_{i \in \mathcal{N}_c} \beta_{b,i} \le v_c &&\forall b \in \mathcal{B}, c \in \mathcal{C} \label{prob:complete_transposition:beta_feas}\\
& T^\sigma(\{i: \beta_{1,i} = 1 \},\ldots,\{i: \beta_{B,i} = 1 \}) \notag \\ 
&\;\; \neq T^*(\{i: \beta_{1,i} = 1 \},\ldots,\{i: \beta_{B,i} = 1 \}) && \forall \sigma \in \widehat{\Sigma}. \tag{--} \label{prob:complete_transposition:beta_not_defeated}
\end{align}
The optimization problem~\eqref{prob:complete_transposition} can be interpreted as follows. Each  binary decision variable $\beta_{b,i} \in \{0,1\}$ is equal to one if and only if target $i$ is selected in the $b$-th filled-out ballot in the test deck.  Hence, each vector $\beta_b \equiv (\beta_{b,1},\ldots,\beta_{b,N}) \in \{0,1\}^{\mathcal{N}}$ serves as a binary encoding of the targets that are selected in the $b$-th filled-out ballot. 
  Constraint~\eqref{prob:complete_transposition:beta_feas} ensures that each filled-out ballot in the test deck is feasible, that is, there does not exist a filled-out ballot that contains more votes for candidates in a contest than are allowed. Constraint~\eqref{prob:complete_transposition:beta_not_defeated} says that the output of a voting machine whose mapping from  candidates to targets is  $\sigma \in \widehat{\Sigma}$ must be different from the output of a correctly operating voting machine when using the test deck $(\{i: \beta_{1,i} = 1 \},\ldots,\{i: \beta_{B,i} = 1 \})$.

We observe from inspection that the above optimization problem~\eqref{prob:complete_transposition} consists exclusively of binary decision variables. Moreover, the objective function~\eqref{prob:complete_transposition:obj} and the constraints~\eqref{prob:complete_transposition:beta_feas}  are linear functions of the decision variables. Therefore, the final step in our reformulation of the optimization problem~\eqref{prob:robust-subset-B} as a mixed-integer linear optimization problem is to reformulate the constraint~\eqref{prob:complete_transposition:beta_not_defeated}.  This can be done through the introduction of new decision variables $\gamma$, $y$, and $p$, resulting in the following mixed-integer linear optimization problem:
\begin{align}
\underset{\substack{\beta \in \{0,1\}^{\mathcal{B} \times \mathcal{N}}\\
\gamma \in \{0,1\}^{\mathcal{N} \times \mathcal{B}_0}, \; y \in \R_{\ge 0}^{\mathcal{N} \times \mathcal{N}}\\
p^\sigma \in \{0,1\}^{\mathcal{B} \times \mathcal{C}} \forall \sigma \in \widehat{\Sigma}}}{\textnormal{minimize}}\quad & \eqref{prob:complete_transposition:obj} \notag \\
\textnormal{subject to} \quad &  \eqref{prob:complete_transposition:beta_feas} \notag \\
& \sum_{g \in \mathcal{B}_0} \gamma_{i,g} = 1 && \forall i \in \mathcal{N} \label{prob:complete_transposition:gamma_1}\\
& \sum_{b \in \mathcal{B}} \beta_{b,i} = \sum_{g \in \mathcal{B}_0} g \gamma_{i,g}  && \forall i \in \mathcal{N} \label{prob:complete_transposition:gamma_2}\\
& y_{i,j} \ge -1 + \gamma_{i,g} + \gamma_{j,g}  && \forall i,j \in \mathcal{N}, g \in \mathcal{B}_0 \label{prob:complete_transposition:y}\\
& p^\sigma_{b,c} \geq 1 - \frac{1}{v_c + 1}  \sum_{i \in \mathcal{N}_c} \beta_{b, \sigma(i)} && \forall \sigma \in \widehat{\Sigma}, b \in \mathcal{B}, c \in \mathcal{C}\label{prob:complete_transposition:p}\\
& \sum_{i \in \mathcal{N} } \left(1 -  y_{i,\sigma(i)} \right) + \sum_{b \in \mathcal{B}} \sum_{c \in \mathcal{C}} \left(1 - p^\sigma_{b,c} \right) \ge 1&& \forall \sigma \in \widehat{\Sigma}.\label{prob:complete_transposition:defeated}
\end{align}
  \end{subequations}
 The decision variables of the above optimization problem can be interpreted as follows.
 
 First, we observe that  constraints~\eqref{prob:complete_transposition:gamma_1} and \eqref{prob:complete_transposition:gamma_2} together enforce that each decision variable $\gamma_{i,g} \in \{0,1\}$ will be equal to one if and only if candidate $i$ appears in exactly $g$ ballots. 
 
Second, constraint~\eqref{prob:complete_transposition:y} requires that $y_{i,j}$ must be 
greater than or equal to one if candidates $i$ and $j$ receive the same number of votes across $\mathcal{B}$, and may be as low as zero if they receive a different number of votes. Similarly, constraint~\eqref{prob:complete_transposition:p} requires that $p_{b,c}^\sigma$ must be equal to one if contest $c$ receieves $v_c$ or fewer votes on ballot $b$ under swap $\sigma$, and may be zero if the contest is instead overvoted.\looseness=-1

Finally, constraint~\eqref{prob:complete_transposition:defeated} ensures that a feasible solution for the optimization problem~\eqref{prob:complete_transposition} exists if and only if at least one $y_{i,\sigma(i)}$ or $p_{b,c}^\sigma$ is zero for each $\sigma \in \widehat{\Sigma}$. In other words, the problem has a feasible solution if and only if, for each swap in $\widehat{\Sigma}$, there exists at least one candidate that is mapped to a target with a different number of votes \emph{or} there exists at least one ballot that is unexpectedly interpreted as containing an overvote. This ensures that any feasible solution to this optimization problem corresponds to a deck of ballots which will detect every swap in our subset.
 
In summary, we have shown in the present \S\ref{appx:mip:robust-subset} that the optimization problem~\eqref{prob:robust-subset} can be solved by a procedure that consists of fixing the integer $B$ to a lower bound on the optimal objective value of the optimization problem~\eqref{prob:robust-subset} and then incrementing $B$ until the optimization problem~\eqref{prob:robust-subset-B} has a feasible solution. Moreover, for each fixed choice of the integer $B$, we showed that the optimization problem~\eqref{prob:robust-subset-B} can be reformulated as the mixed-integer linear optimization problem~\eqref{prob:complete_transposition}. Thus, we have shown that an optimal solution for the optimization problem~\eqref{prob:robust-subset} can be obtained by a procedure that consists of solving one or more mixed-integer linear optimization problems. 

\subsubsection{Mixed-Integer Reformulation of \eqref{prob:oracle}} \label{appx:mip:oracle}
We conclude \S\ref{appx:mip} by reformulating the optimization problem~\eqref{prob:oracle} as a mixed-integer linear optimization problem.\footnote{More precisely, our reformulation~\eqref{prob:oracle_mip} of the optimization problem~\eqref{prob:oracle} is a \emph{binary} linear optimization problem.} Given any test deck $(\beta_1,\ldots,\beta_{B})$, our mixed-integer linear optimization reformulation of the optimization problem~\eqref{prob:oracle}  is the following:
\begin{subequations} \label{prob:oracle_mip}
    \begin{align}
    \underset{x \in \{0,1\}^{\mathcal{N} \times \mathcal{N}}}{\textnormal{minimize}} \quad\quad & 0 \notag \\
    \textnormal{subject to} \quad\quad  & \sum_{j \in \mathcal{N}} x_{i,j} = 1 && \forall i \in \mathcal{N}\label{prob:oracle_mip:bijective_1} \\
    & \sum_{i \in \mathcal{N}} x_{i,j} = 1 && \forall j \in \mathcal{N}\label{prob:oracle_mip:bijective_2} \\
         & \sum_{i \in \mathcal{N}} x_{i,i} \leq |\mathcal{N}|-2 \label{prob:oracle_mip:dishonest}\\
    &  \sum_{i \in \mathcal{N}_c} \sum_{j \in \beta_b} x_{i,j} \leq v_c && \forall b \in \mathcal{B}, c \in \mathcal{C} \label{prob:oracle_mip:overvote}\\
     & x_{i,j} = 0 && \forall i,j \in \mathcal{N} : \; \left| \left \{b \in \mathcal{B}: i \in \beta_b \right \} \right|  \neq \left| \left \{b \in \mathcal{B}: j \in \beta_b \right \} \right|. 
     \label{prob:oracle_mip:vote_total} 
    \end{align}
\end{subequations}
The constraints \eqref{prob:oracle_mip:bijective_1}-\eqref{prob:oracle_mip:dishonest} enforce that the decision variables $x \in \{0,1\}^{\mathcal{N} \times \mathcal{N}}$ are a binary encoding of a non-identity bijection $\sigma \in \Sigma$. Indeed,  constraints  \eqref{prob:oracle_mip:bijective_1} and  \eqref{prob:oracle_mip:bijective_2}  ensure that each feasible solution of \eqref{prob:oracle_mip}  can be transformed into a bijection $\sigma: \mathcal{N} \to \mathcal{N}$ using the rule that $\sigma(i) = j$ if and only if $x_{i,j} = 1$ for each $i,j \in \mathcal{N}$. Constraint~\eqref{prob:oracle_mip:dishonest} enforces that there exists $i \in \mathcal{N}$ that satisfies $\sigma(i) \neq i$. 

The last two constraints~\eqref{prob:oracle_mip:overvote} and \eqref{prob:oracle_mip:vote_total} enforce that the non-identity bijection $\sigma \in \Sigma$  corresponding to the decision variables $x \in \{0,1\}^{\mathcal{N} \times \mathcal{N}}$ satisfies the equality $T^\sigma(\beta_1,\ldots,\beta_{B}) = T^*(\beta_1,\ldots,\beta_{B})$. To see why this is the case, we first observe that constraint~\eqref{prob:oracle_mip:overvote} enforces that the inequality $\left| \left \{ \sigma(j) \in \beta_b: j \in \mathcal{N}_c \right \}  \right| \le v_c$ holds for all contests $c \in \mathcal{C}$ and $b \in \mathcal{B}$. Moreover,  constraint~\eqref{prob:oracle_mip:vote_total} enforces that the equality $| \{ b \in \{1,\ldots,B\}: i \in \beta_b \}| =  |\{ b \in \{1,\ldots,B\}: \sigma(i) \in \beta_b \}|$  holds for all candidates $i \in \mathcal{N}$. Therefore, it follows from Theorem~\ref{thm:diff_votes} in \S\ref{sec:math:discussion} that  constraints~\eqref{prob:oracle_mip:overvote} and \eqref{prob:oracle_mip:vote_total} are satisfied if and only if the non-identity bijection $\sigma \in \Sigma$  corresponding to the decision variables $x \in \{0,1\}^{\mathcal{N} \times \mathcal{N}}$ satisfies the equality $T^\sigma(\beta_1,\ldots,\beta_{B}) = T^*(\beta_1,\ldots,\beta_{B})$.

\section{Improvements to Exact Algorithm} \label{appx:improvements}
In this section, we present five improvements to the mixed-integer linear optimization reformulations from \S\ref{appx:mip} that significantly increase the practical efficiency of the cutting plane method  from \S\ref{appx:cutting}. Our five improvements to the mixed-integer linear optimization reformulations are  presented  and analyzed in the subsequent \S\ref{appx:improvements:reducingdecisions}-\S\ref{appx:improvements:noncompetitive}. 
In  Appendix~\ref{appx:additional_experiments}, we demonstrate  via numerical experiments on synthetic elections that   each of the five improvements from this section, when applied in isolation, generates between a $20$x to $3000$x speedup to the cutting plane method.   

\subsection{Improvement 1: Reducing Number of Decision Variables and Constraints} \label{appx:improvements:reducingdecisions}
As our first  step in increasing the practical efficiency of the cutting plane method  from \S\ref{appx:cutting}, we show that a number of the decision variables and constraints in the mixed-integer linear optimization problem~\eqref{prob:complete_transposition} from \S\ref{appx:mip:robust-subset} can be removed without any loss of generality. By removing these unnecessary decision variables and constraints from the mixed-integer linear optimization problem~\eqref{prob:complete_transposition}, we demonstrate through numerical experiments in Appendix~\ref{appx:additional_experiments} that the computation time of each iteration of the cutting plane method can be significantly decreased. 

To motivate our subsequent developments in \S\ref{appx:improvements:reducingdecisions}, we begin by analyzing the size of the mixed-integer linear optimization problem~\eqref{prob:complete_transposition}. Indeed, we observe that the number of binary decision variables, the number of continuous decision variables, and the number of constraints in  the mixed-integer linear optimization problem~\eqref{prob:complete_transposition} are as follows:
\begin{align*}
   & \begin{aligned}\text{\# binary decision variables} =& \underbrace{\left( | \mathcal{B}| \times | \mathcal{N}| \right)}_{ \beta}  + \underbrace{\left(|\mathcal{B}_0| \times | \mathcal{N}| \right)}_{\gamma} + \underbrace{\left( | \mathcal{B}| \times | \mathcal{C}| \times | \widehat{\Sigma} | \right)}_{p}\\
    =&BN +  (B+1)N  +BC | \widehat{\Sigma} |\\
    =& \mathcal{O} \left( BN  +BC | \widehat{\Sigma}|  \right);\end{aligned}
    \\
   & \begin{aligned}
\text{\# continuous decision variables} &= \underbrace{| \mathcal{N}| \times |\mathcal{N}|}_{y} 
\\&=  N^2;\end{aligned}\\
&\begin{aligned}\text{\# constraints} =& \underbrace{\left( | \mathcal{B}| \times| \mathcal{C}| \right)}_{\eqref{prob:complete_transposition:beta_feas}} + \underbrace{|\mathcal{N}|}_{\eqref{prob:complete_transposition:gamma_1}} + \underbrace{|\mathcal{N}|}_{\eqref{prob:complete_transposition:gamma_2}}+ \underbrace{\left( |\mathcal{N}| \times |\mathcal{N}| 
 \times |\mathcal{B}_0| \right) }_{\eqref{prob:complete_transposition:y}}+  \underbrace{\left( | \mathcal{B}| \times | \mathcal{C}| \times \left| \widehat{\Sigma} \right| \right)}_{\eqref{prob:complete_transposition:p}} + \underbrace{| \widehat{\Sigma}|}_{\eqref{prob:complete_transposition:defeated}}\\
 =& BC + 2N + N^2(B+1) + BC | \widehat{\Sigma}| + | \widehat{\Sigma}| \\
 =& \mathcal{O} \left(B   N^2  + BC | \widehat{\Sigma}|  \right). 
 \end{aligned}
\end{align*}
In real-world elections such as those from Michigan, the number of contests typically satisfies $15 \le C \le 40$  (see Figure~\ref{fig:michigan_data} in \S\ref{sec:math:notation}), the number of candidates typically satisfies $60 \le N \le 120$  (see Figure~\ref{fig:michigan_data} in \S\ref{sec:math:notation}),   the number of ballots in an optimal test deck for \eqref{prob:robust} typically satisfies $20 \le B \le 50$ (see Figure~\ref{fig:michigan} in \S\ref{sec:experiments:length}), and the number of iterations of our cutting plane method typically satisfies  $50 \le | \widehat{\Sigma} | \le 300$ (see Figure~\ref{fig:time_vs_iteration} in \S\ref{sec:experiments:time}). Combining the equations derived above with the real-world data observed from Michigan,  we conclude that the size of the mixed-integer linear optimization problem~\eqref{prob:complete_transposition} is driven primarily by the binary decision variables $p$, the  continuous decision variables $y$, and the constraints \eqref{prob:complete_transposition:y} and \eqref{prob:complete_transposition:p}.

In view of the above motivation, we first show that a number of the binary decision variables $p^\sigma_{b,c} \in \{0,1\}$ and constraints~\eqref{prob:complete_transposition:p} can be removed  from the mixed-integer linear optimization problem~\eqref{prob:complete_transposition} without loss of generality. Indeed, we recall from the discussion in \S\ref{appx:mip:robust-subset}  that there  always exists an optimal solution of the mixed-integer linear optimization problem~\eqref{prob:complete_transposition} in which each binary decision variable $p^\sigma_{b,c}$  satisfies the equality $p^\sigma_{b,c} = 0$ if and only if contest $c$ in ballot $b$ is interpreted as containing an overvote by the voting machine whose mapping is $\sigma$. To decrease the number of these binary decision variables, we utilize the following intermediary result:

\begin{lemma} \label{lem:decrease_p}
Let $\sigma \in \widehat{\Sigma}$ and $c \in \mathcal{C}$. If  the inequality 
$\sum_{c' \in \mathcal{C}} \min \left \{ \left| \left \{ \sigma(i) \in \mathcal{N}_{c'}: i \in \mathcal{N}_c \right \} \right|, v_{c'} \right \} \le v_c $ holds, then every feasible solution  of the mixed-integer linear optimization problem~\eqref{prob:complete_transposition} satisfies the equality  $p^\sigma_{b,c} = 1$ for all $b \in \mathcal{B}$. 
\end{lemma}

\noindent To make sense of the above lemma, we remark that $\sum_{c' \in \mathcal{C}} \min \left \{ \left| \left \{ \sigma(i) \in \mathcal{N}_{c'}: i \in \mathcal{N}_c \right \} \right|, v_{c'} \right \}$ is equal to the maximum number of votes that may be mapped to contest $c$ under mapping $\sigma$ for any filled-out ballot in $\mathscr{B}$.  If this summation is less than or equal to $v_c$, then contest $c$ will never be overvoted by a filled-out ballot from $\mathscr{B}$ under mapping $\sigma$.

In view of Lemma~\ref{lem:decrease_p}, we now demonstrate that a subset of the binary decision variables of the form $p^\sigma_{b,c}$ and a subset of the constraints~\eqref{prob:complete_transposition:p} can be removed from the mixed-integer linear optimization problem~\eqref{prob:complete_transposition} without loss of generality. Indeed, 
for each non-identity bijection $\sigma \in \widehat{\Sigma}$, let the subset of contests that have the possibility of being overvoted  under a voting machine with mapping $\sigma$ be denoted by
\begin{align*}
     \widehat{\mathcal{C}}^\sigma \triangleq \left \{c \in \mathcal{C}: \;  \sum_{c' \in \mathcal{C}} \min \left \{ \left| \left \{ \sigma(i) \in \mathcal{N}_{c'}: i \in \mathcal{N}_c \right \} \right|, v_{c'} \right \} \ge v_c + 1  \right \}. 
\end{align*}
We observe that the subset of contests $\widehat{\mathcal{C}}^\sigma$ for each $\sigma \in \widehat{\Sigma}$ can be efficiently precomputed.\footnote{By \emph{precomputed}, we mean that the set $\widehat{\mathcal{C}}^\sigma$ can be computed independent of $B$ and only needs to be computed once per $\sigma$. Hence, it suffices to compute $\widehat{\mathcal{C}}^\sigma$ when $\sigma$ is first added by the cutting plane method into the set $\widehat{\Sigma}$. Moreover, the set $\widehat{\mathcal{C}}^\sigma$ can be computed in $\mathcal{O}(C^2 + N)$ time by the following straightforward algorithm: (1) initialize an  $C \times C$-dimension array of all zeros; (2) for each   $i \in \mathcal{N}$,  increment the value in the array at position $(c,c')$ if $i \in \mathcal{N}_c$ and $\sigma(i) \in \mathcal{N}_{c'}$; (3) for each $c \in \mathcal{C}$, calculate the quantity $\sum_{c' \in \mathcal{C}} \min \left \{ \left| \left \{ \sigma(i) \in \mathcal{N}_{c'}: i \in \mathcal{N}_c \right \} \right|, v_{c'} \right \}$ by summing the minimum of the value of the array at position $(c,c')$ and $v_{c'}$ over all $c' \in \mathcal{C}$.}
 Using these subsets of contests, it follows immediately from  Lemma~\ref{lem:decrease_p} that constraints~\eqref{prob:complete_transposition:p} and \eqref{prob:complete_transposition:defeated} can without loss of generality be replaced by the following constraints:
\begin{subequations}
    \begin{align}
        &p^\sigma_{b,c} \geq 1 -  \frac{1}{v_c + 1} \sum_{i \in \mathcal{N}_c} \beta_{b, \sigma(i)}
        && \forall \sigma \in \widehat{\Sigma}, b \in \mathcal{B}, c \in \widehat{\mathcal{C}} ^
        \sigma \label{prob:reduced_variables:p}\\
&\sum_{i \in \mathcal{N} } \left(1 -  y_{i,\sigma(i)} \right)  + \sum_{b \in \mathcal{B}} \sum_{c \in \widehat{\mathcal{C}}^\sigma} \left(1 - p^\sigma_{b,c} \right) \ge 1 && \forall \sigma \in \widehat{\Sigma}. \label{prob:reduced_variables:defeated}
    \end{align}
In particular, we observe that the binary decision variable $p^\sigma_{b,c} \in \{0,1\}$  for each $\sigma \in \widehat{\Sigma}$, $b \in \mathcal{B}$, and $c \in \mathcal{C}$ that satisfies  $c \notin \widehat{\mathcal{C}}^\sigma$ no longer appears in the  mixed-integer linear optimization problem~\eqref{prob:complete_transposition} and can thus be eliminated.

Next, we show that a number of the continuous decision variables $y_{i,j}$ and constraints~\eqref{prob:complete_transposition:y} can be removed from the mixed-integer linear optimization problem~\eqref{prob:complete_transposition} without loss of generality. Indeed, we recall from the discussion in \S\ref{appx:mip:robust-subset} that each decision variable $y_{i,j} \in \R_{\ge 0}$ will at optimality be equal to zero only if candidates $i, j \in \mathcal{N}$ do not appear in the same number of ballots. 
Moreover, we observe that variable $y_{i,j}$ is only referenced in the constraint~\eqref{prob:complete_transposition:defeated} by the terms $y_{i,\sigma(i)}$ for each $i \in \mathcal{N}$ and $\sigma \in \widehat{\Sigma}$. Therefore, we observe that the decision variable $y_{i,j}$ only needs to be defined for the pairs of candidates $(i,j)$ in the set 
\begin{align*}
\mathcal{P}(\widehat{\Sigma} ) \triangleq \left \{(i,j) \in \mathcal{N}^2:  i \neq j \textnormal{ and there exists } \sigma \in \widehat{\Sigma} \textnormal{ that satisfies  } \sigma(i) = j \right \},
\end{align*}
and we can replace the constraint~\eqref{prob:complete_transposition:y} with 
\begin{align}
y_{i,j} \ge -1 + \gamma_{i,g} + \gamma_{j,g}  && \forall (i,j) \in \mathcal{P}(\widehat{\Sigma}), g \in \mathcal{B}_0. \label{prob:reduced_variables:y}
\end{align}
\end{subequations}

We conclude that the mixed-integer linear optimization problem~\eqref{prob:complete_transposition} can be reduced to an optimization problem with the following number of binary decision variables, number of continuous decision variables, and number of constraints: 
\begin{align*}
   &\begin{aligned} \text{\# binary decision variables} =& \underbrace{\left( | \mathcal{B}| \times | \mathcal{N}| \right)}_{ \beta}  + \underbrace{\left( | \mathcal{B}_0| \times | \mathcal{N}| \right)}_{\gamma} + \underbrace{\left( | \mathcal{B}|  \times \sum_{\sigma \in \widehat{\Sigma}} \left| \widehat{\mathcal{C}}^\sigma \right| \right)}_{p}\\
    =& BN + (B+1)N  +B \sum_{\sigma \in \widehat{\Sigma}} \left| \widehat{\mathcal{C}}^\sigma \right|\\
    =& \mathcal{O} \left( BN  +B \sum_{\sigma \in \widehat{\Sigma}} \left| \widehat{\mathcal{C}}^\sigma \right|   \right);
    \end{aligned}\\
    \end{align*}
    \begin{align*}
&\begin{aligned}\text{\# continuous decision variables} =& \underbrace{| \mathcal{P}(\widehat{\Sigma})|}_{y};\end{aligned}\\
&\begin{aligned}\text{\# constraints} =&  \underbrace{\left( | \mathcal{B}| \times| \mathcal{C}| \right)}_{\eqref{prob:complete_transposition:beta_feas}} + \underbrace{|\mathcal{N}|}_{\eqref{prob:complete_transposition:gamma_1}} + \underbrace{|\mathcal{N}|}_{\eqref{prob:complete_transposition:gamma_2}}+  \underbrace{\left( | \mathcal{B}| \times \sum_{\sigma \in \widehat{\Sigma}} \left| \widehat{\mathcal{C}}^\sigma \right|  \right)}_{\eqref{prob:reduced_variables:p}} + \underbrace{| \widehat{\Sigma}|}_{\eqref{prob:reduced_variables:defeated}} +  \underbrace{\left( |\mathcal{P}(\widehat{\Sigma})| 
 \times |\mathcal{B}_0| \right) }_{\eqref{prob:reduced_variables:y}} \\
 =& BN + BC + 2N + B\sum_{\sigma \in \widehat{\Sigma}} \left| \widehat{\mathcal{C}}^\sigma \right|  + | \widehat{\Sigma}|+ |\mathcal{P}(\widehat{\Sigma})|  (B+1)  \\
 =& \mathcal{O} \left(B  |\mathcal{P}(\widehat{\Sigma})|   + B\sum_{\sigma \in \widehat{\Sigma}} \left| \widehat{\mathcal{C}}^\sigma \right| \right). 
 \end{aligned}
\end{align*}
As we show in Appendix~\ref{appx:additional_experiments}, the above reductions in the number of decision variables and constraints lead to a significant decrease in the computation time for solving the mixed-integer linear optimization problem~\eqref{prob:complete_transposition} in each iteration of the cutting plane method. 

\subsection{Improvement 2: Distinct Votes for Candidates in the Same Contest} \label{appx:improvements:symmetry}
As our second  step in increasing the practical   efficiency of the  cutting plane method, we add a set of extra constraints into the  mixed-integer linear optimization problem~\eqref{prob:complete_transposition}. 
These extra constraints force the mixed-integer linear optimization problem~\eqref{prob:complete_transposition} to output a test deck that  proactively satisfies many of the constraints $\sigma \in \Sigma$ from the optimization problem~\eqref{prob:robust} that were not explicitly included in subset $\widehat{\Sigma}$.  
As we demonstrate through numerical experiments in Appendix~\ref{appx:additional_experiments}, the addition of this set of extra constraints leads to a significant decrease in the number of iterations of the cutting plane method without any meaningful increase in the computation time for solving the mixed-integer linear optimization problem~\eqref{prob:complete_transposition} in each iteration.

The motivation for the set of extra constraints is given by the following Lemma~\ref{lem:symmetry_withincontest} and  Proposition~\ref{prop:symmetry_withincontest}. In Lemma~\ref{lem:symmetry_withincontest}, we establish a structural property that is satisfied by every test deck that satisfies the constraints of the optimization problem~\eqref{prob:robust}. Specifically, the following lemma shows that a test deck satisfies the constraints of the optimization problem~\eqref{prob:robust} only if all of the candidates that appear in the same contest receive a different number of votes, where we say that two candidates $i,j  \in \mathcal{N}$ appear in the same contest if there exists a contest $c \in \mathcal{C}$ that satisfies  $i,j \in \mathcal{N}_c$.\looseness=-1
\begin{lemma} \label{lem:symmetry_withincontest}
If $(B,\beta_1,\ldots,\beta_B)$ is a feasible solution for the optimization problem~\eqref{prob:robust}, then   $| \{b \in \{1,\ldots,B\}: i \in \beta_b \}| \neq | \{b \in \{1,\ldots,B\}: j \in \beta_b \}|$  for all candidates $i < j$ that appear in the same contest.
\end{lemma}
\noindent Equipped with the above intermediary lemma, we  show in the following  Proposition~\ref{prop:symmetry_withincontest} that there always exists an optimal solution for the optimization problem~\eqref{prob:robust} in which all of the candidates that appear in the same contest have a strictly increasing number of votes. 
\begin{proposition} \label{prop:symmetry_withincontest}
There exists an optimal solution for the optimization problem~\eqref{prob:robust} that satisfies  $| \{b \in \{1,\ldots,B\}: i \in \beta_b \}| < | \{b \in \{1,\ldots,B\}: j \in \beta_b \}|$   for all candidates $i < j$ that appear in the same contest. 
\end{proposition}
\noindent Hence, the above proposition implies that we can, without loss of generality, restrict the solution spaces of the optimization problems~\eqref{prob:robust} and \eqref{prob:robust-subset} by adding extra constraints which enforce that  $| \{b \in \{1,\ldots,B\}: i \in \beta_b \}| < | \{b \in \{1,\ldots,B\}: j \in \beta_b \}|$   for all candidates $i < j$ that appear in the same contest.

Motivated by the structure of optimal solutions for the optimization problem~\eqref{prob:robust} that is established by Proposition~\ref{prop:symmetry_withincontest}, we now describe the set of extra constraints that we add into the mixed-integer linear optimization problem~\eqref{prob:complete_transposition}. The purpose of this set of extra constraints is to ensure that every feasible solution $(\beta,\gamma,y,p)$  of the mixed-integer linear optimization problem~\eqref{prob:complete_transposition} satisfies the inequality \looseness=-1
\begin{align*}
| \{b \in \{1,\ldots,B\}: \beta_{b,i} = 1  \}| < | \{b \in \{1,\ldots,B\}: \beta_{b,j} = 1 \}|  
\end{align*}
for all candidates $i < j$ that appear in the same contest. 
We accomplish this by adding the following set of extra constraints~\eqref{prob:complete_transposition:extra_withincontest} into the   mixed-integer linear optimization problem~\eqref{prob:complete_transposition}. In the following extra constraints, we use the shorthand notation  $\mathcal{N}_c^k$ to denote the candidate with the $k$th smallest index among all candidates in the $c$th contest (with the indexing of candidates that appear in the same contest starting at index 1). 
\begin{align}
    \sum_{b \in \mathcal{B}} \left( \beta_{b,\mathcal{N}^{k+1}_c} - \beta_{b,\mathcal{N}^{k}_c} \right) \geq 1
        \quad \forall c \in \mathcal{C} \text{ and } k \in \{1,\ldots,|\mathcal{N}_c|-1\}.\label{prob:complete_transposition:extra_withincontest}
\end{align}
Indeed, we observe that the set of  extra constraints~\eqref{prob:complete_transposition:extra_withincontest} ensure that the number of votes for each candidate $i$ is strictly less than the number of votes for candidate $j > i$ whenever candidates $i$ and $j$ appear in the same contest. 

\subsection{Improvement 3: Distinct Votes for Candidates in Similar Contests} \label{appx:improvements:n_choose_k}
As our third  step in increasing the practical   efficiency of the  cutting plane method, we add a second set of extra constraints into the  mixed-integer linear optimization problem~\eqref{prob:complete_transposition}. Similarly as \S\ref{appx:improvements:symmetry}, the set of extra constraints from the present \S\ref{appx:improvements:n_choose_k} force the mixed-integer linear optimization problem~\eqref{prob:complete_transposition} to output a test deck that  proactively satisfies many of the constraints $\sigma \in \Sigma$ from the optimization problem~\eqref{prob:robust} that were not explicitly included in subset $\widehat{\Sigma}$.  
As we demonstrate through numerical experiments in Appendix~\ref{appx:additional_experiments}, the addition of this second set of extra constraints  leads to a significant decrease in the number of iterations of the cutting plane method without any meaningful increase in the computation time for solving the mixed-integer linear optimization problem~\eqref{prob:complete_transposition} in each iteration. 

To describe our second set of extra constraints,  we require some additional terminology. We begin with the following Definition~\ref{defn:equivalence}, which provides a way of referring to  contests that are similar to one another.  
\begin{definition}[Equivalence of contests] \label{defn:equivalence} 
    We say that  two contests $c,c' \in \mathcal{C}$ are equivalent, denoted by $c \equiv c'$, if and only if they satisfy  $| \mathcal{N}_c| = | \mathcal{N}_{c'}|$ and  $v_c = v_{c'}$.
\end{definition}
\noindent  In other words, we say that two contests are equivalent if and only if the contests have the same number of candidates and the same maximum number of votes. Next, recall from \S\ref{appx:improvements:symmetry} that $\mathcal{N}^k_c$ refers to the candidate with the $k$th smallest index among all candidates in contest $c$. Equipped with this notation, the second additional terminology, which is denoted below by Definition~\ref{defn:lexicographic}, provides a way to compare the votes received by candidates in two equivalent contests. 
\begin{definition}[Lexicographic ordering of contests] \label{defn:lexicographic}
We say that two contests $c,c' \in \mathcal{C}$ are lexicographically ordered with respect to a test deck $\beta_1,\ldots,\beta_B \in \mathscr{B}$, denoted by $c \prec_{\beta_1 \cdots \beta_B} c'$, if and only if $c \equiv c'$ and there exists  $k \in \{1,\ldots,| \mathcal{N}_c| \}$ that satisfies  
\begin{gather*}
\left| \left \{ b \in \{1,\ldots,B\}: \mathcal{N}^{| \mathcal{N}_{c}|}_c \in \beta_b \right \} \right| = \left| \left \{ b \in \{1,\ldots,B\}: \mathcal{N}^{| \mathcal{N}_{c}|}_{c'} \in \beta_b \right \} \right|  \\
\vdots \\
\left| \left \{ b \in \{1,\ldots,B\}: \mathcal{N}^{k+1}_c \in \beta_b \right \} \right| = \left| \left \{ b \in \{1,\ldots,B\}: \mathcal{N}^{k+1}_{c'} \in \beta_b \right \} \right|  \\
\left| \left \{ b \in \{1,\ldots,B\}: \mathcal{N}^{k}_c \in \beta_b \right \} \right| < \left| \left \{ b \in \{1,\ldots,B\}: \mathcal{N}^{k}_{c'} \in \beta_b \right \} \right|.
\end{gather*}
\end{definition}
\noindent  In other words, we say that two contests are lexicographically ordered with respect to a test deck if and only if the contests are equivalent and the number of votes received by candidates in the first contest is lexicographically less than the number of votes received by candidates in the second contest, beginning with the highest-indexed candidates in the contests.

In view of the additional terminology given by Definitions~\ref{defn:equivalence} and \ref{defn:lexicographic}, we now describe the second set of extra constraints that we add into the mixed-integer linear optimization problem~\eqref{prob:complete_transposition}.  The motivation for this second set of extra constraints is given by the following Lemma~\ref{lem:symmetry_acrosscontests} and Proposition~\ref{prop:symmetry_acrosscontests}. In Lemma~\ref{lem:symmetry_acrosscontests}, we establish a structural property that is satisfied by every test deck that satisfies the constraints of the optimization problem~\eqref{prob:robust}. Specifically, the following lemma shows that a test deck satisfies the constraints of the optimization problem~\eqref{prob:robust} only if all equivalent contests are lexicographically distinct.
    \begin{lemma} \label{lem:symmetry_acrosscontests}
If $(B,\beta_1,\ldots,\beta_B)$ is a feasible solution for the optimization problem~\eqref{prob:robust}, then  $c \prec_{\beta_1 \cdots \beta_B} c'$ or $c' \prec_{\beta_1 \cdots \beta_B} c$ for all  contests $c < c'$ that satisfy $c \equiv c'$. 
\end{lemma}
\noindent Equipped with the above intermediary lemma, we  show in the following  Proposition~\ref{prop:symmetry_acrosscontests} that there always exists an optimal solution for the optimization problem~\eqref{prob:robust} in  which all of the equivalent contests are lexicographically ordered according to the indices of the contests. 
    \begin{proposition} \label{prop:symmetry_acrosscontests}
There exists an optimal solution for the optimization problem~\eqref{prob:robust} that satisfies $c \prec_{\beta_1 \cdots \beta_B} c'$ for all  contests $c < c'$ that satisfy $c \equiv c'$. 
\end{proposition}
\noindent Hence, the above proposition implies that we can, without loss of generality, restrict the solution spaces of the optimization problems~\eqref{prob:robust} and \eqref{prob:robust-subset} by adding extra constraints which enforce that  $c \prec_{\beta_1 \cdots \beta_B} c'$   for all contests $c < c'$ that satisfy $c \equiv c'$.

Motivated by the structure of optimal solutions for the optimization problem~\eqref{prob:robust} that is established by Proposition~\ref{prop:symmetry_acrosscontests}, we now describe the second set of  extra constraints that we  add into the mixed-integer linear optimization problem~\eqref{prob:complete_transposition}. 
The purpose of this second set of extra constraints is to ensure that any feasible solution $(\beta,\gamma,y,p)$  of the mixed-integer linear optimization problem~\eqref{prob:complete_transposition} satisfies the property that for each pair of candidates $c < c'$ that satisfy $c \equiv c'$, there exists a $k \in \{1,\ldots,| \mathcal{N}_c|\}$ that satisfies
\begin{gather*}
\left| \left \{ b \in \{1,\ldots,B\}: \beta_{b,\mathcal{N}^{| \mathcal{N}_{c}|}_c} = 1 \right \} \right| = \left| \left \{ b \in \{1,\ldots,B\}: \beta_{b,\mathcal{N}^{| \mathcal{N}_{c}|}_{c'}} = 1  \right \} \right|  \\
\vdots \\
\left| \left \{ b \in \{1,\ldots,B\}: \beta_{b,\mathcal{N}^{k+1}_c} = 1 \right \} \right| = \left| \left \{ b \in \{1,\ldots,B\}: \beta_{b,\mathcal{N}^{k+1}_{c'}} = 1 \right \} \right|  \\
\left| \left \{ b \in \{1,\ldots,B\}: \beta_{b,\mathcal{N}^{k}_c} = 1  \right \} \right| < \left| \left \{ b \in \{1,\ldots,B\}: \beta_{b,\mathcal{N}^{k}_{c'}} = 1  \right \} \right|.
\end{gather*}
We accomplish this by adding the following set of extra constraints~\eqref{line:lexicographic:sequential_last}-\eqref{line:lexicographic:sequential_first} into the mixed-integer linear optimization problem~\eqref{prob:complete_transposition}.  In the following extra constraints, we use the shorthand notation  $\mathcal{I}$ to denote to the set of sequential equivalent contests,
$$\mathcal{I} \triangleq  \left\{(c, c') \in \mathcal{C} \times \mathcal{C} : \quad \begin{aligned}
 &c < c', \; c \equiv c', \textnormal{ and there does not exist a} \\
 &\textnormal{contest } \bar{c}  \in \left\{ c + 1,\ldots,c' - 1 \right \} \textnormal{ that satisfies } c \equiv \bar{c}
 \end{aligned} \right\},$$
  and we use extra binary decision variables  $\lambda_{c,c'}^k \in \{0,1\}$ which are added to the mixed-integer linear optimization problem~\eqref{prob:complete_transposition} for each $(c,c') \in \mathcal{I}$ and  $k \in \{1,\ldots,| \mathcal{N}_c| \}$.
  \begin{subequations}
 \begin{align}
   \lambda_{c,c'}^{| \mathcal{N}_c|} &\le \sum_{b \in \mathcal{B}} \beta_{b,\mathcal{N}^{| \mathcal{N}_c|}_{c'}} -  \sum_{b \in \mathcal{B}} \beta_{b,\mathcal{N}^{| \mathcal{N}_c|}_{c}}  && \forall (c,c') \in \mathcal{I} \label{line:lexicographic:sequential_last}\\
 \lambda_{c,c'}^k &\le B \lambda_{c,c'}^{k+1} + \sum_{b \in \mathcal{B}} \beta_{b,\mathcal{N}^{k}_{c'}} -  \sum_{b \in \mathcal{B}} \beta_{b,\mathcal{N}^{k}_{c}}  && \forall (c,c') \in \mathcal{I}, \; k \in\left \{1,\ldots, | \mathcal{N}_c|-1 \right \}  \label{line:lexicographic:sequential}\\
    \lambda_{c,c'}^{1} &=1 &&  \forall (c,c') \in \mathcal{I}  \label{line:lexicographic:sequential_first}.
 \end{align}
 \end{subequations}
 Indeed, constraints~\eqref{line:lexicographic:sequential_last} and \eqref{line:lexicographic:sequential} ensure for each $(c,c') \in \mathcal{I}$ that the equality $\lambda^k_{c,c'} = 1$  can be satisfied if and only if there exists $k' \in \{k,\ldots,| \mathcal{N}_c| \}$ that satisfies the equality $ \sum_{b \in \mathcal{B}} \beta_{b,\mathcal{N}^{k''}_{c}} = \sum_{b \in \mathcal{B}} 
 \beta_{b,\mathcal{N}^{k''}_{c'}}$  for all $k'' \in \{k'+1,\ldots,| \mathcal{N}_c| \}$ as well as satisfies the strict inequality $ \sum_{b \in \mathcal{B}} \beta_{b,\mathcal{N}^{k'}_{c}} < \sum_{b \in \mathcal{B}} 
 \beta_{b,\mathcal{N}^{k'}_{c'}}$. Hence, constraint~\eqref{line:lexicographic:sequential_first} ensures that there exists $k' \in \{1,\ldots,| \mathcal{N}_c| \}$ that satisfies the equality $ \sum_{b \in \mathcal{B}} \beta_{b,\mathcal{N}^{k''}_{c}} = \sum_{b \in \mathcal{B}} 
 \beta_{b,\mathcal{N}^{k''}_{c'}}$  for all $k'' \in \{k'+1,\ldots,| \mathcal{N}_c| \}$ as well as satisfies the strict inequality $ \sum_{b \in \mathcal{B}} \beta_{b,\mathcal{N}^{k'}_{c}} < \sum_{b \in \mathcal{B}} 
 \beta_{b,\mathcal{N}^{k'}_{c'}}$.

\subsection{Improvement 4: Heuristic for Finding Good Cuts} \label{appx:improvements:swap}
As our fourth  step in increasing the practical efficiency of the cutting plane method, we propose a modification to the objective function of the  mixed-integer linear optimization problem~\eqref{prob:oracle_mip} from \S\ref{appx:mip:oracle}. The purpose of this modification is to guide the mixed-integer optimization problem~\eqref{prob:oracle_mip} to choosing non-identity bijections in each iteration that eliminate large numbers of feasible test decks from the mixed-integer linear optimization problem~\eqref{prob:complete_transposition}. We demonstrate through numerical experiments in Appendix~\ref{appx:additional_experiments} that the proposed modification to the objective function of the mixed-integer linear optimization problem~\eqref{prob:oracle_mip}  can significantly decrease the number of iterations of the cutting plane method. 

To motivate our subsequent developments in \S\ref{appx:improvements:swap}, we begin by presenting a framework for analyzing the quality of the non-identity bijections $\sigma \in \Sigma$ that are added  to the set $\widehat{\Sigma}$  at the end of each iteration of the cutting plane method. Indeed, we recall that the practical efficiency of the cutting plane method from \S\ref{appx:cutting} depends on  the number of iterations until an optimal test deck for the optimization problem~\eqref{prob:robust} is obtained. The number of iterations of the cutting plane method, in turn, depends on whether the optimization problem~\eqref{prob:oracle} in each iteration of the cutting plane method yields a non-identity bijection $\sigma \in \Sigma$  that eliminates a large number of feasible test decks from the optimization problem~\eqref{prob:robust-subset}. 
Letting  $\mathscr{F}(\widehat{\Sigma})$ denote the set of test decks that are feasible for the optimization problem~\eqref{prob:robust-subset}, we will henceforth say (informally)  that the optimization problem~\eqref{prob:oracle} yields a high-quality non-identity bijection $\sigma \in \Sigma$ if  the set of feasible test decks in the next iteration $\mathscr{F}(\widehat{\Sigma} \cup \{\sigma \})$ is much smaller than the set of feasible test decks in the current iteration $\mathscr{F}(\widehat{\Sigma})$.

A priori, it might appear difficult to determine whether a non-identity bijection $\sigma \in \Sigma$  will eliminate a large number of feasible test decks from the optimization problem~\eqref{prob:robust-subset}. Nonetheless, we  demonstrate below that high-quality non-identity bijections $\sigma \in \Sigma$ always have a  structural property that we refer to as  \emph{minimal}. To define this structural property, consider any given non-identity bijection $\sigma \in \Sigma$, and let $\mathscr{G}^\sigma \equiv (\mathscr{V}^\sigma,\mathscr{E}^\sigma)$ denote the undirected graph that is generated by that non-identity bijection. The set of vertices of this undirected graph is defined as the subset of contests that include a candidate that is swapped by the non-identity bijection $\sigma$, 
\begin{align*}
\mathscr{V}^\sigma \triangleq \left \{c \in \mathcal{C}: \textnormal{ there exists } i \in \mathcal{N}_c \textnormal{ that satisfies } \sigma(i) \neq i  \right \},
\end{align*}
and the set of edges of this undirected graph is defined as  the pairs of contests containing candidates that are swapped by the non-identity bijection $\sigma$, 
\begin{align*}
\mathscr{E}^\sigma \triangleq \left \{(c,c') \in \mathcal{C} \times \mathcal{C}:  \begin{gathered} \textnormal{ there exist } i \in \mathcal{N}_c \textnormal{ and } i' \in \mathcal{N}_{c'} \textnormal{ that satisfy the} \\ 
\textnormal{equality } \sigma(i) = i' \textnormal{ or satisfy the equality } \sigma(i') = i \end{gathered}  \right \}.
\end{align*}
We recall from graph theory that the vertices of an undirected graph can always be partitioned into a unique collection of connected components, and we henceforth let   $K^\sigma$ denote the number of connected components and let $\mathscr{K}^\sigma_1,\ldots,\mathscr{K}^\sigma_{K^\sigma} \subseteq \mathscr{V}^\sigma$ denote the connected components of the undirected graph  $\mathscr{G}^\sigma \equiv (\mathscr{V}^\sigma,\mathscr{E}^\sigma)$.\footnote{It is a straightforward exercise to show that $\mathscr{K}^\sigma_1,\ldots,\mathscr{K}^\sigma_{K^\sigma} \subseteq \mathscr{V}^\sigma$ are the connected components of the undirected graph $\mathscr{G}^\sigma \equiv (\mathscr{V}^\sigma,\mathscr{E}^\sigma)$ if and only if (1) $\mathscr{K}^\sigma_1,\ldots,\mathscr{K}^\sigma_{K^\sigma}$ are disjoint, (2) the union of $\mathscr{K}^\sigma_1,\ldots,\mathscr{K}^\sigma_{K^\sigma}$ is equal to $\mathscr{V}^\sigma$, and (3) $(c,c') \in \mathscr{E}^\sigma$ implies that there exists   $k \in \{1,\ldots,|K^\sigma| \}$ that satisfies $c,c' \in \mathscr{K}^\sigma_k$.   }  Equipped with this terminology, we are ready to define the structural property of non-identity bijections that will form the basis of our subsequent discussions:%
\begin{definition}[Minimal] \label{defn:minimal} 
We say that a non-identity bijection $\sigma \in \Sigma$ is minimal if and only if the number of connected components of  $\mathscr{G}^\sigma \equiv (\mathscr{V}^\sigma,\mathscr{E}^\sigma)$ satisfies $K^\sigma = 1$. 
\end{definition}

Figure~\ref{fig:minimal} provides an illustration of Definition~\ref{defn:minimal} by showing an example of a non-identity bijection that is not minimal. Specifically, Figure~\ref{fig:nonminimal} presents a non-identity bijection $\sigma$ in a ballot style with five contests. Figure~\ref{fig:nonminimal_graph} shows the undirected graph $\mathscr{G}^\sigma \equiv (\mathscr{V}^\sigma,\mathscr{E}^\sigma)$ corresponding to the non-identity bijection $\sigma$. The undirected graph in Figure~\ref{fig:nonminimal_graph} has two connected components, which implies that the non-identity bijection from Figure~\ref{fig:nonminimal}  is not minimal.  

 \begin{figure}[t]
\centering 
\subfloat[]{%
\includegraphics[width=0.4\linewidth]{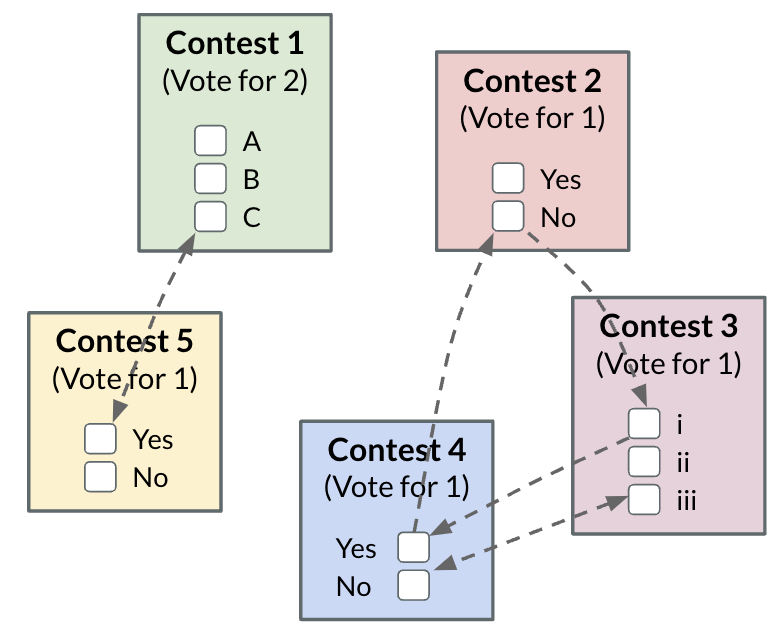}
\label{fig:nonminimal}
}\hspace{0.5in}
\subfloat[]{%
\includegraphics[width=0.35\linewidth]{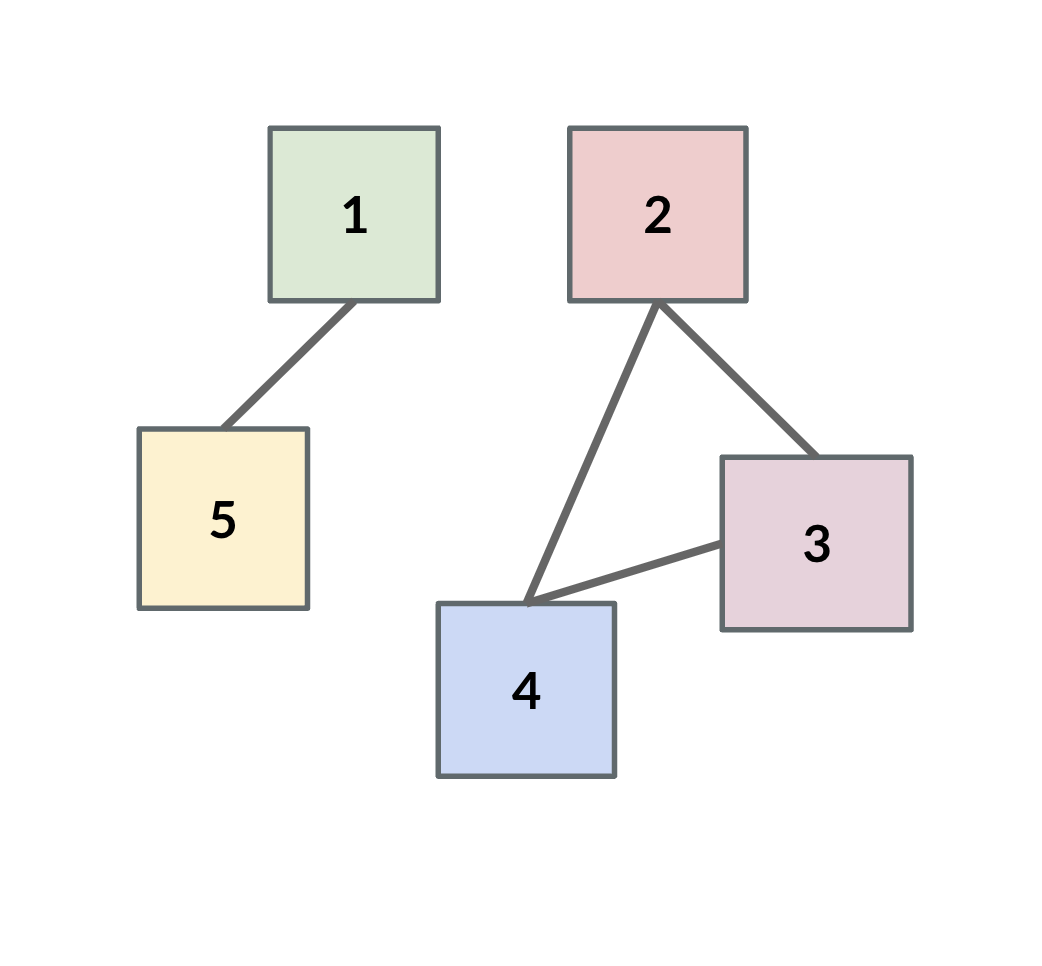}
\label{fig:nonminimal_graph}
}

\caption{Example of a non-identity bijection that is not minimal in a ballot style with five contests. }
 \label{fig:minimal}
\end{figure}
Our main result of \S\ref{appx:improvements:swap}, which is presented below as Theorem~\ref{thm:minimal}, establishes the significance of   non-identity bijections that are minimal. In particular, the following theorem shows that there always exists a feasible solution for the optimization problem~\eqref{prob:oracle} that is minimal. More importantly, the following Theorem~\ref{thm:minimal} shows that minimal non-identity bijections are always preferred to non-minimal non-identity bijections from the perspective of  eliminating the greatest number of feasible test decks from the optimization problem~\eqref{prob:robust-subset}.
\begin{theorem} \label{thm:minimal}
Let $\sigma \in \Sigma$ denote a feasible solution for the optimization problem~\eqref{prob:oracle}. For each $k \in \{1,\ldots,K^\sigma \}$, let $\sigma_k: \mathcal{N} \to \mathcal{N}$ be defined for each $c \in \mathcal{C}$ and $i \in \mathcal{N}_c$ by
\begin{align*}
\sigma_k(i) &\triangleq \begin{cases}
    \sigma(i),&\textnormal{if } c \in \mathscr{K}^\sigma_k,\\
    i,&\textnormal{if } c \notin \mathscr{K}^\sigma_k. 
\end{cases}
\end{align*}
Then $\sigma_1,\ldots,\sigma_{K^\sigma}$ are feasible solutions for  the optimization problem~\eqref{prob:oracle} and 
\begin{align}
    \bigcup_{k =1}^{K^\sigma} \mathscr{F} \left (\widehat{\Sigma} \cup \{\sigma_k \} \right) =  \mathscr{F} \left (\widehat{\Sigma} \cup \{\sigma \} \right). \label{line:sigma_breakdown}
\end{align}
\end{theorem}

 To appreciate the significance of Theorem~\ref{thm:minimal}, let us make several observations. First, we observe that each non-identity bijection $\sigma_k \in \Sigma$ can be interpreted as a restriction of the non-identity bijection $\sigma \in \Sigma$ that only affects the candidates from contests in the connected component $\mathscr{K}^\sigma_k$. Therefore, it follows that the number of connected components in the undirected graph $\mathscr{G}^{\sigma_k} \equiv (\mathscr{V}^{\sigma_k},\mathscr{E}^{\sigma_k})$ corresponding to $\sigma_k$ is equal to one, which implies that each of the non-identity bijections $\sigma_1,\ldots,\sigma_{K^\sigma}  \in \Sigma$ is minimal. Second, we observe from line~\eqref{line:sigma_breakdown} that the inclusion $ \mathscr{F}  (\widehat{\Sigma} \cup \{\sigma_k \} ) \subseteq  \mathscr{F}  (\widehat{\Sigma} \cup \{\sigma \} )$  holds for each of the connected components $k \in \{1,\ldots,K^\sigma \}$. Hence, Theorem~\ref{thm:minimal} implies that each of the non-identity bijections $\sigma_1,\ldots,\sigma_{K^\sigma}$ is  preferred to $\sigma$ from the perspective of eliminating feasible test decks from the optimization problem~\eqref{prob:robust-subset}.

Thus motivated, we now turn to  the algorithmic question of how to find a minimal non-identity bijection that is feasible for the optimization problem~\eqref{prob:oracle}. In the following Theorem~\ref{thm:minimal_reform}, we show that such a minimal non-identity bijection can be found though making a simple modification to the objective function of the mixed-integer linear optimization problem~\eqref{prob:oracle_mip}.
\begin{theorem} \label{thm:minimal_reform}
Consider the following mixed-integer linear optimization problem:

    \begin{equation} \label{prob:oracle_mip_modified}
    \begin{aligned}
    \underset{x \in \{0,1\}^{\mathcal{N} \times \mathcal{N}}}{\textnormal{minimize}} \quad\quad & \sum_{i,j \in \mathcal{N}: i \neq j} x_{i,j}  \\
    \textnormal{subject to} \quad\quad  & \eqref{prob:oracle_mip:bijective_1}, \eqref{prob:oracle_mip:bijective_2}, \eqref{prob:oracle_mip:dishonest}, \eqref{prob:oracle_mip:overvote}, \eqref{prob:oracle_mip:vote_total}.
    \end{aligned}
    \end{equation}
Let $x \in \{0,1\}^{\mathcal{N} \times \mathcal{N}}$ be an optimal solution of the mixed-integer linear optimization problem~\eqref{prob:oracle_mip_modified}, and let  $\sigma: \mathcal{N} \to \mathcal{N}$ be the function that satisfies  the equality $\sigma(i) = j$ if and only if $x_{i,j} = 1$ for all $i,j \in \mathcal{N}$. Then $\sigma$ is a minimal non-identity bijection that is feasible for the optimization problem~\eqref{prob:oracle}. 
\end{theorem}

\noindent We observe that the mixed-integer linear optimization problems~\eqref{prob:oracle_mip} and \eqref{prob:oracle_mip_modified} have the same decision variables and constraints. Hence, Theorem~\ref{thm:minimal_reform}  shows that obtaining a minimal non-identity bijection that is feasible for the optimization problem~\eqref{prob:oracle} can be achieved by simply modifying the objective function of the mixed-integer linear optimization problem~\eqref{prob:oracle_mip}.

\subsection{Improvement 5: Combining Noncompetitive Contests} \label{appx:improvements:noncompetitive}
As our fifth  step in increasing the practical efficiency of the cutting plane method, we show that noncompetitive contests can be combined into one without loss of generality.  By combining these contests,  we demonstrate through numerical experiments in Appendix~\ref{appx:additional_experiments} that the number of iterations of the cutting plane method can be significantly decreased. 

We begin by introducing the terminology and notation that will be used throughout \S\ref{appx:improvements:noncompetitive}. Indeed, let the original ballot style be denoted by the tuple $({\mathcal{N}}, {\mathcal{C}}, \{ {\mathcal{N}}_c \}_{c \in {\mathcal{C}}}, \{{v}_c\}_{c \in {\mathcal{C}}}). $ For the original ballot style, we recall from  \S\ref{sec:math:notation} that a contest $c \in \mathcal{C}$ is noncompetitive if and only if the number of candidates in the contest $| \mathcal{N}_c|$ is equal to the maximum number of votes $v_c$. We represent the ballot style in which all of the noncompetitive contests from the original ballot style are combined into a single contest by the tuple $({\mathcal{N}}, \widetilde{\mathcal{C}}, \{ \widetilde{\mathcal{N}}_c \}_{c \in \widetilde{\mathcal{C}}}, \{\widetilde{v}_c\}_{c \in \widetilde{\mathcal{C}}})$, where
\begin{align*}
    \widetilde{\mathcal{C}} &\triangleq \{0\} \cup  \left \{ c \in \mathcal{C}:  | \mathcal{N}_c| > v_c  \right \};\\
    \widetilde{\mathcal{N}}_c &\triangleq \begin{cases}
    \mathcal{N}_c,&\text{if } c \neq 0,\\
    \bigcup \limits_{c' \in \mathcal{C}:  | \mathcal{N}_{c'}|  = v_{c'}} \mathcal{N}_{c'},&\text{if } c = 0;
    \end{cases}\\
    \widetilde{v}_c &\triangleq 
    \begin{cases}
v_c,&\text{if } c \neq 0 ,\\
\left| \bigcup \limits_{c' \in \mathcal{C}:  | \mathcal{N}_{c'}| = v_{c'}} \mathcal{N}_{c'} \right|  ,&\text{if } c = 0.
    \end{cases}
\end{align*}
We observe in the new ballot style that  contest $0$ denotes the contest that is constructed by combining all of the noncompetitive contests from the original ballot style. Finally, for each non-identity bijection $\sigma \in \Sigma$,  let the output of a voting machine in the new ballot style whose mapping from candidates to targets is the bijection $\sigma$ be given for each candidate $i \in \widetilde{\mathcal{N}}_c$ in each contest $c \in \widetilde{\mathcal{C}}$ by
\begin{align*}
        \widetilde{T}^\sigma_i(\beta_1,\ldots,\beta_B) \triangleq \sum_{b=1}^B  \mathbb{I}\left \{   \sigma(i) \in \beta_b  \textnormal{ and } \left| \left \{ \sigma(j) \in \beta_b: j \in \widetilde{\mathcal{N}}_c \right \}  \right| \leq \widetilde{v}_c \right \}. 
\end{align*}

Equipped with the above notation, we now present the main  result of \S\ref{appx:improvements:noncompetitive}. This main result, presented below as Proposition~\ref{prop:noncompetitive}, establishes that the set of optimal solutions for the optimization problem~\eqref{prob:robust} for any given original ballot style will not change if all of the noncompetitive contests in the original ballot style are combined into a single contest. 
\begin{proposition} \label{prop:noncompetitive}
Consider any original ballot style $({\mathcal{N}}, {\mathcal{C}}, \{ {\mathcal{N}}_c \}_{c \in {\mathcal{C}}}, \{{v}_c\}_{c \in {\mathcal{C}}})$, and let the ballot style in which all of the noncompetitive contests from the original ballot style are combined into a single contest be denoted by  $({\mathcal{N}}, \widetilde{\mathcal{C}}, \{ \widetilde{\mathcal{N}}_c \}_{c \in \widetilde{\mathcal{C}}}, \{\widetilde{v}_c\}_{c \in \widetilde{\mathcal{C}}})$.  Then the following equality holds for all $B \in \N$,    $\beta_1,\ldots,\beta_B  \subseteq \mathcal{N}$, $\sigma \in \Sigma \cup \{*\}$, and $i \in \mathcal{N}$: 
\begin{align*}
T^\sigma_i(\beta_1,\ldots,\beta_B) = \widetilde{T}_i^\sigma(\beta_1,\ldots,\beta_B). 
\end{align*}
\end{proposition}

We conclude \S\ref{appx:improvements:noncompetitive} by discussing why  combining the noncompetitive contests into a single contest can decrease the number of iterations of the cutting plane method. In essence,  the value of combining the noncompetitive contests stems from the second improvement to the cutting plane method that is proposed in \S\ref{appx:improvements:symmetry}. Indeed, we recall from \S\ref{appx:improvements:symmetry} that our second improvement to the cutting plane method  consisted  of adding the following set of extra constraints into the optimization problem~\eqref{prob:robust-subset}: 
\begin{align}
\begin{aligned}
| \{b \in \{1,\ldots,B\}: i \in \beta_{b}   \}| < | \{b \in \{1,\ldots,B\}: j \in  \beta_{b}  \}| \quad \forall c \in \mathcal{C}, i,j \in \mathcal{N}_c: i < j.&
\end{aligned}
\label{line:repeat_symmetry}
\end{align} 
The constraints~\eqref{line:repeat_symmetry} ensure that the optimization problem~\eqref{prob:robust-subset} in each iteration of the cutting plane method outputs a test deck in which candidates in the same contest receive a strictly increasing number of votes. 
In view of our recollection of the second improvement from  \S\ref{appx:improvements:symmetry}, we conclude the present \S\ref{appx:improvements:noncompetitive} with an example which shows that combining the noncompetitive contests into a single contest can  decrease the number of iterations of the cutting plane method.

\begin{example} \label{example:noncompetitive}
Consider an original ballot style consisting of two contests, where the first contest is defined by the equalities $\mathcal{N}_1 = \{1\}$ and $v_1 = 1$, and the second contest is defined by the equalities $\mathcal{N}_2 = \{2,3 \}$ and $v_2 = 2$. We observe that each of the two contests is a noncompetitive contest.  

We first analyze the number of iterations of the cutting plane method in the case where the noncompetitive contests are combined into a single contest. Indeed, if the noncompetitive contests are combined into a single contest, then we observe that the new ballot style $({\mathcal{N}}, \widetilde{\mathcal{C}}, \{ \widetilde{\mathcal{N}}_c \}_{c \in \widetilde{\mathcal{C}}}, \{\widetilde{v}_c\}_{c \in \widetilde{\mathcal{C}}})$ consists of a single contest, $\widetilde{\mathcal{C}} = \{0 \}$, wherein the candidates in that contest are given by $\widetilde{\mathcal{N}}_0 = \{1,2,3\}$ and the maximum number of votes in that contest is given by  $\widetilde{v}_0 = 3$. In the first iteration of the cutting plane method, we start with $\widehat{\Sigma} = \emptyset$, in which the optimization problem~\eqref{prob:robust-subset} with the constraints~\eqref{line:repeat_symmetry} can be written as 
\begin{align} \label{prob:robust-subset:example:part1:iteration1}
\begin{aligned}
\underset{B \in \N, \; \beta_1,\ldots,\beta_B \in \mathscr{B}}{\textnormal{minimize}}\quad & B\\
\textnormal{subject to} \quad & T^*_1(\beta_1,\ldots,\beta_B) <  T^*_2(\beta_1,\ldots,\beta_B) <  T^*_3(\beta_1,\ldots,\beta_B).
\end{aligned}
\end{align}
We observe from inspection that  the optimization problem~\eqref{prob:robust-subset:example:part1:iteration1} has two optimal solutions, which are stated as follows:
\begin{align*}
    (B^1,\beta_1^1,\beta_2^1) = (2, \{2,3 \}, \{3 \}),\\
     (B^2,\beta_1^2,\beta_2^2) = (2, \{3 \}, \{2,3 \}).
\end{align*}
In particular, we observe that both of those optimal solutions are feasible solutions of the optimization problem~\eqref{prob:robust}. Hence, if the noncompetitive contests are combined into a single contest, then we observe for this example that the cutting plane method will  terminate after a single iteration.\looseness=-1

We conclude Example~\ref{example:noncompetitive} by showing that the number of iterations of the cutting plane method will always be strictly greater than one if  the noncompetitive contests are not combined into a single contest. Indeed, suppose that we apply the cutting plane method to the original ballot style $({\mathcal{N}}, {\mathcal{C}}, \{ {\mathcal{N}}_c \}_{c \in {\mathcal{C}}}, \{{v}_c\}_{c \in {\mathcal{C}}})$. In the first iteration of the cutting plane method, we start with $\widehat{\Sigma} = \emptyset$, in which the optimization problem~\eqref{prob:robust-subset} with the constraints~\eqref{line:repeat_symmetry} can be written as 
\begin{align} \label{prob:robust-subset:example:part2:iteration1}
\begin{aligned}
\underset{B \in \N, \; \beta_1,\ldots,\beta_B \in \mathscr{B}}{\textnormal{minimize}}\quad & B\\
\textnormal{subject to} \quad &  T^*_2(\beta_1,\ldots,\beta_B) <  T^*_3(\beta_1,\ldots,\beta_B).
\end{aligned}
\end{align}
We observe from inspection that  the optimization problem~\eqref{prob:robust-subset:example:part2:iteration1} has two optimal solutions, which are stated as follows:
\begin{align*}
    (B^1,\beta_1^1) &= (1, \{3 \}),\\
    (B^2,\beta_1^2) &= (1, \{1,3 \}). 
\end{align*}
However, neither of those optimal solutions are feasible solutions for the optimization problem~\eqref{prob:robust}.\footnote{The fact that neither $(B^1,\beta_1^1) = (1, \{3 \})$ nor $(B^2,\beta_1^2) = (1, \{1,3 \})$ is a feasible solution of the optimization problem~\eqref{prob:robust} follows immediately from the fact that the optimal objective value of the optimization problem~\eqref{prob:robust} is equal to two.} Hence, if the noncompetitive contests are not combined into a single contest, then we observe for this example that the cutting plane method will always require at least two iterations. \qed
\end{example}

\section{Numerical Experiments in Real-World Election} \label{sec:experiments}
In partnership with the Michigan Bureau of Elections, we applied our approach to each of the state's 6928 ballot styles from the  November 2022 general election. 
Through conversations with state leadership, local election officials, and vendors,  we found that the ease of deploying our approach depended on two main factors:  the length of the test decks with rigorous security guarantees obtained by solving our optimization problem~\eqref{prob:robust}, and the computation time required by our algorithm to find optimal test decks for all 6928 ballot styles.  We report below on the performance of our approach with respect to those two factors.

\subsection{Length of Optimal Test Decks} \label{sec:experiments:length}
The length of test decks is a crucial factor in conducting LAT in real-world elections. Long decks pose challenges, both in terms of cost and difficulty for election officials. Consequently, the practicality of our approach to achieving rigorous security guarantees in LAT depends on whether the test decks obtained by solving the optimization problem~\eqref{prob:robust} are significantly longer than the heuristic-based test decks  that would otherwise be used. 

\begin{figure}[bp]
\centering 
\subfloat[ \label{fig:michigan}]{%
\includegraphics[width=0.99\linewidth]{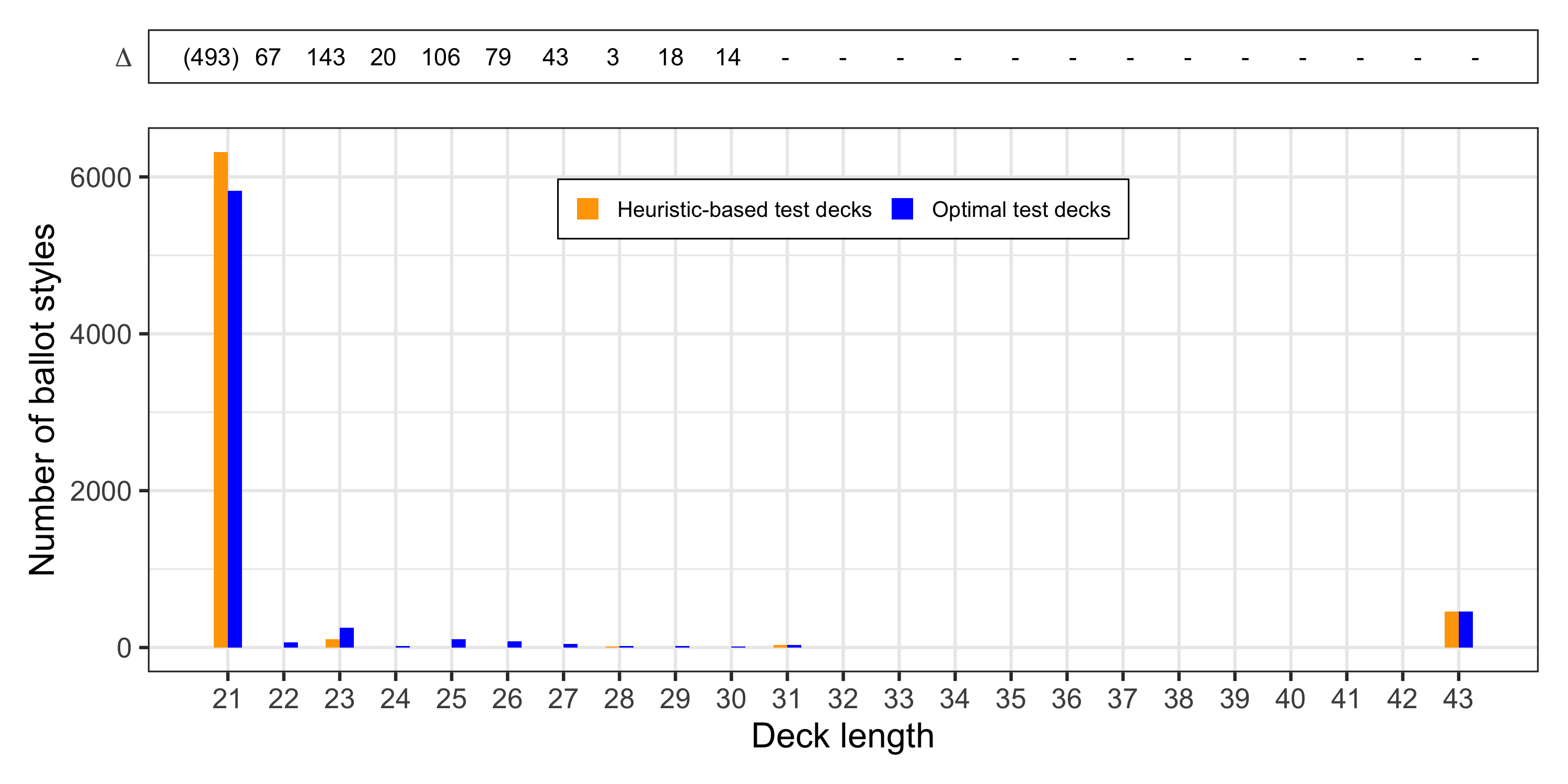}}

\subfloat[\label{fig:max}]{%
\includegraphics[width=0.7\linewidth]{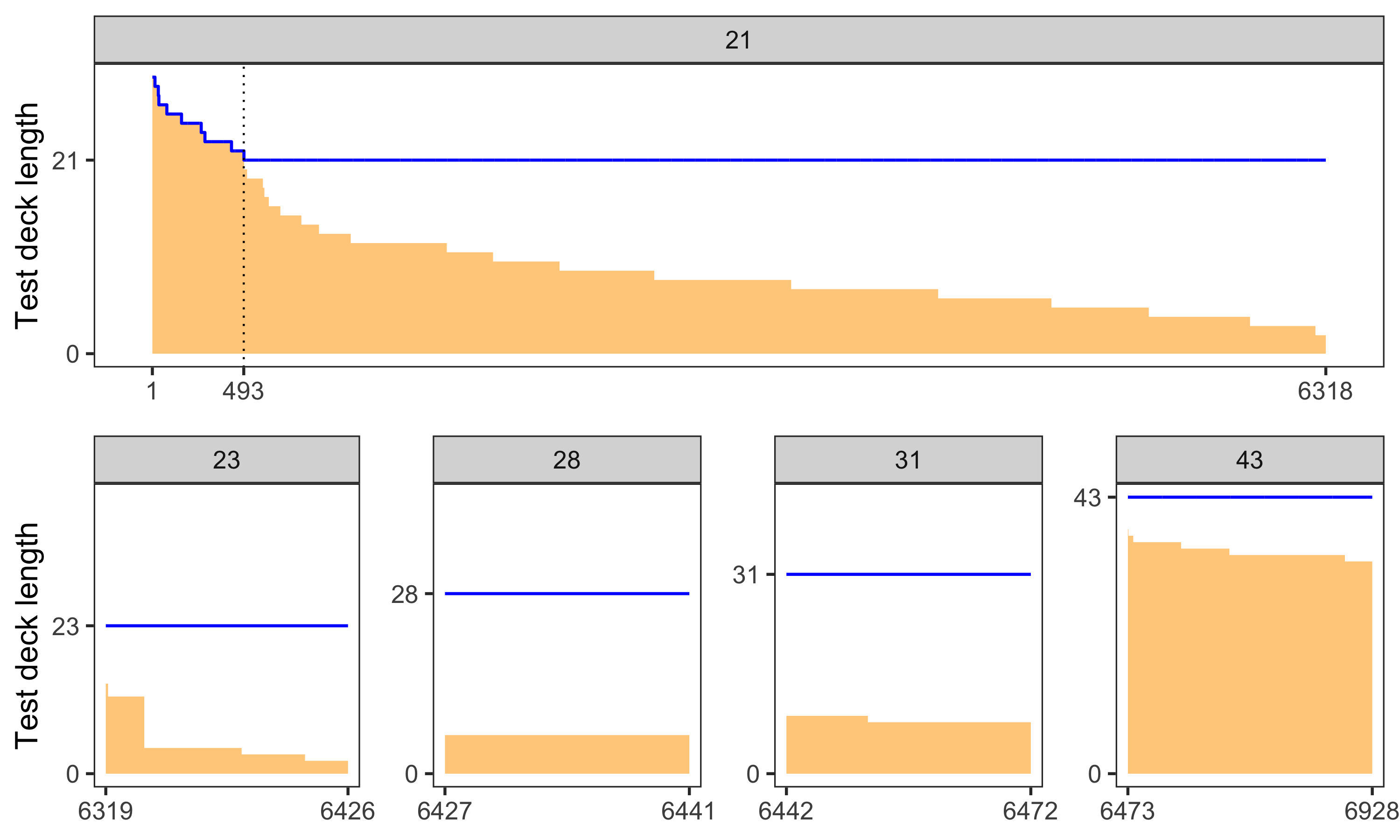} }

\caption{
\textnormal{\textbf{(a)}} Distributions of lengths of heuristic-based test decks (orange) and optimal test decks (blue) across the 6928 ballot styles. Top of figure shows the changes in number of ballot styles (blue minus orange). The similarity of the two distributions indicates that our approach requires only minor increases in test deck length, and only for those test decks which were the shortest to begin with. \textnormal{\textbf{(b)}} Length of the optimal test deck (blue) and the number of candidates in noncompetitive contests (orange) for each of the 6928 ballot styles used in Michgian's 2022 general election. The ballot styles are split across subgraphs according to the number of ballots in the heuristic-based test deck, where the top graph shows the ballot styles where the heuristic-based test decks satisfy $H^k = 21$, and bottom four graphs show the ballot styles  where the heuristic-based test decks satisfy $H^k \in \{23,28,31,43\}$. Results show that $O^k = \max \{ H^k,NC^k \}$ for all ballot styles in this election. } 
\end{figure}

To evaluate the practicality of our approach in application to Michigan's November 2022 general election, we performed the following steps. First, we calculated the lengths of optimal test decks for each of the 6928 ballot styles by  solving  the optimization problem~\eqref{prob:robust} once per ballot style.   To comply with minimum  requirements  and guidance 
in the state of Michigan, the following rules were added as constraints into the   optimization problem~\eqref{prob:robust} (see Appendix~\ref{appx:statelaws}): 
\begin{itemize}
    \item ``A different number of valid votes shall be assigned to each candidate for an office, and for and against each question'' \cite[MCL 168.798(1)]{michiganlat}. 
    \item ``None of the candidates, write-in positions, or proposals shall have an accumulated vote total of zero'' \cite[R168.773 - Rule 3(10)(a)]{michiganlat}.
\end{itemize}
Second, we calculated the lengths of heuristic-based test decks for each of the 6928 ballot styles. Our implemented heuristic involves selecting a test deck with the minimum possible length that fulfills 
\cite[MCL 168.798(1)]{michiganlat} and \cite[R168.773 - Rule 3(10)(a)]{michiganlat}  for each specific ballot style.\footnote{We compute the minimal possible length of a legally-compliant  test deck  for each ballot style as the maximum of $| \mathcal{N}_c|$ and $\lceil |\mathcal{N}_c|(|\mathcal{N}_c|+1) / 2v_c \rceil $ over all contests $c$ in the ballot style, where $| \mathcal{N}_c|$ is the number of candidates in the contest and $v_c$ is the maximum number of candidates that can be selected in the contest per ballot.} In comparison to the optimal test decks that were obtained by solving the optimization problem~\eqref{prob:robust}, the heuristic-based test decks are not feasible solutions for \eqref{prob:robust} and do not offer rigorous security guarantees for any practically important class of cyberattacks. We note that the lengths of these heuristic-based test decks serve as lower bounds on the lengths of test decks that could be obtained by \emph{any} 
heuristic that complies with \cite[MCL 168.798(1)]{michiganlat} and \cite[R168.773 - Rule 3(10)(a)]{michiganlat}.

Figure~\ref{fig:michigan} compares the distributions  of the lengths of optimal test decks and the lengths of heuristic-based test decks across the 6928 ballot styles from Michigan's November 2022 general election.   A priori, one might have anticipated that test decks that provide rigorous security guarantees  would contain significantly more ballots than the shortest test decks that satisfy a state's minimum legal requirements. However, the numerical results of our experiments in Figure~\ref{fig:michigan} show this is not the case. The results from Figure~\ref{fig:michigan} for Michigan's November 2022 general election show that the optimal test decks obtained by solving the optimization problem~\eqref{prob:robust} require only 1.2\% more ballots on average than the heuristic-based test decks across the 6928 ballot styles. Moreover, the optimal test decks require the same number of ballots as the heuristic-based test decks for all but 493 of the 6928 ballot styles. 
These results suggest that the rigorous security guarantees of our robust optimization approach to designing test decks can  be enjoyed with essentially no additional cost or difficulty to  election officials for performing LAT.

Furthermore, we find that the increases in test deck lengths in 493 of the 6928 ballot styles can be explained  by a  simple mathematical formula. To present this formula, let $H^1,\ldots,H^{6928} \ge 0$ denote the lengths of the heuristic-based test decks and $O^1,\ldots,O^{6928} \ge 0$ denote the lengths of the optimal test decks. Let a noncompetitive contest refer to any contest $c$ in a ballot style in which the maximum number of candidates that a voter is allowed to select,  denoted by $v_c$, is equal to the number of candidates in the contest, denoted by $| \mathcal{N}_c|$. In order for a test deck to satisfy the minimum legal requirement \cite[MCL 168.798(1)]{michiganlat}, we observe the test deck must assign a different number of votes to each candidate within each noncompetitive contest. Moreover, we prove in \S\ref{appx:improvements:noncompetitive} that any feasible solution of the optimization problem~\eqref{prob:robust} must assign a different number of votes to each candidate {across} \emph{all} noncompetitive contests. Letting $NC^1,\ldots,NC^{6928}$ denote the  number of candidates in noncompetitive contests, where $NC^k = \sum_{c: | \mathcal{N}_c| = v_c}  | \mathcal{N}_c|$ for each ballot style $k$,  we show in Figure~\ref{fig:max} that the formula  $O^k = \max \left \{ H^k, NC^k \right \}$ is satisfied for all ballot styles $k = 1,\ldots,6928$.
In other words, Figure~\ref{fig:max} shows that the optimization problem~\eqref{prob:robust}  yielded test decks of an equal length to current practice for every ballot style, except for the 493 ballot styles which require longer test decks to distinguish candidates in noncompetitive contests.
\begin{remark}
 When imposing 
 \cite[MCL 168.798(1)]{michiganlat} and \cite[R168.773 - Rule 3(10)(a)]{michiganlat}, we note that it is possible to construct ballot styles for which the formula $O = \max \{ H, NC \}$ does not hold. For example, consider a ballot style comprised of two contests, each with two candidates and a maximum vote of one ($| \mathcal{N}_1| = | \mathcal{N}_2| = 2$ and $v_1 = v_2 = 1$). For this ballot style, there are $NC = 0$ candidates in noncompetitive contests, and we observe that the heuristic-based test deck requires $H = 3$ ballots. However, it follows from \S\ref{appx:improvements:n_choose_k}  that the optimal test deck for this ballot style will require at least $O \ge 4$ ballots. This example demonstrates that while the formula $ O^k = \max \left \{ H^k, NC^k \right \}$ explains the lengths of optimal test decks in all 6928 ballot styles from Michigan's November 2022 general election, the formula is not guaranteed to hold in general.   
\end{remark}

\begin{figure}[bp]
\centering
\subfloat[]{\includegraphics[width=0.91\linewidth]{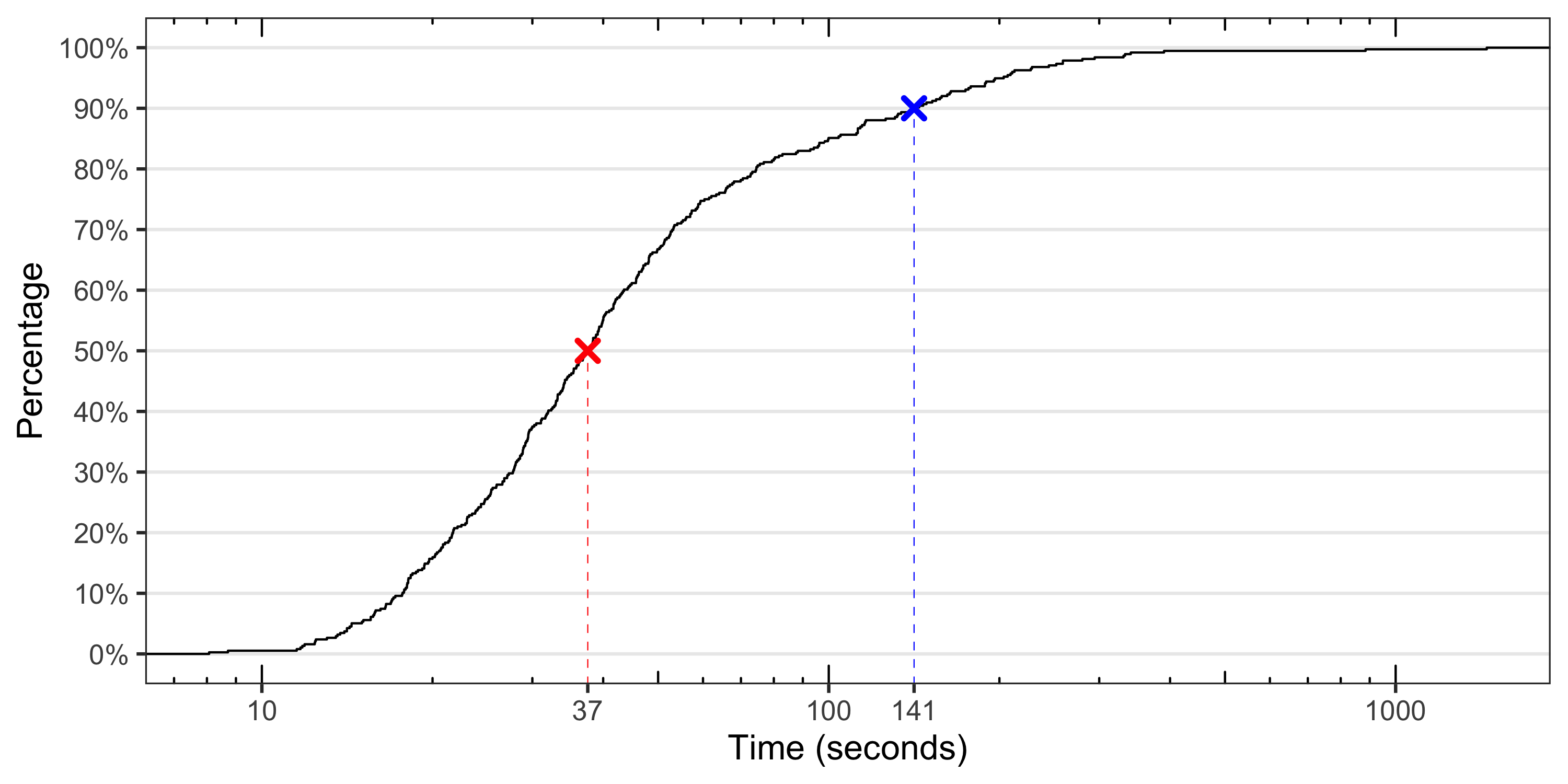} \label{fig:time_cdf}}

\subfloat[]{
\includegraphics[width=0.45\linewidth]{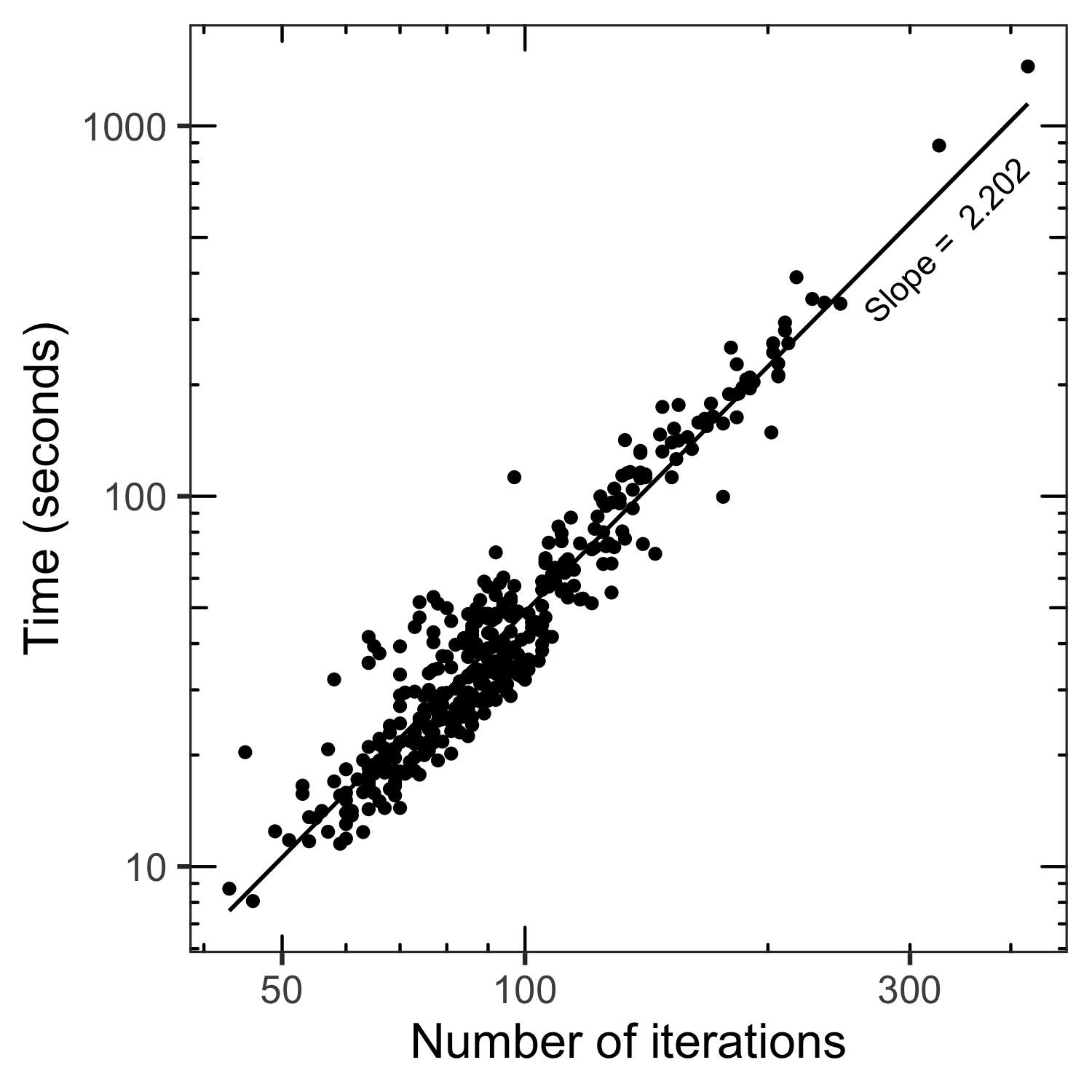} \label{fig:time_vs_iteration}}

\caption{{\normalfont \textbf{(a)}}  Computation times for the 376 invocations of our algorithm 
 that were run to compute optimal test decks for all 6928 ballot styles. Red indicates that 50\% of the invocations required less than 37 seconds. Blue indicates that 90\% of the invocations required less than 141 seconds.  Total computation time  was 6 hours and 42 minutes. {\normalfont \textbf{(b)}}  Scatterplot of the number of iterations and total computation time  of our algorithm for each  of the 376 invocations. Diagonal line shows the function $\textnormal{Time} = \exp(-6.255 + 2.202 \log (\textnormal{Iterations}))$ with  coefficient of determination $R^2 = 0.9426$. 
 }
\end{figure}

\subsection{Practical Computational Time} \label{sec:experiments:time} Due to a strict schedule for finalizing ballot styles and conducting LAT, we have found that our approach must find optimal test decks for an entire state within 24-48 hours. 

To apply our approach at scale, we  developed strategies for reusing optimal solutions across instances of the optimization problem~\eqref{prob:robust} that corresponded to similar ballot styles, which allowed us to decrease the number of invocations of our  exact algorithm from 6928 to 376 (see Appendix~\ref{appx:reuse}). Applying those strategies, our exact algorithm from \S\ref{appx:cutting_sec} computed optimal test decks for all 6928 ballot styles in less than seven hours on a home computer. Figure~\ref{fig:time_cdf} illustrates the distribution of computation times across the 376 invocations of our algorithm, showing that 90\% of the invocations required less than 2.5 minutes. These findings demonstrate that our exact algorithm from \S\ref{appx:cutting_sec} can find optimal test decks for all of the ballot styles across a state in practical computation times. Figure~\ref{fig:time_vs_iteration}  displays the number of iterations and the total computation time required for each of the 376 invocations 
of our algorithm. The results of Figure~\ref{fig:time_vs_iteration} show that the computation time of the algorithm is driven by the number of iterations of our cutting plane method,  demonstrating the value of the  algorithmic developments in \S\ref{appx:improvements} which decrease the number of iterations of the cutting plane method. 

\section{Conclusion}

This paper  describes the first formal procedure for detecting cyberattacks in computerized voting machines \emph{prior} to their use in elections. We achieve this by applying rigorous scientific reasoning to a widely used pre-election procedure, Logic and Accuracy Testing, which for more than a century has been performed using human intuition and simple heuristics. Unlike the longstanding  practice of LAT, our approach provides a  guarantee that LAT will detect any misconfiguration that swaps voting targets between candidates, whether those misconfigurations were induced deliberately or by human error. Such misconfigurations have occurred accidentally in recent elections in Michigan, Pennsylvania, and Georgia. Although these errors were later caught and corrected, they generated negative publicity, hurt public confidence, and  served as the basis for a draft executive order which would have instructed the military to seize voting machines nationwide. We showed in \S\ref{sec:vulnerabilities}  that similar misconfigurations could be strategically induced by technically unsophisticated adversaries to undermine public trust or change the outcome of an election. 
By applying tools from robust optimization to LAT, this paper offers a practical and scientifically rigorous way to defend against the aforementioned risks in future elections, and demonstrates that advanced computational tools can be used to realize novel benefits to public institutions.\looseness=-1

 Through our partnership with the Michigan Bureau of Elections, we found that our approach
offered valuable security guarantees with only a 1.2\% average increase in the number of test ballots  compared to existing testing procedure in the state’s November 2022 general election. Coupled with the practical computation time of our algorithm, we conclude that our approach to obtaining rigorous security guarantees with LAT is well suited to deployment throughout the United States. We 
hope that other states and countries will adopt our approach as a low-cost tool to improving the security
and increasing public confidence in election outcomes.

         There are many interesting directions for future work. 
        \begin{itemize}
            \item First, we foresee ways that randomization can be used to generate  short test decks with probabilistic  security guarantees. We  provide evidence in \S\ref{sec:experiments} that for states with legal requirements like Michigan's, the benefit of using randomization to design short test decks  would be minuscule. This is because the test decks produced by our deterministic approach in Michigan's November 2022 election  were only 1.2\% longer on average than the shortest test decks that satisfy Michigan's minimal legal requirements.  
        However,  randomization may be useful for designing short test decks when considering more expansive classes of uncertainty sets (i.e., uncertainty sets that go beyond incorrect bijection mappings). Randomization could also help  facilitate adoption of the robust optimization approach in states that currently have weak legal requirements that must be satisfied by test decks (e.g., states that do not require every candidate within a contest to have a distinct number of votes), as election officials in such states would be  accustomed to shorter test decks.  
        
        \item 
        It is straightforward to see that the length of test decks could theoretically decrease if an election official could output the vote totals after each ballot is fed into the voting machine. However, it takes significant time for the voting machine to print the poll tape (i.e., the grocery store-like receipt that the machine prints out to show the vote totals) and significant time for the election official to then reset the machine after it prints a poll tape in order for the machine to scan more ballots. From conversations with election officials, we learned that a test deck that requires more than a few poll tapes to be printed is viewed by election officials as too time consuming to perform, too different from the current practice of logic and accuracy testing, and  would thus be unlikely to be followed by election officials. In view of these practical considerations, an interesting future direction would be to characterize the savings that could be obtained in  test deck length if election officials were required to print the poll tape a small 
 (but greater than one) number of times during testing. 

        \item On the theoretical side, many open questions remain about the computational complexity of the robust optimization problem, whether it would be possible to design  approximation algorithms, and if special cases of the robust optimization problem can be solved in polynomial time.  
        
        \end{itemize}

\section*{Acknowledgements}
The authors are grateful to the Michigan Department of State Bureau of Elections for their partnership in this research. This material is based upon work supported by the U.S. National Science Foundation under Grant No.\ CNS-1518888. Any opinions, findings, and conclusions or recommendations expressed in this material are those of the authors and do not necessarily reflect the views of the National Science Foundation.

\bibliographystyle{plainnat}
\bibliography{scibib}

\clearpage

\appendix

\begin{center}
    \huge{\textbf{Appendices}}
\end{center}

The appendices have the following organization:
\begin{itemize}
\item Appendix~\ref{appx:upperbound} proposes two simple heuristics for obtaining feasible solutions for the optimization problem~\eqref{prob:robust}. Using real-world  data, we show that these simple heuristics will result in test decks that contain too many ballots to be used in practice. 

\item Appendix~\ref{appx:malicious-decks} describes the capabilities of an adversary who chooses both the voting machine's configuration and the test deck used in LAT.

\item Appendix~\ref{appx:additional_experiments} contains  additional numerical experiments  that showcase the value of the five improvements from \S\ref{appx:improvements} on the practical efficiency of the cutting plane method from \S\ref{appx:cutting_sec}.

\item Appendix~\ref{appx:statelaws}  shows that various state-level legal requirements on the design of test decks can be  enforced either by adding constraints into the optimization problem~\eqref{prob:robust} or by augmenting the output of the optimization problem~\eqref{prob:robust}. 

\item Appendix~\ref{appx:reuse} identifies circumstances in which an optimal test deck for the optimization problem~\eqref{prob:robust} for one ballot style can be efficiently translated into an optimal test deck for another similar ballot style.

\item Appendix~\ref{appx:proofs} contains the proofs of the paper's technical results. 
\end{itemize}
\clearpage

\section{Upper Bounds} \label{appx:upperbound}
Our exact algorithm for solving the optimization problem~\eqref{prob:robust} is found in \S\ref{appx:cutting_sec}. In this appendix, we motivate the exact algorithm by presenting and analyzing two simple heuristics for the optimization problem~\eqref{prob:robust}. These two heuristics, which can be found below in Propositions~\ref{prop:heuristic1} and \ref{prop:heuristic2}, obtain a feasible solution for the optimization problem~\eqref{prob:robust} by constructing a test deck that contains a distinct positive number of votes for each candidate. Our purpose for presenting these heuristics is  (a)  to show that the optimization problem~\eqref{prob:robust} always has a feasible solution and (b) to show using real-world data that  heuristics based on  assigning a distinct number of votes for each candidate will result in test decks that contain too many ballots to be implementable in practice.\looseness=-1 

Our first simple heuristic for the optimization problem~\eqref{prob:robust}  is stated formally in the proof of the following Proposition~\ref{prop:heuristic1}. The heuristic consists of constructing a test deck in which each filled-out ballot in the test deck contains a vote for exactly one candidate (i.e., $| \beta_1| = \cdots = | \beta_B| = 1$) and in which each candidate $i \in \mathcal{N}$ is selected in exactly $i$ of the  filled-out ballots.  The fact that this heuristic yields a  test deck that is feasible for the optimization problem~\eqref{prob:robust} is shown in the proof of Proposition~\ref{prop:heuristic1} to follow from Corollary~\ref{cor:diff_votes} coupled with the fact that the heuristic gives a distinct total number of votes to each of the candidates.  More generally, this heuristic is useful because it yields a simple, closed-form upper bound on the length of optimal test decks for the optimization problem~\eqref{prob:robust}.

\begin{proposition} \label{prop:heuristic1}
There  exists a feasible solution for the optimization problem~\eqref{prob:robust} that   satisfies  $B \le N(N+1) / 2$. 
\end{proposition}

Our second heuristic for the optimization problem~\eqref{prob:robust} can be viewed as a refinement of the first heuristic from Proposition~\ref{prop:heuristic1}. Like the first heuristic, our second heuristic  yields a  test deck that is feasible for the optimization problem~\eqref{prob:robust} by giving a distinct  positive number of votes to each candidate. However, our second heuristic  assigns a different positive number of votes to each candidate in such a way that allows the votes to be packed into the fewest number of ballots. More specifically, our second heuristic consists of solving the following optimization problem~\eqref{prob:distinct} to  find a minimum-length test deck that assigns a distinct positive number of votes to each of the candidates across each of the contests: 
\begin{subequations}\label{prob:distinct}
\begin{align} 
\underset{B \in \N, \; \beta_1,\ldots,\beta_B \in \mathscr{B}}{\textnormal{minimize}}\quad & B \label{prob:distinct:objective}\\
\textnormal{subject to} \quad & \left|\left \{ b \in \{1,\ldots,B \}: i \in \beta_b \right \}  \right|  \neq \left|\left \{ b \in \{1,\ldots,B \}: j \in \beta_b \right \}  \right|  && \forall i,j \in \mathcal{N}: i \neq j \label{prob:distinct:distinct}\\
&\left|\left \{ b \in \{1,\ldots,B \}: i \in \beta_b \right \}  \right| \ge 1 && \forall i \in \mathcal{N}.\label{prob:distinct:positive} 
\end{align}
\end{subequations}
Indeed, constraint~\eqref{prob:distinct:distinct} ensures that the test deck gives a distinct number of votes to each candidate, and constraint~\eqref{prob:distinct:positive} ensures that each candidate receives at least one vote. 
It follows immediately from Corollary~\ref{cor:diff_votes} that any test deck that is feasible for the optimization problem~\eqref{prob:distinct} is feasible for the optimization problem~\eqref{prob:robust}. Hence, the optimization problem~\eqref{prob:distinct}  provides the tightest upper bound on the optimal objective value of the optimization problem~\eqref{prob:robust} that can be obtained by a test deck that assigns a distinct positive number of votes to each of the candidates across all of the contests. 

In the following Proposition~\ref{prop:heuristic2}, we show that the optimization problem~\eqref{prob:distinct} can be reformulated as a mixed-integer linear optimization problem.  In contrast to the optimization problem~\eqref{prob:distinct}, the  mixed-integer linear optimization problem~\eqref{prob:distinct_mip} from the following Proposition~\ref{prop:heuristic2}  can be easily implemented and solved using widely available open-source and commercial optimization software such as Gurobi and Mosek.   
\begin{proposition}\label{prop:heuristic2}
The optimal objective value of the optimization problem~\eqref{prob:distinct} is equal to the optimal objective value of the following mixed-integer linear optimization problem~\eqref{prob:distinct_mip}, and every optimal solution for \eqref{prob:distinct_mip} can be transformed into an optimal solution for \eqref{prob:distinct}. 
  \begin{subequations} \label{prob:distinct_mip}
    \begin{align}
            &\underset{B \in \mathbb{N}, \gamma \in \{0,1\}^{\mathcal{C} \times \mathcal{N}}}{\textnormal{minimize}} && B \label{prob:distinct_mip:objective}\\
            &\textnormal{subject to} && \sum_{g \in \mathcal{N}} \gamma_{c,g} = |\mathcal{N}_c| && \forall c \in \mathcal{C} \label{prob:distinct_mip:N_c}\\
            &&& \sum_{c \in \mathcal{C}} \gamma_{c,g} = 1 && \forall g \in \mathcal{N} \label{prob:distinct_mip:sum_to_one}\\
            &&& B \geq \frac{1}{v_c} \sum_{g \in \mathcal{N}} g \gamma_{c,g} && \forall c \in \mathcal{C} \label{prob:distinct_mip:bound_B}\\
            &&& B \ge N. \label{prob:distinct_mip:N}
    \end{align}
    \end{subequations}
\end{proposition}
\noindent Let us provide an interpretation of  the decision variables and constraints of the mixed-integer linear optimization problem~\eqref{prob:distinct_mip}. Each binary decision variable $\gamma_{c,g}$ will be equal to one if and only if there exists a candidate $i \in \mathcal{N}_c$ in contest $c$ that satisfies the equality $\left|\left \{ b \in \{1,\ldots,B \}: i \in \beta_b \right \}  \right| = g$. 
Constraints~\eqref{prob:distinct_mip:N_c} and \eqref{prob:distinct_mip:sum_to_one} ensure that a distinct number of votes are given to each of the candidates across  each of the contests. 
Constraint~\eqref{prob:distinct_mip:bound_B} enforces, for each contest $c \in \mathcal{C}$, the fact that a test deck $\beta_1,\ldots,\beta_B \in \mathscr{B}$ needs to be comprised of  at least $B \ge \lceil \frac{1}{v_c} \sum_{g \in \mathcal{N}} g \gamma_{c,g} \rceil$ ballots in order for it to be possible for the test deck to give $ \sum_{g \in \mathcal{N}} g \gamma_{c,g} $ votes to the candidates in contest $c$ without causing an overvote for that contest in any of the ballots. 
Constraint~\eqref{prob:distinct_mip:N} enforces that we include at least the $N$ ballots which are necessary for the candidate who receives $N$ votes.

We conclude Appendix~\ref{appx:upperbound} by showing using real-world data that any  heuristic that  assigns  a distinct positive number of votes to each candidate across all of the contests  will result in test decks that contain too many ballots to be useful in practice. Specifically, we applied our second  heuristic to the 6928 ballot styles that appeared in the state of Michigan in the November 2022 general election.   In Figure~\ref{fig:upper_bound}, we compare the number of ballots for test decks that are optimal for the optimization problem~\eqref{prob:robust}  and the number of ballots for test decks that are optimal for the optimization problem~\eqref{prob:distinct_mip}.   The results of Figure~\ref{fig:upper_bound} thus demonstrate that test decks that assign distinct positive numbers of votes for candidates across contests can require significantly (2.46x to 6.38x) more ballots than the test decks obtained by solving the optimization problem~\eqref{prob:robust}.

 \begin{figure}[t]
    \centering
\includegraphics[width=0.99\linewidth]{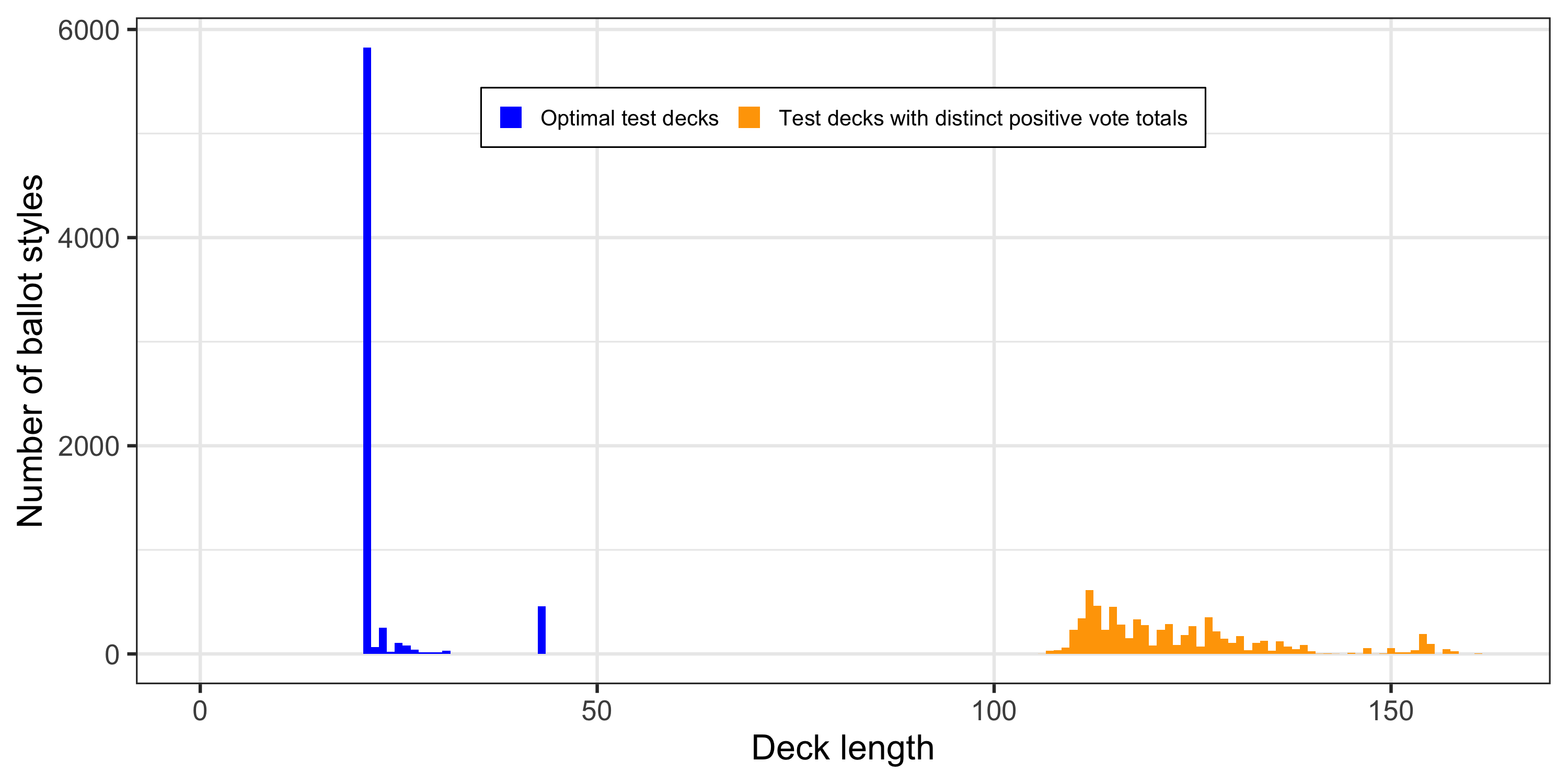}  

\caption{Distributions of lengths of test decks obtained by solving the optimization problem~\eqref{prob:robust} (blue) and test decks obtained by solving the optimization problem~\eqref{prob:distinct_mip} (orange) across the 6928 ballot styles. The large gap between the two distributions indicates that test decks that are optimal for the optimization problem~\eqref{prob:robust} require significantly fewer ballots than test decks that give a distinct positive number of votes to each candidate across all contests. }\label{fig:upper_bound}

\end{figure}

\section{Deliberately Flawed Test Decks} \label{appx:malicious-decks}

In \S\ref{sec:vulnerabilities}, we discussed a category of threats (dubbed `Deliberately Flawed Test Decks') in which an adversary may have the opportunity to both configure the voting machine as well as design the test deck used in LAT. If this is the case, then the adversary could choose a mapping $\sigma \in \Sigma$ to use on the machine, then choose a test deck $(\beta_1, \ldots, \beta_B)$ that satisfies $T^\sigma(\beta_1, \ldots, \beta_B) = T^*(\beta_1, \ldots, \beta_B)$, thereby causing their  misconfiguration to go undetected by LAT.

Such an adversary has significant freedom with respect to the $\sigma \in \Sigma$ they choose if the test deck $(\beta_1, \ldots, \beta_B)$ is not well constrained by the state's minimum legal requirements. For example, if the only minimum legal requirement  on the test deck is that each candidate receives at least one  vote---as is the case in a number of states~\cite{walker2022lat}---the adversary can find a suitable deck for any mapping $\sigma \in \Sigma$. We formalize this observation through the following Proposition~\ref{prop:adversary}. 
\begin{proposition} \label{prop:adversary}
    For every $\sigma \in \Sigma$, there exists a test deck $\beta_1, \ldots, \beta_B  \in \mathscr{B}$ that satisfies
    \begin{align*}
        T^*_i(\beta_1,\ldots,\beta_B) \ge 1 \; \forall i \in \mathcal{N} \text{ and }   T^\sigma(\beta_1,\ldots,\beta_B) = T^*(\beta_1,\ldots,\beta_B).
    \end{align*}
\end{proposition}

Seventeen states also require that each candidate receives at least one vote and that no two candidates in the same contest receive the same number of votes~\cite{walker2022lat}. This minimum legal requirement rules out all $\sigma \in \Sigma$ that have a cycle of swaps that includes two candidates from the same contest. For any other $\sigma \in \Sigma$, however, a suitable test deck can be generated by the adversary to hide their chosen non-identity bijection, as shown by the following Theorem~\ref{thm:swap_adversary}. In the following Theorem~\ref{thm:swap_adversary} and throughout the paper, we use the notation $\sigma^n(\cdot)$ to denote the $n$-fold composition of $\sigma(\cdot)$.\footnote{For example, $\sigma^1(\cdot) \triangleq \sigma(\cdot)$ and $\sigma^2(\cdot) \triangleq \sigma(\sigma(\cdot))$.}\looseness=-1
\begin{theorem} \label{thm:swap_adversary}
    Consider any $\sigma \in \Sigma$ that satisfies the following property for all contests $c \in \mathcal{C}$, candidates $i \in \mathcal{N}_c$, and integers $n \in \N$:
    \begin{align*}
        \sigma^n(i) \in \mathcal{N}_c \iff \sigma^n(i) = i. 
    \end{align*}
    Then there exists a test deck $\beta_1,\ldots,\beta_B \in \mathscr{B}$ that satisfies 
    \begin{align*}
        T^*_i(\beta_1,\ldots,\beta_B) \ge 1 & \quad \forall i \in \mathcal{N}\\
        T^*_i(\beta_1, \ldots, \beta_B) \neq T^*_j(\beta_1, \ldots, \beta_B) & \quad \forall c \in \mathcal{C}, i,j \in \mathcal{N}_c: i \neq j \\
        T^\sigma(\beta_1,\ldots,\beta_B) = T^*(\beta_1,\ldots,\beta_B).
    \end{align*}
\end{theorem}

\section{Additional Numerical Experiments} \label{appx:additional_experiments}
In this appendix, we present additional numerical experiments  to explore and showcase the value of the proposed improvements from \S\ref{appx:improvements} on the practical efficiency of the cutting plane method from \S\ref{appx:cutting_sec}. 

\subsection{Experiment Setup} \label{appx:additional_experiments:setup}
We investigate the value of the five improvements from \S\ref{appx:improvements} using three classes of numerical experiments. The three classes of numerical experiments are described as follows.
\begin{itemize}
\item \emph{Experiment 1}: 
In this experiment, we generate ballot styles with varying numbers of candidates per contest. Specifically, 
for each  $C \in \{2,\ldots,12\}$, we generate a ballot style with $C$ contests where the contests are comprised as  
\begin{align*}
    |\mathcal{N}_1| &= 1,&| \mathcal{N}_2| &= 2,&\cdots&&|\mathcal{N}_{C-1}| &= C-1,&|\mathcal{N}_{C}| &= C,\\
    v_1 &= 1,&v_2 &= 1,&\cdots && v_{C-1} &= 1,& v_C &= 1.
\end{align*} 
\item \emph{Experiment 2}: 
In this experiment, we generate ballot styles with contests that have the same numbers of candidates. Specifically, 
for each  $C \in \{2,\ldots,12\}$, we generate a ballot style with $C$ contests where the contests are comprised as  
\begin{align*}
    |\mathcal{N}_1| &= 2,&| \mathcal{N}_2| &= 2,&\cdots&&|\mathcal{N}_{C-1}| &=2,&|\mathcal{N}_{C}| &= 2,\\
    v_1 &= 1,&v_2 &= 1,&\cdots && v_{C-1} &= 1,& v_C &= 1.
\end{align*} 
\item \emph{Experiment 3}: 
In this experiment, we generate ballot styles with noncompetitive contests. Specifically, 
for each  $C \in \{2,\ldots,12\}$, we generate a ballot style with $C$ contests where the contests are comprised as  
\begin{align*}
    |\mathcal{N}_1| &= 1,&| \mathcal{N}_2| &= 2,&\cdots&&|\mathcal{N}_{C-1}| &=C-1,&|\mathcal{N}_{C}| &= C,\\
    v_1 &= 1,&v_2 &= 2,&\cdots && v_{C-1} &= C-1,& v_C &= C.
\end{align*} 
\end{itemize}
Our goal in each experiment is to examine the individual impact of each of the proposed improvements from \S\ref{appx:improvements} on the practical efficiency of the cutting plane method. To this end, we report on the performance of the following solution methods:
\begin{itemize}
    \item \emph{All Improvements}: In this solution method, we find an optimal solution for the optimization problem~\eqref{prob:robust}  by using the cutting plane method from \S\ref{appx:cutting}. In each iteration of the cutting plane method, we solve the optimization problems~\eqref{prob:robust-subset} and \eqref{prob:oracle} using the mixed-integer linear optimization reformulations~\eqref{prob:complete_transposition} and \eqref{prob:oracle_mip} from \S\ref{appx:mip}. Moreover, we use each of the five improvements from \S\ref{appx:improvements}. 
    \item \emph{No Improvement 1}: Same solution method as {All Improvements}, but we do not use Improvement 1 from \S\ref{appx:improvements:reducingdecisions}.
        \item \emph{No Improvement 2}: Same solution method as {All Improvements}, but we do not use Improvement 2 from \S\ref{appx:improvements:symmetry}.
 \item \emph{No Improvement 3}: Same solution method as {All Improvements}, but we do not use Improvement 3 from \S\ref{appx:improvements:n_choose_k}.
 \item \emph{No Improvement 4}: Same solution method as {All Improvements}, but we do not use Improvement 4 from \S\ref{appx:improvements:swap}.
 \item \emph{No Improvement 5}: Same solution method as {All Improvements}, but we do not use Improvement 5 from \S\ref{appx:improvements:noncompetitive}.
\end{itemize}
For each $C \in \{2,\ldots,12\}$ in each of the Experiments 1, 2, and 3, we used each of the solution methods  to compute an optimal test deck. We recorded the total computation time and number of iterations of the cutting plane method for each solution method. If a solution method on a ballot style  required a computation time that exceeded  one hour (3600 seconds), then we terminated the solution method early without finding an optimal solution.  All numerical experiments were conducted using the Gurobi optimization solver on a laptop with a 2.6 GHz 6-Core Intel Core i7 processor and 16 GB of RAM. In all experiments, we also impose a constraint in the optimization problem~\eqref{prob:robust}  that each candidate must receive at least one vote (see Appendix~\ref{appx:statelaws:atleastone}).

\subsection{Results}  \label{appx:additional_experiments:results}

The results of our numerical experiments from Appendix~\ref{appx:additional_experiments:setup} are presented in Figure~\ref{fig:toy_experiments}. For  visual clarity,  we  do not display the numerical results for No Improvement 3 in Experiments 1 and 3\footnote{Experiments 1 and 3 do not include ballot styles in which there are contests that are equivalent (see Definition~\ref{defn:equivalence} in  \S\ref{appx:improvements:n_choose_k}), and so No Improvement 3 is identical to All Improvements in the context of  Experiments 1 and 3.} and  do not display the numerical results for No Improvement 5 in Experiments 1 and 2.\footnote{Experiments 1 and 2 do not include multiple noncompetitive contests  (see \S\ref{appx:improvements:noncompetitive}), and so  No Improvement 5 is identical to All Improvements in the context of  Experiments 1 and 2.} We reflect below on the key numerical findings from Figure~\ref{fig:toy_experiments}: 
\begin{itemize}
    \item \emph{No Improvement 1}: Experiments 1 and 3 show that the improvement from  \S\ref{appx:improvements:reducingdecisions} significantly increases the practical efficiency of the cutting plane method by decreasing the per-iteration computation cost. This is seen most clearly in Experiment 3 in the ballot style with $C = 12$ contests, where the computation time of No Improvement 1 is approximately 25x greater than the computation time of All Improvements,  despite both solution methods requiring only a single iteration.
    \item \emph{No Improvement 2}: Experiments 1, 2, and 3 show that the improvement from \S\ref{appx:improvements:symmetry} significantly increases the practical efficiency of the cutting plane method by decreasing the number of iterations.  For example, in Experiment  1 in the ballot style with $C = 12$ contests, the number of iterations of No Improvement 2 is approximately 21x greater than the number of iterations of All Improvements, leading to a computation time of No Improvement 2 that is approximately 45x greater than the computation time of All Improvements. Moreover, we observe that No Improvement 2 did not terminate in less than one hour in Experiment 2 with $C \ge 9$ contests and in Experiment 3 with $C \ge 11$ contests. 
    \item   \emph{No Improvement 3}: Experiment 2 shows that the improvement from \S\ref{appx:improvements:n_choose_k} significantly increases the practical efficiency of the cutting plane method by decreasing the number of iterations. Indeed, in Experiment 2 in the ballot style with $C = 8$ contests, the number of iterations of No Improvement 3 is approximately 16x greater than the number of iterations of All Improvements, leading to a computation time of No Improvement 3 that is approximately 1106x greater than the computation time of All Improvements. Moreover, we observe that No Improvement 3 did not terminate in less than one hour in Experiment 2 with $C \ge 9$ contests.
    \item \emph{No Improvement 4}: Experiment 2 shows that the improvement from \S\ref{appx:improvements:swap} significantly increases the practical efficiency of the cutting plane method by decreasing the number of iterations. This is seen most clearly in Experiment 2 in the ballot style with $C = 12$ contests, where the number of iterations of No Improvement 4 is approximately 6x greater than the number of iterations of All Improvement, leading to a computation time of No Improvement 4 that is approximately 59x greater than the computation time of All Improvements. Experiments 1 and 3, in contrast,  did do show any meaningful advantages or disadvantages of using the improvement from \S\ref{appx:improvements:swap}. 
        \item \emph{No Improvement 5}: Experiment 3 shows that the improvement from \S\ref{appx:improvements:noncompetitive} significantly increases the practical efficiency of the cutting plane method by decreasing the number of iterations. This is seen most clearly in Experiment 3 in the ballot style with $C = 7$ contests, where the number of iterations of No Improvement 5 is 323x greater than the number of iterations of All Improvement, leading to a computation time of No Improvement 5 that is approximately  3129x greater than the computation time of All Improvements.  Moreover, we observe that No Improvement 5 did not terminate in less than one hour in Experiment 3 with $C \ge 11$ contests.

\end{itemize}

\begin{figure}[t]
\centering 
\includegraphics[width=1.0\linewidth]{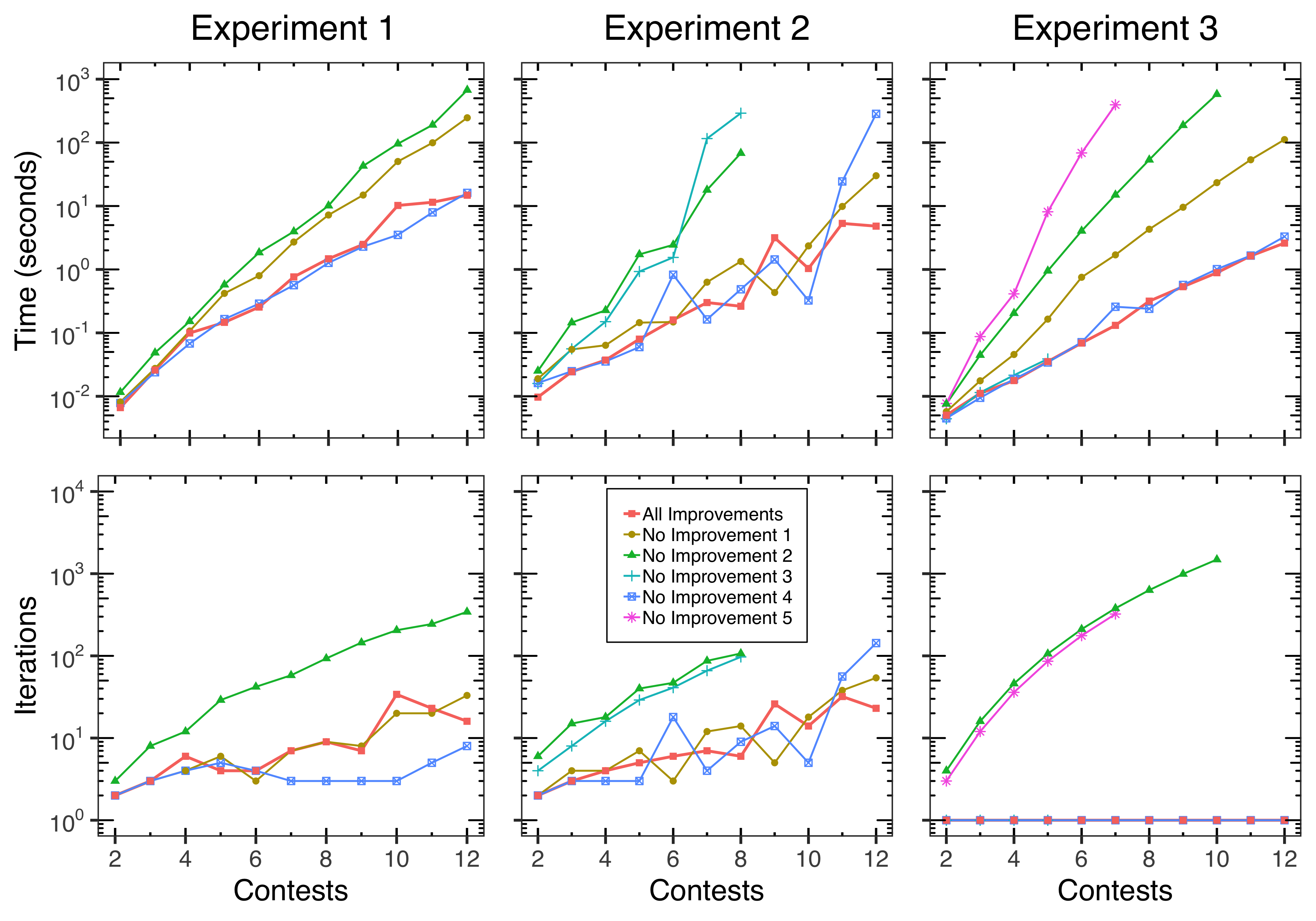}

\caption{{Numerical results for Appendix~\ref{appx:additional_experiments:results}.}  } \label{fig:toy_experiments}

\end{figure}

\section{State-Level Requirements} \label{appx:statelaws}
Each of the fifty states has minimal requirements on the design of test decks that can be legally used in LAT. For example, at least forty states have a minimum  requirement that test decks must include at least one vote for  each candidate on the ballot \cite{walker2022lat}. 
In this appendix, we demonstrate how various state-level requirements can be enforced in the optimization problem~\eqref{prob:robust}. 

\subsection{At Least One Vote Per Candidate}  \label{appx:statelaws:atleastone}
At least forty states 
recommend  that none of the candidates, write-in positions, or proposals shall have an
accumulated vote total of zero~\cite{walker2022lat}. To add this recommendation as a constraint into the optimization problem~\eqref{prob:robust}, we add it as a constraint into the optimization problem~\eqref{prob:robust-subset} in each iteration of the cutting plane method that is described in \S\ref{appx:cutting}. In particular, we recall from \S\ref{appx:mip:robust-subset} that the optimization problem~\eqref{prob:robust-subset} is equivalent to the mixed-integer linear optimization problem~\eqref{prob:complete_transposition}, where the mixed-integer linear optimization problem~\eqref{prob:complete_transposition} includes the following constraints:
\begin{align}
& \sum_{g \in \mathcal{B}_0} \gamma_{i,g} = 1 && \forall i \in \mathcal{N} \tag{\ref{prob:complete_transposition:gamma_1}}\\
& \sum_{b \in \mathcal{B}} \beta_{b,i} = \sum_{g \in \mathcal{B}_0} g \gamma_{i,g}  && \forall i \in \mathcal{N} .\tag{\ref{prob:complete_transposition:gamma_2}}
\end{align}
We recall from the discussion in \S\ref{appx:mip:robust-subset} that the above constraints~\eqref{prob:complete_transposition:gamma_1} and \eqref{prob:complete_transposition:gamma_2} enforce for each candidate $i \in \mathcal{N}$ that the binary decision variable $\gamma_{i,g} \in \{0,1\}$ in the mixed-integer linear optimization problem~\eqref{prob:complete_transposition} will be equal to one if and only if that candidate receives exactly $g \in \mathcal{B}_0 \equiv \{0,\ldots,B\}$ votes in the test deck. Therefore, to enforce that each candidate receives at least one vote in the test deck, we can add the following constraint to the mixed-integer linear optimization problem~\eqref{prob:complete_transposition}:
\begin{align}
    &\gamma_{i,0} = 0  \quad  \forall i \in \mathcal{N}.  \label{prob:remove_zero_gamma}
\end{align}

\subsection{Distinct Votes for Candidates in the Same Contest}
At least seventeen states (including Michigan, see  \cite[MCL 168.798(1)]{michiganlat})
have a requirement that test decks must assign a distinct number
of votes to candidates in the same contest~\cite{walker2022lat}. Following identical reasoning as in Appendix~\ref{appx:statelaws:atleastone}, we observe that enforcing this requirement on test decks can be accomplished  by adding constraints into the mixed-integer linear optimization problem~\eqref{prob:complete_transposition} from \S\ref{appx:mip:robust-subset}. However, we recall that \S\ref{appx:improvements:symmetry} offers an improvement to the cutting plane method that consists of adding the extra  constraints~\eqref{prob:complete_transposition:extra_withincontest} to the mixed-integer linear optimization problem~\eqref{prob:complete_transposition}. Because those extra constraints enforce that candidates within  the same contest receive different numbers of votes, we conclude that the state requirement is accomplished by using the improvement from \S\ref{appx:improvements:symmetry}. 

\subsection{Overvoted Ballots} \label{appx:statelaws:overvotes}
Several states (including Michigan, see \cite[MCL 168.776 Rule 6(4)(f)]{michiganlat}) require that the test deck include at least one ballot that contains an overvote in one or more contests~\cite{walker2022lat}. As we show in the following Proposition~\ref{prop:overvote_1}, this requirement can be satisfied by solving the optimization problem~\eqref{prob:robust} to obtain a test deck, and then appending that test deck with an additional filled-out ballot  that includes a vote for every candidate in every contest.
\begin{proposition} \label{prop:overvote_1}
    Let $\bar{\beta} \triangleq \mathcal{N}$ denote the filled-out ballot that includes a vote for every candidate in every contest. If $(B,\beta_1,\ldots,\beta_B)$ is a feasible solution for the optimization problem~\eqref{prob:robust}, then the following holds:
    \begin{align*}
        T^\sigma(\beta_1,\ldots,\beta_B, \bar{\beta}) \neq T^* (\beta_1,\ldots,\beta_B,  \bar{\beta}) \quad \forall \sigma \in \Sigma. 
    \end{align*}
\end{proposition}
\noindent In particular, the above proposition shows that if we solve the optimization problem~\eqref{prob:robust}, and we then augment the optimal test deck by adding a filled-out ballot that votes for every candidate in every contest, then the augmented test deck will retain the desired security guarantee that the output $T^\sigma(\beta_1,\ldots,\beta_B, \bar{\beta})$ of a voting machine with any mapping $\sigma \in \Sigma$ will be different from the output $T^*(\beta_1,\ldots,\beta_B,\bar{\beta})$ of the voting machine that operates correctly.

While the above strategy can be used in many states including Michigan, certain states impose an additional requirement that 
the overvoted ballot in the test deck must cast precisely $v_c + 1$ votes in each contest $c \in \mathcal{C}$ that satisfies $v_c > | \mathcal{N}_c|$~\cite{walker2022lat}. In those states, the  strategy from Proposition~\ref{prop:overvote_1} cannot be applied when there exist contests that satisfy the inequality $v_c > | \mathcal{N}_c| + 1$, as the strategy from Proposition~\ref{prop:overvote_1} would require the overvoted ballot to vote for  strictly greater than $v_c + 1$ candidates in some contests.\footnote{An example of a contest that typically satisfies $v_c > | \mathcal{N}_c| + 1$ is the Presidential contest; even though voters may select up to one candidate in the presidential contest, voters will typically be allowed to select from candidates from four or more political parties.} 

To develop test decks for states with the aforementioned additional requirement on the overvoted ballot, we consider the following assumption on their voting machines.
\begin{assumption} \label{ass:functionality}
When a voting machine interprets a ballot as containing an overvote in at least one contest, it will produce an ``overvote alert'' notification. This alert will not specify which contest(s) are interpreted as containing an overvote, but will allow for a determination of which \emph{ballots} contain some overvoted contest. 
\end{assumption}

Most modern voting machines are designed to satisfy this assumption. Indeed, in many jurisdictions, the purpose of including overvoted ballots in the test deck is precisely to evaluate whether this functionality works as expected (see, e.g.,~\cite[p.91]{arizona-elections-procedures-manual} or~\cite[p.62]{arkansas-procedures-manual}). If this assumption is believed to hold, then we show in the following Proposition~\ref{prop:overvote_2} that we can satisfy the aforementioned stricter requirement by solving the optimization problem~\eqref{prob:robust} to obtain a test deck, and then appending that test deck with an additional filled-out ballot that casts precisely $v_c + 1$ votes for each contest $c \in \mathcal{C}$ that satisfies $v_c > | \mathcal{N}_c|$.

\begin{proposition} \label{prop:overvote_2}
    Let $\tilde{\beta} \subseteq \mathcal{N}$ denote a filled-out ballot that satisfies the following equality for each contest $c \in \mathcal{C}$:
    \begin{align*}
        | \tilde{\beta} \cap \mathcal{N}_c| &= \begin{cases}
            v_c + 1, &\textnormal{if } | \mathcal{N}_c| > v_c,\\
            0,&\textnormal{otherwise}. 
        \end{cases}
    \end{align*} If $(B,\beta_1,\ldots,\beta_B)$ is a feasible solution for the optimization problem~\eqref{prob:robust}, then for all $\sigma \in \Sigma$, at least one of the following two conditions hold:
    \begin{itemize}
        \item { There exists a contest in at least one of the filled-out ballots  $\beta_1,\ldots,\beta_B \in \mathscr{B}$  that is interpreted by a voting machine with mapping $\sigma$ as containing an overvote; that is, there exist  $b \in \mathcal{B}$ and $c \in \mathcal{C}$ that satisfy $|\{\sigma(j) \in \beta_b : j \in \mathcal{N}_c\}| > v_c$. }
        \item { $ T^\sigma(\beta_1,\ldots,\beta_B, \tilde{\beta}) \neq T^* (\beta_1,\ldots,\beta_B,  \tilde{\beta})$.}
    \end{itemize}
\end{proposition}

Proposition~\ref{prop:overvote_2} tells us that each incorrect mapping $\sigma \in \Sigma$ will be detected by the modified deck $(\beta_1, \ldots, \beta_B, \tilde{\beta})$ if Assumption~\ref{ass:functionality} holds. This is because either the output of the voting machine  will differ from what is expected (i.e., $ T^\sigma(\beta_1,\ldots,\beta_B, \tilde{\beta}) \neq T^* (\beta_1,\ldots,\beta_B,  \tilde{\beta})$) or because some feasible ballot will produce an alert indicating that it has been overvoted (i.e., there will exist some $b \in \mathcal{B}$ and $c \in \mathcal{C}$ which satisfy $ \sum_{i \in \mathcal{N}_c} \beta_{b, \sigma(i)} \geq v_c + 1$).   In either event, the behavior of the machine under  mapping $\sigma \in \Sigma$ will be distinguishable from the behavior of a machine which is operating correctly. This allows us to satisfy the aforementioned state requirement by including ballot $\tilde{\beta}$ while maintaining the desired security guarantees.

\subsection{Party-Line Option}
Several states including Michigan (see \cite[p.18]{michiganlat}) provide an option to voters in certain ballot styles to choose a so-called ``party-line option". A party-line option is a special target on a ballot that, if marked, defaults the ballot to selecting a specific party's candidates (e.g., Republican candidates, Democrat candidates) in each of the contests. The voting machine functionality related to the processing of party-line voting is complex due to the fact that voters can select a party-line vote but also can, if desired, override the default party selection in one or more contests. Because of this complexity, states such as Michigan provide a separate set of test deck requirements for evaluating the functionality of party-line options. Because the  requirements of evaluating the party-line option are distinct from the requirements of test decks, a separate test deck than that obtained from solving the optimization problem~\eqref{prob:robust} can be constructed for testing party-option functionality.

\section{Solution Reuse and Translation} \label{appx:reuse}

In \S\ref{sec:experiments}, we discussed using our approach to generate test decks for Michigan's November 2022 general election. In this election, the state of Michigan used 6928 ballot styles. In principle, an optimal test deck for each ballot style could be found by solving the optimization problem~\eqref{prob:robust} for the specific ballot style. Solving the optimization problem~\eqref{prob:robust} once for each of the 6928 ballot styles, however, would be computationally time consuming.

In this appendix, we describe two strategies (which we refer to as `solution reuse' and `solution translation') that significantly reduced the computation time that was required for obtaining optimal test decks for all 6928 ballot styles.  The two strategies are based on showing that an optimal solution to the optimization problem~\eqref{prob:robust} for one ballot style will, under certain conditions, be a feasible (and sometimes optimal) solution for the optimization problem~\eqref{prob:robust} for another similar ballot style. By using our two strategies to reuse and translate optimal solutions between similar ballot styles,  we were able to obtain optimal test decks for all 6928 ballot styles despite solving the optimization problem~\eqref{prob:robust} to completion only 376 times.

\subsection{Solution Reuse} \label{appx:reuse:1}

Our first strategy consists of identifying conditions under which the optimal solution of the optimization problem~\eqref{prob:robust} for one ballot style is guaranteed to be an optimal solution for another similar ballot style. To apply this strategy, we first convert each ballot style into what we henceforth refer to as its \emph{normal form}. Converting a ballot style into its normal form entails performing the following two transformations:
\begin{enumerate}[leftmargin=*,align=left]
    \item Combine the ballot style's noncompetitive contests into a single contest as described in \S\ref{appx:improvements:noncompetitive}.
    \item Sort the indices of the contests such that $c < c'$ if $[| \mathcal{N}_c| > | \mathcal{N}_{c'}|]$ or $[| \mathcal{N}_c| = | \mathcal{N}_{c'}| \text{ and } v_c > v_{c'}]$.
\end{enumerate}
Any optimal solution to the optimization problem~\eqref{prob:robust} for the normalized version of a ballot style can be efficiently transformed into an optimal solution for the original style, simply by reversing the translation of candidate indicies on the output $\beta$ variables and separating the combined contest into its constituent components as previously described. This means we can eliminate repeated normalized forms, reducing the number of styles from 6928 to 1812.

\subsection{Solution Translation}
Our second strategy consists of identifying conditions under which an optimal solution for the optimization problem~\eqref{prob:robust} for one ballot style is guaranteed  to be a feasible (but possibly suboptimal) solution to the optimization problem~\eqref{prob:robust} for another similar ballot style.
Lemma~\ref{lem:solution_reduction} specifies the strategy in greater detail. To make the greatest use of Lemma~\ref{lem:solution_reduction}, imagine that each noncompetitive contest has been split so that each candidate has a contest of their own; we can do this without loss of generality as a corollary of Proposition~\ref{prop:noncompetitive}.

\begin{lemma} \label{lem:solution_reduction}
   Consider a ballot style parameterized by $(\mathcal{N}, \mathcal{C}, \{\mathcal{N}_c\}_{c \in \mathcal{C}}, \{v_c\}_{c \in \mathcal{C}})$, let $\bar{\mathcal{C}} \subset \mathcal{C}$ denote a subset of contests, and let $\bar{\mathcal{N}} \triangleq \bigcup_{c \in \bar{\mathcal{C}}} \mathcal{N}_c$ denote the set of candidates in those contests. If $(B, \beta_1, \ldots, \beta_B)$ is an optimal solution for the optimization problem~\eqref{prob:robust} for the ballot style parameterized by $(\mathcal{N}, \mathcal{C}, \{\mathcal{N}_c\}_{c \in \mathcal{C}}, \{v_c\}_{c \in \mathcal{C}})$, then $(B, \beta_1 \setminus \bar{\mathcal{N}}, \ldots, \beta_B \setminus \bar{\mathcal{N}})$ is a feasible solution for the optimization problem~\eqref{prob:robust} for the ballot style parameterized by $(\mathcal{N} \setminus \bar{\mathcal{N}}, \mathcal{C} \setminus \bar{\mathcal{C}}, \{\mathcal{N}_c\}_{c \in \mathcal{C} \setminus \bar{\mathcal{C}}}, \{v_c\}_{c \in \mathcal{C} \setminus \bar{\mathcal{C}} })$.
\end{lemma}
\noindent In words, the above lemma shows that an optimal solution from a `complex' ballot style is guaranteed to be a feasible solution for a `simple' ballot style if (a) the competitive contests in the simpler ballot style are a subset of those in the more complex ballot style and (b) there are at least as many candidates in noncompetitive contests in the complex ballot style as in the simpler ballot style.

To understand the practical significance of Lemma~\ref{lem:solution_reduction}, we recall that each iteration of the cutting plane method from \S\ref{appx:cutting} involves solving the optimization problem~\eqref{prob:robust-subset} to obtain a lower bound on the optimal objective value of the optimization problem~\eqref{prob:robust}. If this lower bound is ever equal to the length of some feasible solution to~\eqref{prob:robust} derived for a more complicated style, we can halt {the cutting plane method} early and translate a solution according to this lemma.

To utilize this second strategy in \S\ref{sec:experiments}, we solved the optimization problem~\eqref{prob:robust} for ballot styles in decreasing order by their number of competitive contests. When two ballot styles had the same number of competitive contests, we solved the one with more candidates in noncompetitive contests first.  When using our cutting plane method for each ballot style, we first identified the shortest feasible solution which is suitable for translation (if any such solution exists) from the ballot styles that were solved previously. Finally, we  halted the cutting plane method early if the lower bound reached the length of that solution.

In practice, this second strategy allowed for early termination of the cutting plane method in a majority of ballot styles. Of the 6928 ballot styles and 1812 distinct normalized forms, we were able to terminate computation early in all but 376 cases. This yielded significant time savings; the average time to generate a test deck for a ballot style which terminates early is on the order of one-tenth of a second, while the average time to generate a test deck for the other 376 styles in on the order of a minute.

\clearpage

\section{Proofs} \label{appx:proofs}

\subsection{Proofs from \S\ref{sec:math:discussion}}

\begin{proof}[Proof of Theorem~\ref{thm:diff_votes}.] Let $B \in \N$, $\beta_1, \ldots, \beta_B \in \mathscr{B}$, and $\sigma \in \Sigma$.

To show the first direction of Theorem~\ref{thm:diff_votes}, suppose that the equality $| \{ b \in \{1,\ldots,B\}: i \in \beta_b \}|  =  |\{ b \in \{1,\ldots,B\}: \sigma(i) \in \beta_b \}|$ holds for all candidates $i \in \mathcal{N}$ and that the inequality $\left| \left \{ \sigma(j) \in \beta_b: j \in \mathcal{N}_c \right \}  \right| \leq v_c$ holds for all contests $c \in \mathcal{C}$  and all ballots $b \in \{1,\ldots,B\}$. In this case, we observe for each contest $c \in \mathcal{C}$ and each candidate $i \in \mathcal{N}_c$ that 
       \begin{align*}
T^\sigma_i({\beta}_1,\ldots,{\beta}_{{B}})& = \sum_{b=1}^B  \mathbb{I}\left \{   \sigma(i) \in {\beta}_b  \textnormal{ and } \left| \left \{ \sigma(j) \in {\beta}_b: j \in \mathcal{N}_c \right \}  \right| \leq v_c \right \}\\
&=  \sum_{b=1}^B  \mathbb{I}\left \{   \sigma(i) \in {\beta}_b  \right \}\\
&=  \left| \left \{ b \in \{1,\ldots,B \}: \sigma(i) \in \beta_b \right \} \right| \\
&=  \left| \left \{ b \in \{1,\ldots,B \}: i \in \beta_b \right \} \right| \\
&=  T^*_{i}({\beta}_1,\ldots,{\beta}_{{B}}). 
       \end{align*} 
       The first equality is the definition of  $T^\sigma_i(\cdot)$ from \S\ref{sec:math:swap}. The second equality follows from the supposition that  the inequality $\left| \left \{ \sigma(j) \in \beta_b: j \in \mathcal{N}_{c'} \right \}  \right| \leq v_{c'}$ holds for all  contests $c' \in \mathcal{C}$ and all ballots $b \in \{1,\ldots,B\}$. 
       The third equality follows from algebra. The fourth equality follows from the supposition that  the equality $| \{ b \in \{1,\ldots,B\}: i' \in \beta_b \}|  =  |\{ b \in \{1,\ldots,B\}: \sigma(i') \in \beta_b \}|$ holds for all candidates $i' \in \mathcal{N}$. The fifth equality follows from Remark~\ref{remark:overvotes} and from the fact that $\beta_1,\ldots,\beta_B \in \mathscr{B}$. Because we have shown that the equality $T^\sigma_i(\beta_1, \ldots, \beta_B) = T^*_i(\beta_1, \ldots, \beta_B)$ holds for all candidates $i \in \mathcal{N}$, our  proof of the first direction of Theorem~\ref{thm:diff_votes} is complete. 

To show the other direction of Theorem~\ref{thm:diff_votes}, suppose that the equality $T^*(\beta_1,\ldots,\beta_B) = T^\sigma(\beta_1,\ldots,\beta_B)$ holds. In this case, we observe that
\begin{align*}
      & \sum_{i \in \mathcal{N}} \sum_{b=1}^B  \mathbb{I}\left \{   \sigma(i) \in {\beta}_b  \textnormal{ and } \left| \left \{ \sigma(j) \in {\beta}_b: j \in \mathcal{N}_c \right \}  \right| \leq v_c \right \}\\
 &=  \sum_{i \in \mathcal{N}} T^\sigma_i(\beta_1, \ldots, \beta_B) \\
    &= \sum_{i \in \mathcal{N}} T^*_i(\beta_1, \ldots, \beta_B)\\
    &= \sum_{i \in \mathcal{N}} \sum_{b=1}^B \mathbb{I} \left \{ i \in \beta_b \textnormal{ and }\left| \mathcal{N}_c \cap \beta_b \right| \le v_c \right \} \\
    &= \sum_{i \in \mathcal{N}} \sum_{b=1}^B \mathbb{I} \left \{ i \in \beta_b \right \} \\
    &= \sum_{i \in \mathcal{N}} \sum_{b=1}^B \mathbb{I} \left \{ \sigma(i) \in \beta_b \right \}. 
\end{align*}
The first equality is the definition of $T^\sigma_i(\cdot)$. The second equality follows from the supposition that the equality $T^*(\beta_1,\ldots,\beta_B) = T^\sigma(\beta_1,\ldots,\beta_B)$ holds.  The third equality is the definition of $T^*_i(\cdot)$. The fourth equality follows from the fact that the ballots satisfy  $\beta_1, \ldots, \beta_B \in \mathscr{B}$. The fifth equality follows from the fact that $\sigma$ is a bijection. Combining the above equalities, we conclude that the inequality $\left| \left \{ \sigma(j) \in {\beta}_b: j \in \mathcal{N}_c \right \}  \right| \leq v_c $ must hold for all contests $c \in \mathcal{C}$ and all ballots $b \in \{1,\ldots,B\}$.

It remains for us to show that the equality $| \{ b \in \{1,\ldots,B\}: i \in \beta_b \}|  =  |\{ b \in \{1,\ldots,B\}: \sigma(i) \in \beta_b \}|$ holds for all candidates $i \in \mathcal{N}$. Indeed, we observe for each contest $c \in \mathcal{C}$ and each candidate $i \in \mathcal{N}_c$ that  
\begin{align*}
  | \{ b \in \{1,\ldots,B\}: \sigma(i) \in \beta_b \}|  
      &= \sum_{b=1}^B \mathbb{I} \left \{ \sigma(i) \in \beta_b \right \} \\
    &= \sum_{b=1}^B \mathbb{I} \left \{ \sigma(i) \in \beta_b \textnormal{ and } \left| \left \{ \sigma(j) \in {\beta}_b: j \in \mathcal{N}_c \right \}  \right| \leq v_c \right \} \\
    &= T^\sigma_i(\beta_1, \ldots, \beta_B) \\
    &= T^*_i(\beta_1, \ldots, \beta_B)\\
    &=  | \{ b \in \{1,\ldots,B\}: i \in \beta_b \}| . 
\end{align*}
The first equality follows from algebra. The second equality follows from our prior conclusion that the inequality $\left| \left \{ \sigma(j) \in {\beta}_b: j \in \mathcal{N}_c \right \}  \right| \leq v_c$ holds for all ballots $b \in \{1, \ldots, B\}$. The third equality is the definition of $T^\sigma_i(\cdot)$. The fourth equality follows from the supposition that $T^\sigma(\beta_1, \ldots, \beta_B) = T^*(\beta_1, \ldots, \beta_B)$. The fifth equality follows from Remark~\ref{remark:overvotes} and from the fact that $\beta_1,\ldots,\beta_B \in \mathscr{B}$.  Because we have shown that the equality $| \{ b \in \{1,\ldots,B\}: i \in \beta_b \}|  =  |\{ b \in \{1,\ldots,B\}: \sigma(i) \in \beta_b \}|$  holds for all candidates $i \in \mathcal{N}$, our proof of the other direction of Theorem~\ref{thm:diff_votes} is complete.
\end{proof}

\subsection{Proofs from \S\ref{appx:improvements:reducingdecisions}}

\begin{proof}[Proof of Lemma~\ref{lem:decrease_p}.]
Our proof consists of proving the contrapositive of the desired result. Indeed, consider any mapping $\sigma \in \widehat{\Sigma}$, contest $c \in \mathcal{C}$, and ballot $b \in \mathcal{B}$. Moreover, suppose that there exists a feasible solution of the mixed-integer linear optimization problem~\eqref{prob:complete_transposition} that satisfies the equality $p^\sigma_{b,c} = 0$. In this case, we observe that
\begin{align*}
    v_c + 1 &\le \sum_{i \in \mathcal{N}_c} \beta_{b,\sigma(i)}\\
    &= \sum_{c' \in \mathcal{C}} \left( \sum_{i \in \mathcal{N}_c: \sigma(i) \in \mathcal{N}_{c'}} \beta_{b,\sigma(i)} \right)\\
    &= \sum_{c' \in \mathcal{C}} \min \left \{  \sum_{i \in \mathcal{N}_c: \sigma(i) \in \mathcal{N}_{c'}} \beta_{b,\sigma(i)}, v_{c'}\right \}\\
    &\le \sum_{c' \in \mathcal{C}} \min \left \{  \left| \left \{i \in \mathcal{N}_c: \sigma(i) \in \mathcal{N}_{c'} \right \} \right|, v_{c'}\right \}. 
\end{align*}
Indeed, the first inequality follows from the fact that $p^\sigma_{b,c} = 0$ and constraint~\eqref{prob:complete_transposition:p}. The first equality follows from algebra. The second equality follows from constraint~\eqref{prob:complete_transposition:beta_feas}, which implies for each contest $c' \in \mathcal{C}$ that $\sum_{i \in \mathcal{N}_c: \sigma(i) \in \mathcal{N}_{c'}} \beta_{b,\sigma(i)} \le \sum_{i \in \mathcal{N}_{c'}} \beta_{b,i} \le v_{c'} $. The second inequality follows from algebra. 
\end{proof} 

\subsection{Proofs from \S\ref{appx:improvements:symmetry}}

\begin{proof}[Proof of Lemma~\ref{lem:symmetry_withincontest}.]
Consider any feasible solution of the optimization problem~\eqref{prob:robust}, and let that feasible solution be denoted by $(B,\beta_1,\ldots,\beta_B)$. Suppose for the sake of developing a contradiction that there exists a contest $c \in \mathcal{C}$ and a pair of candidates $i,j \in \mathcal{N}_c$ such that $i \neq j$ and  $| \{ b \in \{1,\ldots,B\}: i \in \beta_b \} |  = | \{ b \in \{1,\ldots,B\}: j \in \beta_b \} |$. In what follows, we will make use of  a non-identity bijection $\sigma \in \Sigma$ constructed for each candidate $i' \in \mathcal{N}$ as follows: 
\begin{align}
    \sigma \left(i' \right) &\triangleq \begin{cases}
    i',&\text{if } i' \in \mathcal{N} \setminus \{i,j \},\\
    j,&\text{if } i' = i,\\
    i,&\text{if } i' = j. 
    \end{cases} \label{line:symmetry_withincontest:defn_sigma}
\end{align}

 The remainder of the proof of Lemma~\ref{lem:symmetry_withincontest} consists of showing that the equality $T^\sigma(\beta_1,\ldots,\beta_B) = T^*(\beta_1,\ldots,\beta_B)$ is satisfied. 
 Indeed, we observe for each  candidate $i' \in \mathcal{N}$  that 
 \begin{align*}
 | \{ b \in \{1,\ldots,B\}: \sigma(i') \in \beta_b \} |   &= \begin{cases}
      | \{ b \in \{1,\ldots,B\}: i' \in \beta_b \} |,&\text{if } i' \in \mathcal{N} \setminus \{i,j \},\\
     | \{ b \in \{1,\ldots,B\}: j \in \beta_b \} |,&\text{if } i' =i,\\
    | \{ b \in \{1,\ldots,B\}: i \in \beta_b \} |,&\text{if } i' =j
 \end{cases}\\
 &=  | \{ b \in \{1,\ldots,B\}: i' \in \beta_b \} |, 
 \end{align*}
 where the first equality follows from our construction  of $\sigma \in \Sigma$ on line~\eqref{line:symmetry_withincontest:defn_sigma}  and the second equality follows from the supposition that $| \{ b \in \{1,\ldots,B\}: i \in \beta_b \} |  = | \{ b \in \{1,\ldots,B\}: j \in \beta_b \} |$. Moreover, we observe for each $b \in \{1,\ldots,B\}$ and each contest $c' \in \mathcal{C}$ that 
\begin{align}
&\left| \left \{ \sigma\left(i' \right) \in \beta_b: i' \in \mathcal{N}_{c'} \right \} \right| \notag \\ 
&= \begin{cases}
    \left| \left \{ i' \in \beta_b: i' \in \mathcal{N}_c \setminus \{i ,j\} \right \} \cup \{j \} \right|,&\text{if } c' = c, \; i \in \beta_b, \textnormal{ and } j \notin \beta_b,\\
     \left| \left \{ i' \in \beta_b: i' \in \mathcal{N}_c \setminus \{i,j \} \right \} \cup \{i \}\right| ,&\text{if } c' = c, \; i \notin \beta_b, \textnormal{ and } j \in \beta_b,\\
\left|  \left \{ i' \in \beta_b: i' \in \mathcal{N}_{c'} \right \} \right| ,&\text{otherwise}
\end{cases} \notag  \\
&= \begin{cases}
    \left| \left \{ i' \in \beta_b: i' \in \mathcal{N}_c \setminus \{i ,j\} \right \}  \right| + 1,&\text{if } c' = c, \; i \in \beta_b, \textnormal{ and } j \notin \beta_b,\\
     \left| \left \{ i' \in \beta_b: i' \in \mathcal{N}_c \setminus \{i,j \} \right \} \right| + 1 ,&\text{if } c' = c, \; i \notin \beta_b, \textnormal{ and } j \in \beta_b,\\
\left|  \left \{ i' \in \beta_b: i' \in \mathcal{N}_{c'} \right \} \right| ,&\text{otherwise}
\end{cases} \notag  \\
&= \left|  \left \{ i' \in \beta_b: i' \in \mathcal{N}_{c'} \right \}  \right| \notag \\ 
&\le v_{c'},\notag  
\end{align}
where the first equality follows from our construction of $\sigma \in \Sigma$ on line~\eqref{line:symmetry_withincontest:defn_sigma}, the second equality follows from algebra,  the third equality follows  from the fact that $i,j \in \mathcal{N}_{c}$, and the inequality follows from the fact that $\beta_1,\ldots,\beta_B \in \mathscr{B}$. Combining the above analysis with Theorem~\ref{thm:diff_votes} from \S\ref{sec:math:discussion}, we conclude that the  equality $T^\sigma(\beta_1,\ldots,\beta_B) = T^*(\beta_1,\ldots,\beta_B)$ is satisfied. 

Because we have shown that there exists a non-identity bijection $\sigma \in \Sigma$ that satisfies the equality $T^\sigma(\beta_1,\ldots,\beta_B) = T^*(\beta_1,\ldots,\beta_B)$, we have obtained a  contradiction with the fact the that $(B,\beta_1,\ldots,\beta_B)$ is a feasible solution of the optimization problem~\eqref{prob:robust}. Our proof of Lemma~\ref{lem:symmetry_withincontest} is thus complete.  \end{proof}

 \begin{proof}[Proof of Proposition~\ref{prop:symmetry_withincontest}.]
Consider any optimal solution of the optimization problem~\eqref{prob:robust}, and let that optimal solution be denoted by $(\beta_1,\ldots,\beta_B)$. We henceforth assume without loss of generality that $| \{ b \in \{1,\ldots,B\}: i \in \beta_b \} |  \le | \{ b \in \{1,\ldots,B\}: j \in \beta_b \} |$ for all candidates $i < j$  that appear in the same contest. To see why this assumption is without loss of generality, suppose for the sake of argument that this assumption was not true. In that case, for each contest $c \in \mathcal{C}$, let $\pi_c: \mathcal{N}_c \to \mathcal{N}_c$ be a bijection that satisfies $| \{ b \in \{1,\ldots,B\}: \pi_c(i) \in \beta_b \} |  \le | \{ b \in \{1,\ldots,B\}: \pi_c(j) \in \beta_b \} |$  for each pair of candidates $i,j \in \mathcal{N}_c$ that satisfies $i < j$. Hence, by replacing the index of each candidate $i \in \mathcal{N}_c$ with the index $\pi_c(i)$, we conclude that the assumption that $T^*_i(\beta_1,\ldots,\beta_B) \le T^*_j(\beta_1,\ldots,\beta_B)$ for all candidates $i < j$  that appear in the same contest can be made without loss of generality. Combining that assumption with Lemma~\ref{lem:symmetry_withincontest}, our proof of Proposition~\ref{prop:symmetry_withincontest} is complete. 
\end{proof}

\subsection{Proofs from \S\ref{appx:improvements:n_choose_k}}

\begin{proof}[Proof of Lemma~\ref{lem:symmetry_acrosscontests}.]
Consider any feasible solution of the optimization problem~\eqref{prob:robust}, and let that feasible solution be denoted by $(B,\beta_1,\ldots,\beta_B)$. Suppose for the sake of developing a contradiction that there exist two contests $c < c'$ that satisfy $c \equiv c'$ and  satisfy 
\begin{gather*}
\left| \left \{ b \in \{1,\ldots,B\}: \mathcal{N}^{| \mathcal{N}_{c}|}_c \in \beta_b \right \} \right| = \left| \left \{ b \in \{1,\ldots,B\}: \mathcal{N}^{| \mathcal{N}_{c}|}_{c'} \in \beta_b \right \} \right|  \\
\vdots \\
\left| \left \{ b \in \{1,\ldots,B\}: \mathcal{N}^{1}_c \in \beta_b \right \} \right| = \left| \left \{ b \in \{1,\ldots,B\}: \mathcal{N}^{1}_{c'} \in \beta_b \right \} \right|.
\end{gather*}
In what follows, we will make use of a non-identity bijection $\sigma \in \Sigma$ that is defined for each candidate $i \in \mathcal{N}$ as  follows:
\begin{align}
    \sigma \left(i \right) &\triangleq \begin{cases}
    \mathcal{N}^k_c,&\text{if there exists } k \in  \{1,\ldots,|\mathcal{N}_{c}| \} \textnormal{ such that } i = \mathcal{N}^k_{c'},\\
    \mathcal{N}^k_{c'},&\text{if there exists } k \in  \{1,\ldots,|\mathcal{N}_{c}| \} \textnormal{ such that } i = \mathcal{N}^k_{c},\\
    i,&\text{otherwise}. 
    \end{cases} \label{line:symmetry_acrosscontests:defn_sigma}
\end{align}
We observe by construction that the bijection $\sigma$ swaps the targets of candidates $\mathcal{N}^k_c$ and $\mathcal{N}^k_{c'}$ for each $k \in \{1,\ldots,| \mathcal{N}_c| \}$. 

 The remainder of the proof of Lemma~\ref{lem:symmetry_acrosscontests} consists of showing that the equality $T^\sigma(\beta_1,\ldots,\beta_B) = T^*(\beta_1,\ldots,\beta_B)$ is satisfied. 
Indeed, we observe for each candidate $i \in \mathcal{N}$ that  
 \begin{align*}
& | \{ b \in \{1,\ldots,B\}: \sigma(i) \in \beta_b \} |  \\
 &= \begin{cases}
     | \{ b \in \{1,\ldots,B\}: \sigma(\mathcal{N}^k_{c'}) \in \beta_b \} |,&\text{if there exists } k \in  \{1,\ldots,|\mathcal{N}_{c}| \} \textnormal{ such that } i = \mathcal{N}^k_{c'},\\
     | \{ b \in \{1,\ldots,B\}:\sigma( \mathcal{N}^k_{c}) \in \beta_b \} |,&\text{if there exists } k \in  \{1,\ldots,|\mathcal{N}_{c}| \} \textnormal{ such that } i = \mathcal{N}^k_{c},\\
      | \{ b \in \{1,\ldots,B\}: \sigma(i) \in \beta_b \} |,&\text{otherwise}
 \end{cases}\\
  &= \begin{cases}
     | \{ b \in \{1,\ldots,B\}: \mathcal{N}^k_c \in \beta_b \} |,&\text{if there exists } k \in  \{1,\ldots,|\mathcal{N}_{c}| \} \textnormal{ such that } i = \mathcal{N}^k_{c'},\\
     | \{ b \in \{1,\ldots,B\}: \mathcal{N}^k_{c'} \in \beta_b \} |,&\text{if there exists } k \in  \{1,\ldots,|\mathcal{N}_{c}| \} \textnormal{ such that } i = \mathcal{N}^k_{c},\\
      | \{ b \in \{1,\ldots,B\}: i \in \beta_b \} |,&\text{otherwise}
 \end{cases}\\
 &=  | \{ b \in \{1,\ldots,B\}: i \in \beta_b \} |, 
 \end{align*}
 where the first equality follows from  algebra and from the fact that $| \mathcal{N}_c| = |\mathcal{N}_{c'}|$, 
the second equality follows from our construction  of $\sigma \in \Sigma$ on line~\eqref{line:symmetry_acrosscontests:defn_sigma}, and the third equality follows from the supposition that the equality  $|  \{ b \in \{1,\ldots,B\}: \mathcal{N}^{k}_c \in \beta_b  \} | = |  \{ b \in \{1,\ldots,B\}: \mathcal{N}^{k}_{c'} \in \beta_b  \} |$ holds for all $k \in \{1,\ldots,| \mathcal{N}_c| \}$. Moreover, we observe for each $b \in \{1,\ldots,B\}$ and each contest $c'' \in \mathcal{C}$ that  
\begin{align*}
&\left| \left \{ \sigma\left(i \right) \in \beta_b: i \in \mathcal{N}_{c''} \right \} \right| \notag \\ 
&= \begin{cases}
    \left| \left \{ \sigma \left( \mathcal{N}^k_{c'} \right): k \in \{1,\ldots,| \mathcal{N}_c| \} \text{ and } \mathcal{N}^k_{c'} \in \beta_b \right \} \right|,&\text{if } c'' = c',\\
    \left| \left \{ \sigma \left( \mathcal{N}^k_{c} \right): k \in \{1,\ldots,| \mathcal{N}_c| \} \text{ and } \mathcal{N}^k_{c} \in \beta_b \right \} \right|,&\text{if } c'' = c,\\
\left|  \left \{ \sigma(i) \in \beta_b: i \in \mathcal{N}_{c''} \right \} \right| ,&\text{otherwise}
\end{cases} \notag  \\
&= \begin{cases}
    \left| \left \{ \mathcal{N}^k_{c}: k \in \{1,\ldots,| \mathcal{N}_c| \} \text{ and } \mathcal{N}^k_{c'} \in \beta_b \right \} \right|,&\text{if } c'' = c',\\
    \left| \left \{ \mathcal{N}^k_{c'}: k \in \{1,\ldots,| \mathcal{N}_c| \} \text{ and } \mathcal{N}^k_{c} \in \beta_b \right \} \right|,&\text{if } c'' = c,\\
\left|  \left \{ i \in \beta_b: i \in \mathcal{N}_{c''} \right \} \right| ,&\text{otherwise}
\end{cases} \notag  \\
&\le \begin{cases}
   v_c,&\text{if } c'' = c',\\
    v_{c'},&\text{if } c'' = c,\\
v_{c''},&\text{otherwise}
\end{cases}\\
&= v_{c''},
\end{align*}
where the first equality follows from algebra and from the fact that $| \mathcal{N}_c| = | \mathcal{N}_{c'}|$,   the second equality follows from our construction of $\sigma \in \Sigma$ on line~\eqref{line:symmetry_acrosscontests:defn_sigma}, the inequality follows from the fact that $\beta_1,\ldots,\beta_B \in \mathscr{B}$, and the third equality follows  from the fact that $v_c = v_{c'}$. Combining the above analysis with Theorem~\ref{thm:diff_votes} from \S\ref{sec:math:discussion}, we conclude that the  equality $T^\sigma(\beta_1,\ldots,\beta_B) = T^*(\beta_1,\ldots,\beta_B)$ is satisfied. 

Because we have shown that there exists a non-identity bijection $\sigma \in \Sigma$ that satisfies the equality $T^\sigma(\beta_1,\ldots,\beta_B) = T^*(\beta_1,\ldots,\beta_B)$, we have obtained a  contradiction with the fact  that $(B,\beta_1,\ldots,\beta_B)$ is a feasible solution of the optimization problem~\eqref{prob:robust}. Our proof of Lemma~\ref{lem:symmetry_acrosscontests} is thus complete. 
\end{proof}

\begin{proof}[Proof of Proposition~\ref{prop:symmetry_acrosscontests}.]
Consider any optimal solution of the optimization problem~\eqref{prob:robust}, and let that optimal solution be denoted by $(B,\beta_1,\ldots,\beta_B)$. We henceforth assume without loss of generality that for all contests $c < c'$ that satisfy $c \equiv c'$, there exists $k \in \{1,\ldots,| \mathcal{N}_c| \}$ that satisfies
\begin{gather*}
\left| \left \{ b \in \{1,\ldots,B\}: \mathcal{N}^{| \mathcal{N}_{c}|}_c \in \beta_b \right \} \right| = \left| \left \{ b \in \{1,\ldots,B\}: \mathcal{N}^{| \mathcal{N}_{c}|}_{c'} \in \beta_b \right \} \right|  \\
\vdots \\
\left| \left \{ b \in \{1,\ldots,B\}: \mathcal{N}^{k+1}_c \in \beta_b \right \} \right| = \left| \left \{ b \in \{1,\ldots,B\}: \mathcal{N}^{k+1}_{c'} \in \beta_b \right \} \right|  \\
\left| \left \{ b \in \{1,\ldots,B\}: \mathcal{N}^{k}_c \in \beta_b \right \} \right| \le \left| \left \{ b \in \{1,\ldots,B\}: \mathcal{N}^{k}_{c'} \in \beta_b \right \} \right|.
\end{gather*}
This assumption is without loss of generality because the indices of contests that are equivalent can always be permuted to ensure that the vectors $(| \{ b \in \{1,\ldots,B\}: \mathcal{N}^{1}_c \in \beta_b \} |,\ldots, | \{ b \in \{1,\ldots,B\}: \mathcal{N}^{| \mathcal{N}_c|}_c \in \beta_b \} |)$ are lexicographically ordered. Combining that assumption with Lemma~\ref{lem:symmetry_acrosscontests}, our proof of Proposition~\ref{prop:symmetry_acrosscontests} is complete.
\end{proof}

\subsection{Proofs from \S\ref{appx:improvements:swap}}

\begin{proof}[Proof of Theorem~\ref{thm:minimal}.]
Let $\sigma \in \Sigma$ denote a feasible solution of the optimization problem~\eqref{prob:oracle}. For each $k \in \{1,\ldots,K^\sigma \}$, let $\sigma_k: \mathcal{N} \to \mathcal{N}$ be defined for each $c \in \mathcal{C}$ and $i \in \mathcal{N}_c$ by
\begin{align*}
\sigma_k(i) &\triangleq \begin{cases}
    \sigma(i),&\textnormal{if } c \in \mathscr{K}^\sigma_k,\\
    i,&\textnormal{if } c \notin \mathscr{K}^\sigma_k. 
\end{cases}
\end{align*}

We begin by showing that each of the functions $\sigma_k: \mathcal{N} \to \mathcal{N}$ is a non-identity bijection. Indeed, we observe for each contest $c \notin \mathscr{K}^\sigma_k$ and each candidate $i \in \mathcal{N}_c$ that the equality $\sigma_k^{-1}(i) = i$ holds. Moreover, for each  contest $c \in \mathscr{K}^\sigma_k$ and each candidate $i \in \mathcal{N}_c$, it follows from our construction of the undirected graph $\mathscr{G}^\sigma \equiv (\mathscr{V}^\sigma,\mathscr{E}^\sigma)$, from the definition of a connected component, and from the inclusion $\sigma \in \Sigma$ that there exists a contest $c' \in \mathscr{K}^\sigma_k$ and a candidate $i' \in \mathcal{N}_{c'}$ that satisfies $i' \neq i$ and  $\sigma_k^{-1}(i) = i'$. Therefore, we have shown for all candidates $i \in \mathcal{N}$ that there exists $i' \in \mathcal{N}$ that satisfies the equality $\sigma^{-1}_k(i') = i$, which concludes our proof that $\sigma_k$ is a bijection. Moreover, since $\mathscr{K}^\sigma_k$ is nonempty, we have argued that there must exist candidates $i' \neq i$ that satisfy $\sigma_k^{-1}(i) = i'$. Therefore,  we conclude that $\sigma_k$ is a non-identity bijection. 

We next show that each of the functions  $\sigma_k: \mathcal{N} \to \mathcal{N}$ is a feasible solution of the optimization problem~\eqref{prob:oracle}. Indeed, we have already shown that $\sigma_k$ is a non-identity bijection, which implies that $\sigma_k \in \Sigma$. Moreover, for each contest $c \in \mathcal{C}$ and each candidate $i \in \mathcal{N}_c$, we observe that
\begin{align*}
&T^{\sigma_k}_i \left (\beta_1,\ldots,\beta_{B} \right)\\
&= \sum_{b=1}^B   \mathbb{I}\left \{   \sigma_k(i) \in \beta_b  \textnormal{ and } \left| \left \{ \sigma_k(j) \in \beta_b: j \in \mathcal{N}_c \right \}  \right| \leq v_c \right \} \\
&= \begin{cases}
    \sum_{b=1}^B \mathbb{I}\left \{   \sigma(i) \in \beta_b  \textnormal{ and } \left| \left \{ \sigma(j) \in \beta_b: j \in \mathcal{N}_c \right \}  \right| \leq v_c \right \},&\text{if } c \in \mathscr{K}^\sigma_k,\\
    \sum_{b=1}^B \mathbb{I}\left \{   i \in \beta_b  \textnormal{ and } \left| \left \{ j \in \beta_b: j \in \mathcal{N}_c \right \}  \right| \leq v_c \right \},&\text{if } c \notin \mathscr{K}^\sigma_k
\end{cases} \\
&= \begin{cases}
    T^{\sigma}_i \left (\beta_1,\ldots,\beta_{B} \right),&\text{if } c \in \mathscr{K}^\sigma_k,\\
    T^*_i \left (\beta_1,\ldots,\beta_{B} \right),&\text{if } c \notin \mathscr{K}^\sigma_k
\end{cases}\\
&= T^*_i \left (\beta_1,\ldots,\beta_{B} \right).
\end{align*}
Indeed, the first equality is the definition of $T^{\sigma_k}_i(\beta_1,\ldots,\beta_{B})$. The  second equality follows from the definition of $\sigma_k$. The third equality follows from the definitions of $T^{\sigma}_i(\beta_1,\ldots,\beta_{B})$ and $T^{*}_i(\beta_1,\ldots,\beta_{B})$. The fourth equality follows from the fact that $\sigma \in \Sigma$ is a feasible solution of the optimization problem~\eqref{prob:oracle}, which implies that the equality $T^{\sigma}_i \left (\beta_1,\ldots,\beta_{B} \right) = T^*_i \left (\beta_1,\ldots,\beta_{B} \right)$ holds for all candidates $i \in \mathcal{N}$. Our proof that $\sigma_k$ is a feasible solution of the optimization problem~\eqref{prob:oracle} is thus complete.  

As our final step, we show that line~\eqref{line:sigma_breakdown} holds. Indeed, we observe for each $k \in \{1,\ldots,K^\sigma \}$, $B \in \N$, $(\beta_1,\ldots,\beta_B) \in \mathscr{B}^B$, $c \in \mathcal{C}$, and $i \in \mathcal{N}_c$ that 
\begin{align}
    &T^{\sigma_k}_i(\beta_1,\ldots,\beta_B) \notag \\
    &= \sum_{b =1}^B  \mathbb{I}\left \{   \sigma_k(i) \in {\beta}_b  \textnormal{ and } \left| \left \{ \sigma_k(j) \in {\beta}_b: j \in \mathcal{N}_c \right \}  \right| \leq v_c \right \} \notag  \\
&= \begin{cases}
    \sum_{b=1}^B  \mathbb{I}\left \{   \sigma(i) \in {\beta}_b  \textnormal{ and } \left| \left \{ \sigma(j) \in {\beta}_b: j \in \mathcal{N}_c \right \}  \right| \leq v_c \right \},&\text{if } c \in \mathscr{K}^\sigma_k,\\
    \sum_{b =1}^B  \mathbb{I}\left \{   i \in {\beta}_b  \textnormal{ and } \left| \left \{ j \in {\beta}_b: j \in \mathcal{N}_c \right \}  \right| \leq v_c \right \},&\text{if } c \notin \mathscr{K}^\sigma_k 
\end{cases}  \notag \\
    &= \begin{cases}
    T^{\sigma}_i \left ({\beta}_1,\ldots,{\beta}_{{B}} \right),&\text{if } c \in \mathscr{K}^\sigma_k,\\
    T^*_i \left ({\beta}_1,\ldots,{\beta}_{{B}} \right),&\text{if } c \notin \mathscr{K}^\sigma_k,
\end{cases} \label{line:breakdown_sigma_k}
\end{align}
where the first equality is the definition of $T^{\sigma_k}_i({\beta}_1,\ldots,{\beta}_{{B}})$, the  second equality follows from the definition of $\sigma_k$, and the third equality follows from the definitions of $T^{\sigma}_i(\beta_1,\ldots,\beta_{B})$ and $T^{*}_i(\beta_1,\ldots,\beta_{B})$. Therefore, we observe that 
\begin{align*}
    &\bigcup_{k=1}^{K^\sigma} \mathscr{F} \left( \widehat{\Sigma} \cup \{ \sigma_k \} \right) \\
    &=  \bigcup_{k=1}^{K^\sigma}  \left(  \mathscr{F} \left(\widehat{\Sigma} \right) \cap \bigcup_{B \in \N} \left \{(\beta_1,\ldots,\beta_B) \in \mathscr{B}^B: T^{\sigma_k}(\beta_1,\ldots,\beta_B) \neq   T^*(\beta_1,\ldots,\beta_B) \right \} \right) \\
 &=    \mathscr{F} \left(\widehat{\Sigma} \right) \cap \bigcup_{B \in \N} \bigcup_{k=1}^{K^\sigma}  \left \{(\beta_1,\ldots,\beta_B) \in \mathscr{B}^B: T^{\sigma_k}(\beta_1,\ldots,\beta_B) \neq   T^*(\beta_1,\ldots,\beta_B) \right \} \\
  &=    \mathscr{F} \left(\widehat{\Sigma} \right) \cap \bigcup_{B \in \N} \bigcup_{k=1}^{K^\sigma} \bigcup_{c \in \mathcal{C}} \bigcup_{i \in \mathcal{N}_c} \left \{(\beta_1,\ldots,\beta_B) \in \mathscr{B}^B: T^{\sigma_k}_i(\beta_1,\ldots,\beta_B) \neq   T^*_i(\beta_1,\ldots,\beta_B) \right \} \\
    &=    \mathscr{F} \left(\widehat{\Sigma} \right) \cap \bigcup_{B \in \N} \bigcup_{k=1}^{K^\sigma} \bigcup_{c \in \mathscr{K}^\sigma_k} \bigcup_{i \in \mathcal{N}_c} \left \{(\beta_1,\ldots,\beta_B) \in \mathscr{B}^B: T^{\sigma}_i(\beta_1,\ldots,\beta_B) \neq   T^*_i(\beta_1,\ldots,\beta_B) \right \} \\
        &=    \mathscr{F} \left(\widehat{\Sigma} \right) \cap \bigcup_{B \in \N} \bigcup_{c \in \mathscr{V}^\sigma} \bigcup_{i \in \mathcal{N}_c} \left \{(\beta_1,\ldots,\beta_B) \in \mathscr{B}^B: T^{\sigma}_i(\beta_1,\ldots,\beta_B) \neq   T^*_i(\beta_1,\ldots,\beta_B) \right \} \\
&=    \mathscr{F} \left(\widehat{\Sigma} \right) \cap \bigcup_{B \in \N} \bigcup_{c \in \mathcal{C}} \bigcup_{i \in \mathcal{N}_c} \left \{(\beta_1,\ldots,\beta_B) \in \mathscr{B}^B: T^{\sigma}_i(\beta_1,\ldots,\beta_B) \neq   T^*_i(\beta_1,\ldots,\beta_B) \right \} \\
&=    \mathscr{F} \left(\widehat{\Sigma} \right) \cap \bigcup_{B \in \N}  \left \{(\beta_1,\ldots,\beta_B) \in \mathscr{B}^B: T^{\sigma}(\beta_1,\ldots,\beta_B) \neq   T^*(\beta_1,\ldots,\beta_B) \right \} \\
&=    \mathscr{F} \left(\widehat{\Sigma} \cup \{ \sigma \} \right).
\end{align*}
Indeed, the first equality follows from the definition of the optimization problem~\eqref{prob:robust-subset}. The second  and third equalities follow from algebra. The fourth equality follows from line~\eqref{line:breakdown_sigma_k}. The fifth equality follows from the fact that $\mathscr{K}^\sigma_1,\ldots,\mathscr{K}^\sigma_{K^\sigma}$ are the connected components of the undirected graph  $\mathscr{G}^\sigma \equiv (\mathscr{V}^\sigma,\mathscr{E}^\sigma)$, which implies that $\mathscr{K}^\sigma_1 \cup \cdots \cup \mathscr{K}^\sigma_{K^\sigma} = \mathscr{V}^\sigma $. The sixth equality follows from the definition of $\mathscr{V}^\sigma$, which implies that the equality $\sigma(i) = i$ is satisfied  for all candidates $i \in \cup_{c \in \mathcal{C} \setminus \mathscr{V}^\sigma} \mathcal{N}_c$. The seventh and eighth equalities follow from algebra. Our proof of Theorem~\ref{thm:minimal} is thus complete. 
\end{proof}

\begin{proof}[Proof of Theorem~\ref{thm:minimal_reform}.]
Consider any optimal solution $x \in \{0,1\}^{\mathcal{N} \times \mathcal{N}}$ of the mixed-integer linear optimization problem~\eqref{prob:oracle_mip_modified}. Let $\sigma: \mathcal{N} \to \mathcal{N}$ be the function that satisfies the equality $\sigma(i) = j$ if and only if $x_{i,j} = 1$ for all $i,j \in \mathcal{N}$. In this case, it follows from the discussion in \S\ref{appx:mip:oracle} that $\sigma$ is a non-identity bijection that is a feasible solution for the optimization problem~\eqref{prob:oracle}. 

Suppose for the sake of developing a contradiction that the number of connected components of the undirected graph  $\mathscr{G}^\sigma \equiv (\mathscr{V}^\sigma,\mathscr{E}^\sigma)$  satisfies $K^{\sigma} \ge 2$. For each $k \in \{1,\ldots,K^\sigma \}$, let $\sigma_k: \mathcal{N} \to \mathcal{N}$ be defined for each $c \in \mathcal{C}$ and $i \in \mathcal{N}_c$ by
\begin{align*}
\sigma_k(i) &\triangleq \begin{cases}
    \sigma(i),&\textnormal{if } c \in \mathscr{K}^\sigma_k,\\
    i,&\textnormal{if } c \notin \mathscr{K}^\sigma_k,
\end{cases}
\end{align*}
where $\mathscr{K}^\sigma_1,\ldots,\mathscr{K}^\sigma_{K^\sigma} \subseteq \mathscr{V}^\sigma$ denote the connected components of the undirected graph  $\mathscr{G}^\sigma \equiv (\mathscr{V}^\sigma,\mathscr{E}^\sigma)$
In this case,  it follows from Theorem~\ref{thm:minimal} that 
 $\sigma_1,\ldots,\sigma_{K^\sigma}$ are feasible solutions of  the optimization problem~\eqref{prob:oracle}. 
Because $\sigma_1 \in \Sigma$ is a feasible solution of the optimization problem~\eqref{prob:oracle}, we observe that a feasible solution for the mixed-integer linear optimization problem~\eqref{prob:oracle_mip_modified} is given by $\bar{x} \in \{0,1\}^{\mathcal{N} \times \mathcal{N}}$, which is defined by $\bar{x}_{i,j} \triangleq \mathbb{I} \left \{ \sigma_1(i) = j \right \}$ for all $i,j \in \mathcal{N}$. We observe that
\begin{align*}
    \sum_{i,j \in \mathcal{N}: i \neq j} \bar{x}_{i,j}  &=     \left| \left \{ i \in \mathcal{N}: \sigma_1(i) \neq i \right \}  \right| \\
    &= \sum_{c \in \mathscr{K}^\sigma_1}  \left| \left \{ i \in \mathcal{N}_c: {\sigma}(i) \neq i \right \}  \right|\\ &<  \sum_{c \in \mathcal{C}}  \left| \left \{ i \in \mathcal{N}_c: {\sigma}(i) \neq i \right \}  \right|\\
   &=  \sum_{i,j \in \mathcal{N}: i \neq j} {x}_{i,j}. 
\end{align*}
Indeed, the first equality follows from our construction of $\bar{\sigma}$. The second equality follows from the definition of $\sigma_1$. The strict  inequality follows from the fact that $K^\sigma \ge 2$, which implies that there exists a candidate $i \notin \cup_{c \in \mathscr{K}^\sigma_1} \mathcal{N}_c$ that satisfies $\sigma(i) \neq i$. The third equality follows from the definition of $\sigma$. 

In conclusion, we have shown that there exists a feasible solution $\bar{x} \in \{0,1\}^{\mathcal{N} \times \mathcal{N}}$ for the mixed-integer linear optimization problem~\eqref{prob:oracle_mip_modified} with an objective value that is strictly better than the objective value associated with ${x} \in \{0,1\}^{\mathcal{N} \times \mathcal{N}}$. We thus have a contradiction with the supposition that $x$ is an optimal solution of the mixed-integer linear optimization problem~\eqref{prob:oracle_mip_modified}, which concludes our proof of Theorem~\ref{thm:minimal_reform}. 
\end{proof}

\subsection{Proofs from \S\ref{appx:improvements:noncompetitive}}

\begin{proof}[Proof of Proposition~\ref{prop:noncompetitive}.]
Consider any original ballot style $({\mathcal{N}}, {\mathcal{C}}, \{ {\mathcal{N}}_c \}_{c \in {\mathcal{C}}}, \{{v}_c\}_{c \in {\mathcal{C}}})$, and let the ballot style in which all of the noncompetitive contests from the original ballot style are combined into a single contest be denoted by  $({\mathcal{N}}, \widetilde{\mathcal{C}}, \{ \widetilde{\mathcal{N}}_c \}_{c \in \widetilde{\mathcal{C}}}, \{\widetilde{v}_c\}_{c \in \widetilde{\mathcal{C}}})$. Moreover, consider any $B \in \N$,    $\beta_1,\ldots,\beta_B  \subseteq \mathcal{N}$, $\sigma \in \Sigma \cup \{*\}$, and $i \in \mathcal{N}$. Finally,  let $c \in \mathcal{C}$ denote the contest from the original ballot style that satisfies $i \in \mathcal{N}_c$, and let $\widetilde{c} \in \widetilde{\mathcal{C}}$ denote the contest from the new ballot style that satisfies $i \in \widetilde{\mathcal{N}}_{\widetilde{c}}$. We observe that 
\begin{align*}
&\widetilde{T}^\sigma_i(\beta_1,\ldots,\beta_B) \\
&= \sum_{b=1}^B  \mathbb{I}\left \{   \sigma(i) \in \beta_b  \textnormal{ and } \left| \left \{ \sigma(j) \in \beta_b: j \in \widetilde{\mathcal{N}}_{\widetilde{c}} \right \}  \right| \leq \widetilde{v}_{\widetilde{c}} \right \}\\
&= \begin{cases}
 \sum_{b=1}^B  \mathbb{I}\left \{   \sigma(i) \in \beta_b  \textnormal{ and } \left| \left \{ \sigma(j) \in \beta_b: j \in \mathcal{N}_c \right \}  \right| \leq v_c \right \},&\text{if } \widetilde{c} \neq 0,  \\
  \sum_{b=1}^B  \mathbb{I}\left \{   \sigma(i) \in \beta_b  \textnormal{ and } \left| \left \{ \sigma(j) \in \beta_b: j \in \widetilde{\mathcal{N}}_0 \right \}  \right| \leq \widetilde{v}_0 \right \},&\text{if } \widetilde{c} = 0
 \end{cases}\\
 &= \begin{cases}
 \sum_{b=1}^B  \mathbb{I}\left \{   \sigma(i) \in \beta_b  \textnormal{ and } \left| \left \{ \sigma(j) \in \beta_b: j \in \mathcal{N}_c \right \}  \right| \leq v_c \right \},&\text{if } \widetilde{c} \neq 0,  \\
  \sum_{b=1}^B  \mathbb{I}\left \{   \sigma(i) \in \beta_b  \right \},&\text{if } \widetilde{c} = 0
 \end{cases}\\
 &= \begin{cases}
 \sum_{b=1}^B  \mathbb{I}\left \{   \sigma(i) \in \beta_b  \textnormal{ and } \left| \left \{ \sigma(j) \in \beta_b: j \in \mathcal{N}_c \right \}  \right| \leq v_c \right \},&\text{if } \widetilde{c} \neq 0,  \\
  \sum_{b=1}^B  \mathbb{I}\left \{   \sigma(i) \in \beta_b  \textnormal{ and } \left| \left \{ \sigma(j) \in \beta_b: j \in {\mathcal{N}}_c \right \}  \right| \leq {v}_c \right \},&\text{if } \widetilde{c} = 0
 \end{cases}\\
 &= 
T^\sigma_i(\beta_1,\ldots,\beta_B). 
\end{align*}
The first equality is the definition of $\widetilde{T}^\sigma_i(\cdot)$.
The second equality follows from the fact that if $\widetilde{c} \neq 0$, then it follows from the construction of the new ballot style that $c = \widetilde{c}$, $\widetilde{\mathcal{N}}_{\widetilde{c}} = \mathcal{N}_c$, and $\widetilde{v}_{\widetilde{c}} = v_c$.  The third equality follows from the fact that $\widetilde{v}_0 = | \widetilde{\mathcal{N}}_0|$. The fourth equality follows from the facts that $i \in \mathcal{N}_c$ and  $v_c = | \mathcal{N}_c|$. The fifth equality follows from the definition of $T^\sigma_i(\cdot)$. Our proof of Proposition~\ref{prop:noncompetitive} is thus complete.  
\end{proof}

\subsection{Proofs from Appendix~\ref{appx:upperbound}}
\begin{proof}[Proof of Proposition~\ref{prop:heuristic1}.]
    Consider a test deck defined by the following equalities:
    \begin{align*}
        \beta_1 &= \{1 \},\\
        \beta_2,\beta_3 &= \{2 \},\\
        \beta_4,\beta_5,\beta_6 &= \{3\},\\
        &\vdots\\
        \beta_{\frac{N(N-1)}{2} + 1},\ldots,\beta_{\frac{N(N+1)}{2}} &= \{N \}. 
    \end{align*}We observe that the above test deck consists of $B = N(N+1) / 2$ filled-out ballots.  Moreover, it follows from the fact that $v_c \ge 1$ for all contests $c \in \mathcal{C}$ that the above filled-out ballots satisfy $\beta_1,\ldots,\beta_B \in \mathscr{B}$. 
    
    It remains for us to show that the test deck defined above satisfies the constraints of the optimization problem~\eqref{prob:robust}.  Indeed, we observe that the test deck satisfies the equality $| \{ b \in \{1,\ldots,B\}: i \in \beta_b \}| = i$ for each candidate $i \in \mathcal{N} \equiv \{1,\ldots,N\}$. Furthermore, we recall for each non-identity bijection $\sigma \in \Sigma$ that there must exist a candidate $i \in \mathcal{N}$ that satisfies $\sigma(i) \neq i$. Therefore, we conclude for each non-identity bijection $\sigma \in \Sigma$ that there  exists a  candidate $i \in \mathcal{N}$ that satisfies $| \{ b \in \{1,\ldots,B\}: i \in \beta_b \}|  \neq  |\{ b \in \{1,\ldots,B\}: \sigma(i) \in \beta_b \}|$, which together with Corollary~\ref{cor:diff_votes} implies that the test deck satisfies the constraints of the optimization problem~\eqref{prob:robust}. Our proof of Proposition~\ref{prop:heuristic1} is thus complete. 
\end{proof}

\begin{proof}[Proof of Proposition~\ref{prop:heuristic2}.]
We begin by showing that the optimal objective value of the optimization problem~\eqref{prob:distinct} is greater than or equal to the optimal objective value of the mixed-integer linear optimization problem~\eqref{prob:distinct_mip}. Indeed, let   $(B,\beta_1,\ldots,\beta_B)$ denote an optimal solution for the optimization problem~\eqref{prob:distinct}. This optimal solution assigns each candidate in $\mathcal{N}$ some distinct number of votes. We assume without loss of generality that the solution assigns each candidate some distinct number of votes between $1$ and $|\mathcal{N}|$; if this property does not hold for a given solution, one can simply omit votes for the candidates receiving more than $|\mathcal{N}|$ votes to achieve this property without requiring any additional ballots.

From this optimal test deck, we construct a binary vector $\gamma \in \{0,1\}^{\mathcal{C} \times \mathcal{N}}$ that is defined for each $c \in \mathcal{C}$ and $g \in \mathcal{N}$ as
\begin{align*}
    \gamma_{c,g} &\triangleq \mathbb{I} \left \{ \text{there exists } i \in \mathcal{N}_c \text{ such that } \left|\left \{ b \in \{1,\ldots,B \}: i \in \beta_b \right \}  \right| = g \right \}.  
\end{align*}
In the following bullet points, we show that the integer $B \in \N$ and the binary vector $\gamma \in \{0,1\}^{\mathcal{C} \times \mathcal{N}}$ satisfy each of the constraints of the mixed-integer linear optimization problem~\eqref{prob:distinct_mip}:
\begin{itemize}
    \item { We first show that  $B,\gamma$ satisfies constraint~\eqref{prob:distinct_mip:N_c}. Indeed, we observe for each contest $c \in \mathcal{C}$ that 
\begin{align*}
    \sum_{g \in \mathcal{N}} \gamma_{c,g} &=  \sum_{g \in \mathcal{N}} \mathbb{I} \left \{ \text{there exists } i \in \mathcal{N}_c \text{ such that } \left|\left \{ b \in \{1,\ldots,B \}: i \in \beta_b \right \}  \right| = g \right \}\\
    &= \sum_{i \in \mathcal{N}_c} \sum_{g \in \mathcal{N}} \mathbb{I} \left \{ \left|\left \{ b \in \{1,\ldots,B \}: i \in \beta_b \right \}  \right| = g \right \}\\
     &= \sum_{i \in \mathcal{N}_c} 1\\
     &= | \mathcal{N}_c|.
\end{align*}
The first equality follows from the definition of $\gamma_{c,g}$. The second equality follows from the  fact that $\beta_1,\ldots,\beta_B$ satisfies constraint~\eqref{prob:distinct:distinct}.   The third equality follows from the  fact that $\beta_1,\ldots,\beta_B$   
 satisfies constraint~\eqref{prob:distinct:positive} and from our earlier assumption that the inclusion $\left|\left \{ b \in \{1,\ldots,B \}: i \in \beta_b \right \}  \right| \in \mathcal{N}$ holds for all candidates $i \in \mathcal{N}$.  The fourth equality follows from algebra.  }
 
    \item { We next show that  $B,\gamma$ satisfies constraint~\eqref{prob:distinct_mip:sum_to_one}. Indeed, we observe for each  $g \in \mathcal{N}$ that 
\begin{align*}
    \sum_{c \in \mathcal{C}} \gamma_{c,g} &=      \sum_{c \in \mathcal{C}}  \mathbb{I} \left \{ \text{there exists } i \in \mathcal{N}_c \text{ such that } \left|\left \{ b \in \{1,\ldots,B \}: i \in \beta_b \right \}  \right| = g \right \}\\
    &= \sum_{c \in \mathcal{C}}  \sum_{i \in \mathcal{N}_c} \mathbb{I} \left \{ \left|\left \{ b \in \{1,\ldots,B \}: i \in \beta_b \right \}  \right| = g \right \}\\
     &=  \sum_{i \in \mathcal{N}} \mathbb{I} \left \{ \left|\left \{ b \in \{1,\ldots,B \}: i \in \beta_b \right \}  \right| = g \right \}\\
     &= 1. 
\end{align*}
The first equality follows from the definition of $\gamma_{c,g}$. The second equality follows from the  fact that $\beta_1,\ldots,\beta_B$ satisfies constraint~\eqref{prob:distinct:distinct}.  The third equality follows from algebra. The fourth equality follows from our earlier assumption that $\left|\left \{ b \in \{1,\ldots,B \}: i \in \beta_b \right \}  \right| \in \mathcal{N}$ for all candidates $i \in \mathcal{N}$, which together with the fact that $\beta_1,\ldots,\beta_B$ satisfies constraint~\eqref{prob:distinct:positive}   implies that there must exist exactly one candidate $i \in \mathcal{N}$ that satisfies the equality $\left|\left \{ b \in \{1,\ldots,B \}: i \in \beta_b \right \}  \right| = g$.   }
 \item { We next show that $B, \gamma$ satisfies constraint~\eqref{prob:distinct_mip:bound_B}. Indeed, we observe for each contest $c \in \mathcal{C}$ that 
\begin{align*}
    &\frac{1}{v_c} \sum_{g \in \mathcal{N}} g \gamma_{c,g} \\
    &=  \frac{1}{v_c} \sum_{g \in \mathcal{N}} g \mathbb{I} \left \{ \text{there exists } i \in \mathcal{N}_c \text{ such that } \left|\left \{ b \in \{1,\ldots,B \}: i \in \beta_b \right \}  \right| = g \right \}\\
    &=   \frac{1}{v_c} \sum_{i \in \mathcal{N}_c}  \sum_{g \in \mathcal{N}} g \mathbb{I} \left \{ \left|\left \{ b \in \{1,\ldots,B \}: i \in \beta_b \right \}  \right| = g \right \}\\
        &=   \frac{1}{v_c} \sum_{i \in \mathcal{N}_c} \left|\left \{ b \in \{1,\ldots,B \}: i \in \beta_b \right \}  \right| \\
         &= \frac{1}{v_c}  \sum_{b=1}^B  \left| \mathcal{N}_c \cap \beta_b \right|\\
         &\le \frac{1}{v_c}  \sum_{b=1}^B v_c\\
         &= B.
         \end{align*}
The first equality follows from the definition of $\gamma$. The second equality follows from the  fact that $\beta_1,\ldots,\beta_B$ satisfies constraint~\eqref{prob:distinct:distinct}.   The third equality follows from the  fact that $\beta_1,\ldots,\beta_B$   
 satisfies constraint~\eqref{prob:distinct:positive} and from our earlier assumption that the inclusion $\left|\left \{ b \in \{1,\ldots,B \}: i \in \beta_b \right \}  \right| \in \mathcal{N}$ holds for all candidates $i \in \mathcal{N}$.  The fourth equality follows from algebra. The inequality follows from the fact that $\beta_1,\ldots,\beta_B \in \mathscr{B}$. The fifth equality follows from algebra. 
}
\item Finally, we show that $B, \gamma$ satisfies constraint~\eqref{prob:distinct_mip:N}. Indeed, it follows from the  fact that $\beta_1,\ldots,\beta_B$   
 satisfies constraint~\eqref{prob:distinct:positive} and from our assumption that $\left|\left \{ b \in \{1,\ldots,B \}: i \in \beta_b \right \}  \right| \in \mathcal{N}$ for all candidates $i \in \mathcal{N}$ that there must exist a candidate $i  \in \mathcal{N}$ that satisfies the equality $\left|\left \{ b \in \{1,\ldots,B \}: i \in \beta_b \right \}  \right| = N$. Therefore, we conclude that the inequality $B \ge N$ must be satisfied. 
\end{itemize}
In summary, we have shown in  the above bullet points  that the integer $B \in \N$ and the binary vector $\gamma \in \{0,1\}^{\mathcal{C} \times \mathcal{N}}$ is  a feasible but possibly sub-optimal solution  for the mixed-integer linear optimization problem~\eqref{prob:distinct_mip}.  Because  $(B,\beta_1,\ldots,\beta_B)$ is an optimal solution for the optimization problem~\eqref{prob:distinct}, our proof that the optimal objective value of the optimization problem~\eqref{prob:distinct} is greater than or equal to the optimal objective value of the  mixed-integer linear optimization problem~\eqref{prob:distinct_mip} is thus complete. 

It remains for us to show that the optimal objective value of the  mixed-integer linear optimization problem~\eqref{prob:distinct_mip}  is greater than or equal to the optimal objective value of the optimization problem~\eqref{prob:distinct}. To show this,  let $B \in \N$ and $\gamma \in \{0,1\}^{\mathcal{C} \times \mathcal{N}}$ denote any optimal solution of the mixed-integer linear optimization problem~\eqref{prob:distinct_mip}. Moreover, let $\pi: \mathcal{N} \to \mathcal{N}$ denote the function that satisfies the equality  $\pi(i) = \sum_{g \in \mathcal{N}} g \gamma_{c,g}$ for all contests $c \in \mathcal{C}$ and candidates $i \in \mathcal{N}_c$. It follows from the fact that $B,\gamma$ is a feasible solution for the mixed-integer linear optimization problem~\eqref{prob:distinct_mip} that the function $\pi$ is a bijection. Given the bijection $\pi$, we now construct a test deck $(\beta_1,\ldots,\beta_B)$ using the following procedure:
  \begin{center}
  \fbox{
  \begin{minipage}{0.45\linewidth}
    \begin{algorithmic}
        \State $\beta_1,\ldots,\beta_B \gets \emptyset$
        \For{$c \in \mathcal{C}$}
            \State $b \gets 1$
            \For{$i \in \mathcal{N}_c$}
                \For{$\ell \in \{1,\ldots,\pi(i) \}$}
                    \State $\beta_{b} \gets \beta_{b} \cup \{i \}$
                    \State $b \gets \left(b \mod{B} \right)  + 1$
                \EndFor
            \EndFor
        \EndFor
    \end{algorithmic}
    \end{minipage}}
\end{center}
    The procedure begins by initializing $B$ blank ballots. Then, for each contest $c \in\mathcal{C}$, the procedure  iterates through the ballots and adds the candidates to the ballots. It follows from the fact that $B \ge N$ and from the fact that $\pi(j) \in \mathcal{N}$ for all candidates $j \in \mathcal{N}$ that each candidate $i \in \mathcal{N}_c$ will be selected by this procedure by $\pi(i)$ different ballots. Moreover, it follows from the fact that $B \ge \frac{1}{v_c} \sum_{i \in \mathcal{N}_c} \pi(i)$ that   the procedure will select no more than $v_c$ of the targets from $\mathcal{N}_c$ in any ballot. Therefore, we conclude that the procedure will output a test deck that satisfies $\beta_1,\ldots,\beta_B \in \mathscr{B}$ as well as satisfies all of the constraints of the optimization problem~\eqref{prob:distinct}.  Because we have shown that any optimal solution for the mixed-integer linear optimization problem~\eqref{prob:distinct_mip} can be transformed into a feasible solution for the optimization problem~\eqref{prob:distinct} with the same objective value, we conclude that the optimal objective value of the mixed-integer linear optimization problem~\eqref{prob:distinct_mip} must be greater than or equal to the optimal objective value of the optimization problem~\eqref{prob:distinct}. Our proof of Proposition~\ref{prop:heuristic2} is thus complete.  
\end{proof}

\subsection{Proofs from Appendix~\ref{appx:malicious-decks}}

\begin{proof}[Proof of Proposition~\ref{prop:adversary}]
    Construct a deck of ballots $\beta_1, \ldots, \beta_N$ such that for each $i \in \mathcal{N}$, $\beta_i = \{i\}$. It holds that $T^*_i(\beta_1, \ldots, \beta_N) = 1$ for each $i \in \mathcal{N}$, since only ballot $\beta_i$ is interpreted as containing a vote for candidate $i$. It also  holds for any $\sigma \in \Sigma$ and candidate $i \in \mathcal{N}$ that $T^\sigma_i(\beta_1, \ldots, \beta_N) = 1$, since only ballot $\beta_{\sigma(i)}$ is interpreted as containing a vote for candidate $i$. Thus, we conclude that $T^\sigma(\beta_1,\ldots,\beta_B) = T^*(\beta_1,\ldots,\beta_B)$ and $T^*_i(\beta_1, \ldots, \beta_N) = 1 \geq 1$ for all $i \in \mathcal{N}$.
\end{proof}

\begin{proof}[Proof of Theorem~\ref{thm:swap_adversary}]
    Recall that each $\sigma \in \Sigma$ can be interpreted as a permutation on $\mathcal{N}$, which implies that it can be decomposed into a number of cycles with disjoint sets of elements.\footnote{We say that $i,j \in \mathcal{N}$ are in the same cycle if and only if there exists an integer $k \in \N$ that satisfies $\sigma^k(i) = j$. } Let the set of elements in each of the $K$ cycles be denoted $\mathcal{O}_1, \ldots, \mathcal{O}_K$. We now construct a test deck $\beta_1, \ldots, \beta_B$ such that $B \triangleq K(K+1)/2$ and 
    \begin{align*}
        \beta_1 &\triangleq \mathcal{O}_1,\\
        \beta_2,\beta_3 &\triangleq \mathcal{O}_2\\
     \beta_4,\beta_5,\beta_6 &\triangleq \mathcal{O}_3\\
        &\vdots \\
        \beta_{\frac{K(K-1)}{2}+1},\ldots,  \beta_{\frac{K(K+1)}{2}} &\triangleq \mathcal{O}_K.
    \end{align*}
    
    Because the bidirectional implication $\sigma^n(i) \in \mathcal{N}_c \iff \sigma^n(i) = i$ holds for all contests $c \in \mathcal{C}$, candidates $i \in \mathcal{N}_c$, and integers $n \in \N$, we know that each set $\mathcal{O}_k$ includes at most one candidate from each contest. This means at most one candidate from each contest is marked on each ballot, thereby implying that $\beta_1, \ldots, \beta_B \in \mathscr{B}$. Moreover, for each candidate $i \in \mathcal{O}_k$, we observe that 
    \begin{align*}
        T^*_i(\beta_1, \ldots, \beta_B) & = \sum_{b=1}^B  \mathbb{I}\left \{   i \in \beta_b  \textnormal{ and } \left| \beta_b \cap \mathcal{N}_c \right| \leq v_c \right \} \\
        & = \sum_{b=1}^B  \mathbb{I}\left \{   i \in \beta_b \right \} \\
        & = k.
    \end{align*}
    The first equality is the definition of $T^*_i(\cdot)$. The second equality follows because $\beta_1, \ldots, \beta_B \in \mathscr{B}$. The third equality follows from the fact that the test deck has been constructed to contain $k$ ballots that vote for the candidates in $\mathcal{O}_k$.
    
    We conclude by showing that the voting machine with mapping $\sigma$ gives the correct output for each candidate $i \in \mathcal{N}_c$ in each contest $c \in \mathcal{C}$:  
    \begin{align*}
        T^\sigma_i(\beta_1, \ldots, \beta_B) & = \sum_{b=1}^B  \mathbb{I}\left \{   \sigma(i) \in {\beta}_b  \textnormal{ and } \left| \left \{ \sigma(j) \in {\beta}_b: j \in \mathcal{N}_c \right \}  \right| \leq v_c \right \} \\
        & = \sum_{b=1}^B  \mathbb{I}\left \{   i \in {\beta}_b  \textnormal{ and } \left| \left \{ j \in {\beta}_b: j \in \mathcal{N}_c \right \}  \right| \leq v_c \right \} \\
        & = \sum_{b=1}^B  \mathbb{I}\left \{   i \in \beta_b  \textnormal{ and } \left| \beta_b \cap \mathcal{N}_c \right| \leq v_c \right \} \\
        & = T^*_i(\beta_1, \ldots, \beta_B).
    \end{align*}
    The first equality is the definition of $T^\sigma(\cdot)$. The second equality holds because each ballot $\beta_b$ marks every candidate that falls in the same cycle under $\sigma$, which implies that  $\sigma(j) \in \beta_b \iff j \in \beta_b$ for all $j \in \mathcal{N}_c$. The third equality follows from algebra. The fourth equality is the definition of $T^*(\cdot)$. Our proof of Theorem~\ref{thm:swap_adversary} is thus complete. 
    \end{proof}

\subsection{Proofs from Appendix~\ref{appx:statelaws}}

\begin{proof}[Proof of Proposition~\ref{prop:overvote_1}.]

    Let $\bar{\beta} \triangleq  \mathcal{N}$ be the ballot that votes for every target, and let $(B,\beta_1,\ldots,\beta_B)$ denote a feasible solution for the optimization problem~\eqref{prob:robust}. For each candidate $i \in \mathcal{N}$, we know that the reported vote total under any incorrect mapping $\sigma \in \Sigma$ is as follows, where $c \in \mathcal{C}$ is the contest containing candidate $i$:
    \begin{align*}
          T^\sigma_i(\beta_1,\ldots,\beta_B,\bar{\beta}) &= \left( \sum_{b=1}^B  \mathbb{I}\left \{   \sigma(i) \in \beta_b  \textnormal{ and } \left| \left \{ \sigma(j) \in \beta_b: j \in \mathcal{N}_c \right \}  \right| \leq v_c \right \} \right) \\
        &\quad + \mathbb{I}\left \{   \sigma(i) \in \bar{\beta} \textnormal{ and } \left| \left \{ \sigma(j) \in \bar{\beta}: j \in \mathcal{N}_c \right \}  \right| \leq v_c \right \}\\
        &= T^\sigma_i(\beta_1,\ldots,\beta_B) + \mathbb{I}\left \{   \sigma(i) \in \bar{\beta} \textnormal{ and } \left| \left \{ \sigma(j) \in \bar{\beta}: j \in \mathcal{N}_c \right \}  \right| \leq v_c \right \}\\
        &= T_i^\sigma(\beta_1,\ldots,\beta_B) + \mathbb{I} \left \{ |\mathcal{N}_c| \le v_c \right \}. 
    \end{align*}
    The first two equalities follow from the definition of $T_i^\sigma(\cdot)$, and the third equality follows from the fact that $\bar{\beta} = \mathcal{N}$. The reported vote total on a properly functioning voting machine, meanwhile, is given by the following:
    \begin{align*}
        T^*_i(\beta_1,\ldots,\beta_B,\bar{\beta}) &= \left( \sum_{b=1}^B  \mathbb{I}\left \{   i \in \beta_b  \textnormal{ and } \left| \beta_b \cap \mathcal{N}_c \right| \leq v_c \right \} \right) \\
        &\quad + \mathbb{I}\left \{   i \in \bar{\beta} \textnormal{ and } \left| \bar{\beta} \cap \mathcal{N}_c  \right| \leq v_c \right \}\\
        &= T^*_i(\beta_1,\ldots,\beta_B) + \mathbb{I}\left \{   i \in \bar{\beta} \textnormal{ and } \left| \bar{\beta} \cap \mathcal{N}_c \right| \leq v_c \right \}\\
        &= T_i^*(\beta_1,\ldots,\beta_B) + \mathbb{I} \left \{ |\mathcal{N}_c| \le v_c \right \}.
    \end{align*}
    The first two equalities follow from the definition of $T_i^*(\cdot)$, and the third equality again follows from the fact that $\bar{\beta} = \mathcal{N}$.

    We observe that because $(B,\beta_1,\ldots,\beta_B)$ is a feasible solution for the optimization problem~\eqref{prob:robust}, it must be the case that $T^\sigma(\beta_1, \ldots, \beta_B) \neq T^*(\beta_1, \ldots, \beta_B)$ for all $\sigma \in \Sigma$. This means that for each such $\sigma$ the resulting vectors must differ in at least one position; that is, there must exist some $i \in \mathcal{N}$ such that $T_i^\sigma(\beta_1, \ldots, \beta_B) \neq T_i^*(\beta_1, \ldots, \beta_B)$. For this $i$, we can conclude the following:
    \begin{align*}
        T^\sigma_i(\beta_1,\ldots,\beta_B,\bar{\beta}) &= T_i^\sigma(\beta_1,\ldots,\beta_B) + \mathbb{I} \left \{ |\mathcal{N}_c| \le v_c \right \} \\
        &  \neq T_i^*(\beta_1,\ldots,\beta_B) + \mathbb{I} \left \{ |\mathcal{N}_c| \le v_c \right \} \\
        & = T^*_i(\beta_1,\ldots,\beta_B,\bar{\beta}).
    \end{align*}
    The two equalities follow from the chain of equalities derived above, and the non-equality follows from the fact that $T_i^\sigma(\beta_1, \ldots, \beta_B) \neq T_i^*(\beta_1, \ldots, \beta_B)$ for the given $i \in \mathcal{N}$. We have therefore shown that there exists some $i \in \mathcal{N}$ for each $\sigma \in \Sigma$ such that $T^\sigma_i(\beta_1,\ldots,\beta_B,\bar{\beta}) \neq T^*_i(\beta_1,\ldots,\beta_B,\bar{\beta})$. This means the vectors $T^\sigma(\beta_1,\ldots,\beta_B,\bar{\beta})$ and $T^*(\beta_1,\ldots,\beta_B,\bar{\beta})$ differ in at least one position, so Proposition~\ref{prop:overvote_1} is proven.
\end{proof}

\begin{proof}[Proof of Proposition~\ref{prop:overvote_2}.]
    Let $(B, \beta_1, \ldots, \beta_B)$ denote a feasible solution to the optimization problem~\eqref{prob:robust}, let $\sigma \in \Sigma$, and let $\tilde{\beta} \subseteq \mathcal{N}$ denote any filled-out ballot that satisfies the following equality for each contest $c \in \mathcal{C}$:
    \begin{align*}
        |\tilde{\beta} \cap \mathcal{N}_c| &= \begin{cases}
            v_c + 1, &\textnormal{if } | \mathcal{N}_c| > v_c,\\
            0,&\textnormal{otherwise}. 
        \end{cases}
    \end{align*}
    If there exists $b \in \mathcal{B}$ and $c \in \mathcal{C}$ such that $|\{\sigma(j) \in \beta_b : j \in \mathcal{N}_c\}| > v_c$,  then our proof is complete. Therefore, we assume for the rest of the proof of Proposition~\ref{prop:overvote_2} that the inequality $|\{\sigma(j) \in \beta_b : j \in \mathcal{N}_c\}| \leq v_c$  holds for all $b \in \mathcal{B}$ and $c \in \mathcal{C}$. This allows us to determine that the following holds:
    \begin{align}
        \sum_{c \in \mathcal{C}}\sum_{i \in \mathcal{N}_c} T^*_i(\beta_1, \ldots, \beta_B) 
        & = \sum_{c \in \mathcal{C}} \sum_{i \in \mathcal{N}_c} \sum_{b=1}^B \mathbb{I} \{ i \in \beta_b   \textnormal{ and } | \beta_b \cap \mathcal{N}_c | \leq v_c\} \notag \\ 
        & = \sum_{c \in \mathcal{C}} \sum_{i \in \mathcal{N}_c} \sum_{b=1}^B \mathbb{I} \{ i \in \beta_b \} \notag \\
        & = \sum_{c \in \mathcal{C}} \sum_{i \in \mathcal{N}_c} \sum_{b=1}^B \mathbb{I} \{ \sigma(i) \in \beta_b \} \notag \\
        & = \sum_{c \in \mathcal{C}} \sum_{i \in \mathcal{N}_c} \sum_{b=1}^B  \mathbb{I}\left \{ \sigma(i) \in \beta_b  \textnormal{ and } \left| \left \{ \sigma(j) \in \beta_b: j \in \mathcal{N}_c \right \}  \right| \leq v_c \right \} \notag \\
        & = \sum_{c \in \mathcal{C}}\sum_{i \in \mathcal{N}_c} T^\sigma_i(\beta_1, \ldots, \beta_B). \label{line:bradenname}
    \end{align}
    The first equality is the definition of $T^*_i(\cdot)$. The second equality holds because $\beta_1, \ldots, \beta_B \in \mathscr{B}$. The third equality holds because $\sigma$ is a bijection over $\mathcal{N}$, so the transformation only permutes the order in which terms are added to the sum. The fourth equality holds due to our assumption that $|\{\sigma(j) \in \beta_b : j \in \mathcal{N}_c\}| \leq v_c$. The fifth equality is the definition of $T^\sigma_i(\cdot)$.
    
    It follows from the fact that $(B,\beta_1,\ldots,\beta_B)$ is a feasible solution to the optimization problem~\eqref{prob:robust} and from the fact that $\sigma \in \Sigma$ that $T^*(\beta_1, \ldots, \beta_B) \neq T^\sigma(\beta_1, \ldots, \beta_B)$.  It follows from this fact and from the equality derived on line~\eqref{line:bradenname} that there must exist a candidate $i \in \mathcal{N}_c$ in some contest $c \in \mathcal{C}$ that satisfies the strict inequality $T^*_i(\beta_1, \ldots, \beta_B) < T^\sigma_i(\beta_1, \ldots, \beta_B)$. For this candidate $i$, the following must hold:
    \begin{align*}
        T_i^*(\beta_1, \ldots, \beta_B, \tilde{\beta})
        & = \left( \sum_{b=1}^B  \mathbb{I}\left \{   i \in \beta_b  \textnormal{ and } \left| \beta_b \cap \mathcal{N}_c \right| \leq v_c \right \} \right) \\
            &\quad\quad + \mathbb{I}\left \{   i \in \tilde{\beta} \textnormal{ and } \left| \tilde{\beta} \cap \mathcal{N}_c  \right| \leq v_c \right \} \\
        & = \left( \sum_{b=1}^B  \mathbb{I}\left \{   i \in \beta_b  \textnormal{ and } \left| \beta_b \cap \mathcal{N}_c \right| \leq v_c \right \} \right) \\
        & = T_i^*(\beta_1, \ldots, \beta_B) \\
        & < T_i^\sigma(\beta_1, \ldots, \beta_B) \\
        & = \left( \sum_{b=1}^B  \mathbb{I}\left \{   \begin{aligned} \sigma(i) \in \beta_b  \textnormal{ and } \\ \left| \left \{ \sigma(j) \in \beta_b: j \in \mathcal{N}_c \right \}  \right| \leq v_c \end{aligned} \right \} \right) \\
        & \leq \left( \sum_{b=1}^B  \mathbb{I}\left \{   \begin{aligned} \sigma(i) \in \beta_b  \textnormal{ and } \\ \left| \left \{ \sigma(j) \in \beta_b: j \in \mathcal{N}_c \right \}  \right| \leq v_c \end{aligned} \right \} \right) \\
            & \quad\quad + \mathbb{I}\left \{   \begin{aligned} \sigma(i) \in \beta_b  \textnormal{ and } \\ \left| \left \{ \sigma(j) \in \beta_b: j \in \mathcal{N}_c \right \}  \right| \leq v_c \end{aligned} \right \} \\
        & = T_i^\sigma(\beta_1, \ldots, \beta_B, \tilde{\beta}).
    \end{align*}
    The first equality is the definition of $T^*_i(\cdot)$. The second equality holds because the construction of the filled-out ballot $\tilde{\beta}$ implies that either $|\tilde{\beta} \cap \mathcal{N}_c| = v_c + 1$ or $i \notin \tilde{\beta}$. The third equality is the definition of $T^*_i(\cdot)$. The first inequality follows for candidate $i$ by earlier reasoning. The  fourth equality is the definition of $T^\sigma_i(\cdot)$. The second inequality holds because $\mathbb{I}\{\cdot\}$ is non-negative. The fifth equality is the definition of $T^\sigma_i(\cdot)$.

    In summary, we have shown that if the inequality  $|\{\sigma(j) \in \beta_b : j \in \mathcal{N}_c\}| \leq v_c$  holds for all $b \in \mathcal{B}$ and $c \in \mathcal{C}$, then there must exist a candidate $i \in \mathcal{N}$ that satisfies $T^*_i(\beta_1, \ldots, \beta_B, \tilde{\beta}) \neq T^\sigma_i(\beta_1, \ldots, \beta_B, \tilde{\beta})$. Our proof of Proposition~\ref{prop:overvote_2} is thus complete. 
\end{proof}

\subsection{Proofs from Appendix~\ref{appx:reuse}}

\begin{proof}[Proof of Lemma~\ref{lem:solution_reduction}.]
    Let $(B, \beta_1, \ldots, \beta_B)$ be an optimal solution for the optimization problem~\eqref{prob:robust} for a ballot style parameterized by the tuple $(\mathcal{N}, \mathcal{C}, \{\mathcal{N}_c\}_{c \in \mathcal{C}}, \{v_c\}_{c \in \mathcal{C}})$. Let $\bar{\mathcal{C}} \subset \mathcal{C}$ be a subset of that ballot style's contests which we are removing from the ballot style, and let  $\bar{\mathcal{N}} \triangleq \bigcup_{c \in \bar{\mathcal{C}}} \mathcal{N}_c$ be the candidates in those contests. Define $\mathcal{C}' \triangleq \mathcal{C} \setminus \bar{\mathcal{C}}$ and $\mathcal{N}' \triangleq \mathcal{N} \setminus \bar{\mathcal{N}}$ as the contests and candidates left over when the subsets $\bar{\mathcal{C}}$ and $\bar{\mathcal{N}}$ are removed.
    Consider the ballot style created when the candidates $\bar{\mathcal{N}}$ and contests $\bar{\mathcal{C}}$ are removed, which is parameterized by the tuple $(\mathcal{N}', \mathcal{C}', \{\mathcal{N}_c\}_{c \in \mathcal{C}'}, \{v_c\}_{c \in \mathcal{C}'})$. Define $\Sigma'$ as the set of non-identity bijections over $\mathcal{N}'$; that is, allow it to be the set of possible mappings for this new ballot style.

    Consider some particular $\sigma' \in \Sigma'$, and let the extension of this mapping to the domain $\mathcal{N}$ be defined for each candidate $i \in \mathcal{N}$ as
    \begin{align*}
        \sigma(i) & \triangleq \begin{cases}
            \sigma'(i), & \textnormal{ if } i \in \mathcal{N}', \\
            i, & \textnormal{ if } i \in \bar{\mathcal{N}}.
        \end{cases}
    \end{align*}
We first observe for each candidate $i \in \bar{\mathcal{N}}$ that
    \begin{align*}
        T^\sigma_i(\beta_1, \ldots, \beta_B) 
        & = \sum_{b=1}^B  \mathbb{I}\left \{   \sigma(i) \in {\beta}_b  \textnormal{ and } \left| \left \{ \sigma(j) \in {\beta}_b: j \in \mathcal{N}_c \right \}  \right| \leq v_c \right \} \\
        & = \sum_{b=1}^B  \mathbb{I}\left \{   i \in {\beta}_b  \textnormal{ and } \left| \left \{ j \in {\beta}_b: j \in \mathcal{N}_c \right \}  \right| \leq v_c \right \} \\
        & = \sum_{b=1}^B  \mathbb{I}\left \{   i \in {\beta}_b  \textnormal{ and } \left| \beta_b \cap \mathcal{N}_c \right| \leq v_c \right \} \\
        & = T^*_i(\beta_1, \ldots, \beta_B).
    \end{align*}
    The first equality holds by the definition of $T^\sigma_i(\cdot)$. The second equality holds because $\sigma(i) = i$ for all $i \in \bar{\mathcal{N}}$. The third equality follows from algebra. The fourth equality follows from the definition of $T^\sigma_i(\cdot)$.

    It follows from the fact that $(B, \beta_1, \ldots, \beta_B)$ is a feasible solution for the optimization problem~\eqref{prob:robust} that that $T^\sigma(\beta_1, \ldots, \beta_B) \neq T^*(\beta_1, \ldots, \beta_B)$. With the equality derived above, this means  that there must exist a candidate $i \in \mathcal{N}'$ that satisfies $T_i^\sigma(\beta_1, \ldots, \beta_B) \neq T_i^*(\beta_1, \ldots, \beta_B)$. Take that candidate $i$ and let $c \in \mathcal{C}'$ be the contest that satisfies $i \in \mathcal{N}_c$. Then, it holds that
    \begin{align*}
        T^*_i(\beta_1 \setminus \bar{\mathcal{N}}, \ldots, \beta_B \setminus \bar{\mathcal{N}})
        & = \sum_{b=1}^B  \mathbb{I}\left \{   i \in {\beta}_b \setminus \bar{\mathcal{N}}  \textnormal{ and } \left| (\beta_b \setminus \bar{\mathcal{N}}) \cap \mathcal{N}_c \right| \leq v_c \right \} \\
        & = \sum_{b=1}^B  \mathbb{I}\left \{   i \in {\beta}_b \textnormal{ and } \left| (\beta_b \setminus \bar{\mathcal{N}}) \cap \mathcal{N}_c \right| \leq v_c \right \} \\
        & = \sum_{b=1}^B  \mathbb{I}\left \{   i \in {\beta}_b \textnormal{ and } \left| \beta_b \cap \mathcal{N}_c \right| \leq v_c \right \} \\
        & = T^*_i(\beta_1, \ldots, \beta_B) \\
        & \neq T^\sigma_i(\beta_1, \ldots, \beta_B) \\
        & = \sum_{b=1}^B  \mathbb{I}\left \{   \sigma(i) \in {\beta}_b  \textnormal{ and } \left| \left \{ \sigma(j) \in {\beta}_b: j \in \mathcal{N}_c \right \}  \right| \leq v_c \right \} \\
        & = \sum_{b=1}^B  \mathbb{I}\left \{   \sigma'(i) \in {\beta}_b  \textnormal{ and } \left| \left \{ \sigma'(j) \in {\beta}_b: j \in \mathcal{N}_c \right \}  \right| \leq v_c \right \} \\
        & = \sum_{b=1}^B  \mathbb{I}\left \{   \sigma'(i) \in {\beta}_b \setminus \bar{\mathcal{N}} \textnormal{ and } \left| \left \{ \sigma'(j) \in {\beta}_b \setminus \bar{\mathcal{N}}: j \in \mathcal{N}_c \right \}  \right| \leq v_c \right \} \\
        & = T^{\sigma'}_i(\beta_1 \setminus \bar{\mathcal{N}}, \ldots, \beta_B \setminus \bar{\mathcal{N}}).
    \end{align*}
    The first equality holds by the definition of $T^*_i(\cdot)$. The second equality holds because $i \notin \bar{\mathcal{N}}$. The third equality holds because $\bar{\mathcal{N}} \cap \mathcal{N}_c = \emptyset$. The fourth equality holds by the definition of $T^*_i(\cdot)$. The non-equality follows from our choice of $i$. The fifth equality holds by the definition of $T^\sigma_i(\cdot)$. The sixth equality holds because $\sigma(i) = \sigma'(i)$ for all $i \in \mathcal{N}'$. The seventh equality holds because $\sigma'$ has a range which excludes $\bar{\mathcal{N}}$. The eighth equality is the definition of $T^{\sigma'}_i(\cdot)$.

    In summary, we have shown for each $\sigma' \in \Sigma'$ that there exists a candidate $i \in \mathcal{N}'$ that satisfies $T^*_i(\beta_1 \setminus \bar{\mathcal{N}}, \ldots, \beta_B \setminus \bar{\mathcal{N}}) \neq T^{\sigma'}_i(\beta_1 \setminus \bar{\mathcal{N}}, \ldots, \beta_B \setminus \bar{\mathcal{N}})$. This fact, along with the observation that $\beta_1 \setminus \bar{\mathcal{N}}, \ldots, \beta_B \setminus \bar{\mathcal{N}} \in \mathscr{B}$ since $\beta_1, \ldots, \beta_B \in \mathscr{B}$, allows us to conclude that $(B, \beta_1 \setminus \bar{\mathcal{N}}, \ldots, \beta_B \setminus \bar{\mathcal{N}})$ is a feasible solution to the optimization problem~\eqref{prob:robust} for the ballot style parameterized by the tuple $(\mathcal{N}', \mathcal{C}', \{\mathcal{N}_c\}_{c \in \mathcal{C}'}, \{v_c\}_{c \in \mathcal{C}'})$. 
\end{proof}

\end{document}